\documentclass[aoas]{imsart}

\RequirePackage{amsthm,amsmath,amsfonts,amssymb}
\RequirePackage[authoryear]{natbib}
\RequirePackage[colorlinks,citecolor=blue,urlcolor=blue]{hyperref}
\RequirePackage{graphicx}

\usepackage{mathtools} % for the \vcentcolon command in \leqdef
\usepackage{float} % more options for the placement of figures and tables
\usepackage{booktabs} % for the \toprule command in tables
\usepackage{subcaption} % for the subfigure environment
\usepackage{dsfont} % for indicator functions
\usepackage[table]{xcolor} % for grey cells in tables
\usepackage{dcolumn} % to align decimal points in tables

\numberwithin{equation}{section}
\numberwithin{table}{section}
\numberwithin{figure}{section}

\theoremstyle{plain}
\newtheorem{theorem}{Theorem}[section]
\newtheorem{proposition}[theorem]{Proposition}
\newtheorem{corollary}[theorem]{Corollary}
\newtheorem{lemma}[theorem]{Lemma}

\theoremstyle{remark}
\newtheorem{remark}{Remark}[section]
\newtheorem{definition}[remark]{Definition}

\startlocaldefs
\DeclareMathOperator*{\argmin}{argmin}

\newcommand{\N}{\mathbb{N}}
\newcommand{\R}{\mathbb{R}}
\newcommand{\PP}{\mathsf{P}} % Russian style (do not change)
\newcommand{\EE}{\mathsf{E}} % Russian style (do not change)
\newcommand{\Var}{\mathsf{Var}} % Russian style (do not change)
\newcommand{\Cov}{\mathsf{Cov}} % Russian style (do not change)
\newcommand{\bb}[1]{\boldsymbol{#1}}
\newcommand{\rd}{\mathrm{d}}
\newcommand{\ind}{\mathds{1}}
\newcommand{\e}{\varepsilon}

\newcommand{\leqdef}{\vcentcolon=}
\newcommand{\reqdef}{=\vcentcolon}
\newcommand{\Rnearrow}{\R_{\raisebox{0.2ex}{\scalebox{0.5}{$\nearrow$}}}^{d+1}}
\newcommand{\Rnearrowtwo}{\R_{\raisebox{0.2ex}{\scalebox{0.5}{$\nearrow$}}}^2}

\endlocaldefs

\allowdisplaybreaks

\begin{document}

\begin{frontmatter}
\title{A linear regression model for quantile function data applied to paired pulmonary 3d CT scans}
\runtitle{A linear regression model for quantile function data}

\begin{aug}
\author[a1]{\fnms{Marie-F\'elicia}~\snm{B\'eclin}\ead[label=e1]{marie-felicia.beclin@umontpellier.fr}\orcid{0009-0005-8117-2983}},
\author[a2]{\fnms{Pierre}~\snm{Lafaye De Micheaux}\ead[label=e2]{lafaye@unsw.edu.au}\orcid{0000-0002-0247-5136}}
\author[a3]{\fnms{Nicolas}~\snm{Molinari}\ead[label=e3]{nicolas.molinari@inserm.fr}
\orcid{0000-0002-1786-0088}}
\and
\author[a4]{\fnms{Fr\'ed\'eric}~\snm{Ouimet}\ead[label=e4]{frederic.ouimet2@mcgill.ca}\orcid{0000-0001-7933-5265}}
%%% ONE ADDRESS PER AUTHOR ONLY
\address[a1]{IDESP, INSERM, PreMEdical INRIA, Univ Montpellier, CHU Montpellier, Montpellier, France\printead[presep={,\ }]{e1}}\vspace{-2mm}
\address[a2]{School of Mathematics and Statistics, UNSW Sydney, Sydney, Australia\printead[presep={,\ }]{e2}}\vspace{-2mm}
\address[a3]{IDESP, INSERM, PreMEdical INRIA, Univ Montpellier, CHU Montpellier, Montpellier, France\printead[presep={,\ }]{e3}}\vspace{-2mm}
\address[a4]{Department of Mathematics and Statistics, McGill University, Montr\'eal, Canada\printead[presep={,\ }]{e4}}
\end{aug}

\begin{abstract}
This paper introduces a new objective measure for assessing treatment response in asthmatic patients using computed tomography (CT) imaging data. For each patient, CT scans were obtained before and after one year of monoclonal antibody treatment. Following image segmentation, the Hounsfield unit (HU) values of the voxels were encoded through quantile functions. It is hypothesized that patients with improved conditions after treatment will exhibit better expiration, reflected in higher HU values and an upward shift in the quantile curve. To objectively measure treatment response, a novel linear regression model on quantile functions is developed, drawing inspiration from \citet{Verde2010}. Unlike their framework, the proposed model is parametric and incorporates distributional assumptions on the errors, enabling statistical inference. The model allows for the explicit calculation of regression coefficient estimators and confidence intervals, similar to conventional linear regression. The corresponding data and \texttt{R} code are available on GitHub to facilitate the reproducibility of the analyses presented.
\end{abstract}

\begin{keyword}
\kwd{Air-trapping}
\kwd{asthma}
\kwd{Benralizumab}
\kwd{biomarker}
\kwd{bronchial remodeling}
\kwd{computed tomography}
\kwd{CT scan}
\kwd{histogram}
\kwd{Hounsfield unit}
\kwd{lung}
\kwd{monoclonal antibody}
\kwd{parametric model}
\kwd{quantile function}
\kwd{quantile function regression}
\kwd{quantile regression}
\kwd{segmentation}
\kwd{treatment response}
\end{keyword}

\end{frontmatter}

\section{Introduction}\label{sec:intro}

Asthma is a chronic inflammatory disease that primarily affects small airways, leading to airway remodeling and air-trapping, which impair lung function and often result in poor exhalation. Monoclonal antibody treatments like Benralizumab have shown significant improvements in exacerbation rates and symptoms but are expensive, making it essential to prescribe them to patients who are likely to respond. It is hypothesized that computed tomography (CT) data, particularly air-trapping quantifications, can be used to identify clinical responders to Benralizumab early in the treatment of asthmatic patients. Such early identification could optimize treatment rule-out times and improve clinical decision-making.

The primary clinical challenge of this study is to find a reliable and effective biomarker to measure and predict treatment response in asthma. Medical images, particularly CT-based, have already been used as biomarkers for diagnosis and monitoring disease progression. Quantitative thoracic computed tomography (QTCT) is especially useful in this context. In asthma, imaging modalities such as CT play a crucial role in identifying phenotypes \citep{marin2016fractal,cabon2019k} and quantifying therapeutic responses \citep{pompe_imaging-derived_2023}. Thoracic CT imaging offers the advantage of directly visualizing bronchial remodeling \emph{in vivo}. For example, the clinical trial conducted by Genofre's team \citep{genofre_effects_2023} utilizes markers such as the volume (number of voxels) of lobes, lungs, and airways. Another indicator for classifying severe asthma, used in the analysis of expiration and inspiration CT scans, is the expiration/inspiration (E/I) ratio. This is the ratio of the mean CT-determined values for the bilateral upper and lower pulmonary segments during full expiration to those during full inspiration \citep{mitsunobu_use_2005}. In a study conducted by \cite{hartley_et_al_2016}, a significant difference was observed in the E/I ratio between asthmatic patients, healthy individuals, and those with chronic obstructive pulmonary disease.

Although these traditional indicators are informative, they do not fully exploit the comprehensive data available from CT scans. Hence, there is a pressing need to develop novel imaging-derived measures \citep{trivedi}.

A more thorough analysis can be achieved by examining CT scan histograms \citep{zarei2024quantitative}. These histograms are valuable tools for analyzing medical image data \citep{sumikawa2009computed}, as they provide a quantifiable means of assessing a patient's health status.

However, histograms have limitations, including sensitivity to bin origin and width, which can suppress important details and obscure granularity, multimodality, and skewness. These binning choices can significantly affect the histogram's appearance, making it difficult to discern real patterns from artifacts. Additionally, comparing histograms is often more challenging than assessing fit in a Q-Q plot \citep{Cleveland2004}.

In contrast, quantile functions provide a concise summary of the data distribution without the need for binning. They are more robust to outliers and facilitate easier comparisons across datasets.

More recently, the concurrent emergence of the theory of optimal transport and Wasserstein spaces in the field of statistics \citep{Panaretos2019} has broadened the possibilities for analyzing histogram and density data. Several regression approaches based on Wasserstein distance have emerged; see, e.g., \citet{IrpinoVerde2015,dias2015linear,bonneel_Wasserstein_2016,chen2023wasserstein,ghodrati2022distribution,ghodrati2024transportation,OkanoImaizumi2024,Irpino2024}. However, none of these techniques provide the capability for inference, which is a critical limitation when working in the field of biostatistics.

Consequently, the present study proposes a statistical methodology centered on analyzing quantile functions derived from segmented thoracic CT scans to quantify expiration quality in asthmatic patients, both before and after treatment. To objectively measure treatment response, a novel linear regression model is developed to analyze $n$ pairs of quantile functions, mimicking the structure of classical linear regression but extending its application to cases where both the explanatory and response variables are quantile functions rather than real-valued observations. The model draws inspiration from \citet{Verde2010} (see also \citet{IrpinoVerde2015,dias2015linear}) but distinguishes itself by being fully parametric and incorporating distributional assumptions on the errors, which enhances both the interpretability of the estimators and the ability to perform statistical inference.

Our method should not be mistaken for classical (scalar-on-scalar) quantile regression \citep{Koenker2005}, nor for scalar-on-function quantile regression \citep{Cardot2005,Chen2012,Kato2012,LiWangMaityStaicu2022,Ghosal2023,Yan2023,Beyaztas2024}, function-on-scalar quantile regression \citep{Yang2020a,Yang2020b,Liu2020,Zhang2022,Liu2023}, or function-on-function quantile regression \citep{Beyaztas2022,Zhu2023,Mutis2024,BeyaztasShangSaricam2024}. Although the latter bears certain similarities to the proposed model, the conditional distribution of the functional response are characterized through quantiles, i.e., the predictor and response themselves are not quantile curves. Furthermore, these models are generally semiparametric as they rely on nonparametric methods to estimate the functional regression coefficients. This reliance significantly complicates interpretability and limits the scope for conducting statistical inference. Conversely, our model is fully parametric and allows for the explicit calculation of (scalar) regression coefficient estimators and confidence intervals, akin to conventional linear regression. Although initially motivated by the asthma problem, our method can be used in other applications.

The paper is organized as follows. Section~\ref{sec:medical.problem} details the methodology for medical image processing to construct a quantile function dataset. Section~\ref{sec:regression} introduces the regression model for quantile functions, along with the regression coefficient estimators and confidence intervals. Confidence regions are also provided for residual quantile functions and the mean response of a new observation. The proofs of the results stated in Section~\ref{sec:regression} are relegated to Section~\ref{sec:proofs}. In Section~\ref{sec:application}, the new linear regression model is applied to a quantile function dataset obtained from CT scans of the lungs of $44$ asthmatic patients treated with Benralizumab. This application demonstrates the effectiveness and practicality of the approach in a clinical setting. A discussion of the findings ensues in Section~\ref{sec:discussion}.

\section{Quantile function data}\label{sec:medical.problem}

\subsection{The raw data}

Patients in the study underwent a 48-week treatment regimen with Benralizumab, administered subcutaneously at a dosage of 30 mg per injection. The treatment schedule consisted of an initial phase with injections every four weeks for the first three doses, followed by subsequent injections every eight weeks for the remaining five doses. Baseline (pre-treatment) and end-of-treatment (post-treatment) data were collected as part of the study protocol.

Each CT scan in the dataset comprises approximately 600 slices, each with dimensions of $512 \times 512$, resulting in a three-dimensional array of size $600 \times 512 \times 512$. Each array thus encodes about 157 million voxel values, stored in Digital Imaging and Communication in Medicine (DICOM) files. During preprocessing, voxel values are converted into Hounsfield Unit (HU) values, which quantify the X-ray attenuation characteristics of tissues. The HU scale is defined such that 0 represents the attenuation of pure water at standard temperature and pressure (STP), $-1000$ corresponds to pure air, and higher values (e.g., 1000) are assigned to dense materials like compact bone and metal. Lower HU values indicate higher air content, making them critical for evaluating lung function. For reference, most tissue densities fall between $20$ and $80$ HU, with blood typically ranging from $50$ to $80$ HU.

In practice, CT images are traditionally stored using 12 bits, and a mapping is used to represent radiodensities between $-1023$ HU (assimilated as air) and $3072$ HU. The HU values are processed and reconstructed by a computer into grayscale images, enabling detailed quantification of lung structure and function. Lung segmentation is performed using a threshold-based segmentation algorithm \citep{heuberg_lung_2005}, isolating regions with HU values ranging from $-1023$ to $-200$. This segmentation highlights the air-trapping characteristics of asthmatic lungs, which are key to assessing treatment response. For details on the methodology and associated code, see \cite{Beclin2024}.

After segmentation, the dataset consists of $N_i^x$ and $N_i^y$ voxels for each patient $i$, corresponding to pre- and post-treatment scans, respectively. Typically, $N_i^x$ and $N_i^y$ range between $3$ and $15$ million voxels and are not necessarily equal. These data are compactly stored in two rectangular tables, available on the GitHub repository of \cite{BeclinMicheauxMolinariOuimet_github_2024}. The first column of each table contains $824$ unique HU values (ranging from $-1023$ to $-200$), while the remaining $44$ columns represent the voxel counts for each one of the $n=44$ patients.

\subsection{From raw data to quantile functions}\label{sec:raw.to.quantile}

To quantify treatment response, the HU value distributions for each patient before and after treatment are analyzed through their corresponding quantile functions. By focusing on quantile functions, the analysis discards spatial information about voxel locations, emphasizing changes in the distribution of HU values instead. This approach allows for a more robust assessment of the overall shift in lung function metrics, facilitating the identification of clinical responders to Benralizumab. The raw data are transformed into quantile functions as illustrated in Figure~\ref{fig:raw.to.quantile}.

\begin{figure}[!ht]
\includegraphics[width=1\textwidth]{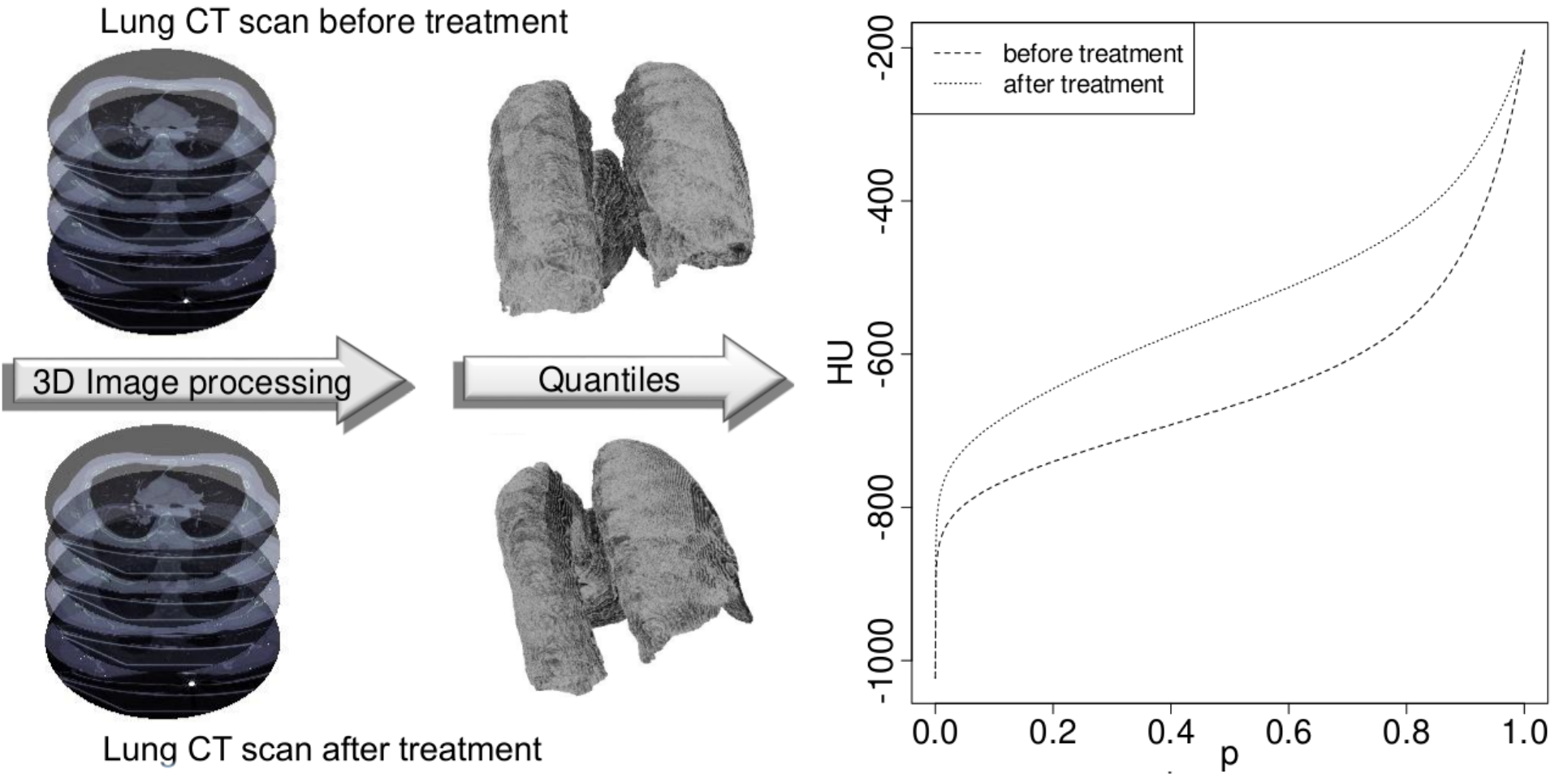}
\caption{Visualization of the process for extracting quantile function data from CT scan images: 3D CT scans, acquired before and after treatment, undergo segmentation to isolate the lungs. Once segmented, the distribution of HU values within the lung voxels is used to construct the corresponding empirical quantile function for each image.}
\label{fig:raw.to.quantile}
\end{figure}

For the $i$th patient ($i \in {1, \ldots, n}$, with $n=44$), let $x_{i,1} \leq \dots \leq x_{i,N_i^x}$ denote the observed pre-treatment HU values arranged in nondecreasing order. Each value belongs to the discrete set $\{-1023,\ldots,-200\}$. The empirical quantile function for these HU values is defined as
\[
\widehat{q}_{x,i}(p)
= \inf\bigg\{t\in \R : p \leq \frac{1}{N_i^x} \sum_{j=1}^{N_i^x} \ind_{\{x_{i,j}\leq t\}}\bigg\}
= \sum_{j=1}^{N_i^x} x_{i,j} \, \ind_{\{\frac{j-1}{N_i^x} < p \leq \frac{j}{N_i^x}\}}, \quad p\in [0,1],
\]
encoding all information about the observed HU values and their frequencies \citep{Parzen1983}.

The empirical quantile function $\widehat{q}_{x,i}$ lies in the space of quantile functions $q_F$, which characterizes a probability distribution $F$ through the formula
\begin{equation}\label{eq:space.quantile}
q_F(p) \leqdef \inf\{x\in \R : p \leq F(x)\}, \quad p\in [0,1].
\end{equation}
This function can be viewed as a generalized inverse of the cumulative distribution function (cdf) $F$. However, the space of quantile functions is too large to derive interpretable linear regression models for pre- and post-treatment quantile functions, especially models that allow explicit computation of regression coefficients and confidence intervals.

To address this limitation, the analysis focuses on a flexible parametric family of quantile functions. Specifically, consider the space of nondecreasing polynomials of degree $d$ with the standard normal quantile function $\Phi^{-1}(p)$ as their argument. The standard normal quantile function is defined as
\[
\Phi^{-1}(p) = \inf\{x \in \R : p \leq \Phi(x)\}, \quad p \in [0,1],
\]
where $\Phi(x) = \smash{\int_{-\infty}^x e^{-z^2/2} / \sqrt{2\pi} \, \rd z}$ is the cdf of the standard normal distribution. In this parametric family, each quantile function $q$ is represented as
\begin{equation}\label{eq:parametric.quantile.space}
q(p) = \sum_{i=0}^d a_i (\Phi^{-1}(p))^i, \quad \bb{a} = (a_0, \ldots, a_d) \in \Rnearrow,
\end{equation}
where the vector of coefficients $\bb{a}$ belongs to the set
\begin{equation}\label{eq:R.inc}
\Rnearrow \coloneqq \left\{\bb{a} \in \R^{d+1} : \frac{\rd}{\rd x} \sum_{i=0}^d a_i x^i \geq 0 \quad \forall x \in \R \right\}.
\end{equation}
The restriction $\bb{a} \in \Rnearrow$ ensures that the polynomial remains nondecreasing, preserving the monotonicity characteristic of quantile functions.

\begin{remark}
The motivation for the above parametric space is that any quantile function $q_F\in \mathcal{L}^2([0,1])$ can be approximated with arbitrary precision by a parametric quantile function of the form \eqref{eq:parametric.quantile.space}, provided the polynomial degree $d$ is sufficiently large. To see this, note that the corresponding nondecreasing map $z \mapsto q_F(\Phi(z))$, which lies in the Wiener space $(\mathcal{L}^2(\mathbb{R}), \|\cdot\|_{\phi})$ where $\phi(z) = e^{-z^2/2} / \sqrt{2\pi}$ and $\|h\|_{\phi} = \int_{-\infty}^{\infty} (h(z))^2 \phi(z) \rd z$, can be approximated arbitrarily closely by nondecreasing polynomials. A sketch of the argument is as follows: Choose $M > 0$ large enough so that $\smash{\int_{|z|>M} \phi(z) \rd z}$ is negligible. On $[-M,M]$, approximate $z \mapsto q_F(\Phi(z))$ by a nondecreasing continuous function, making linear interpolations around any jump discontinuities to ensure minimal Gaussian-integrated error. Next, approximate this continuous nondecreasing function uniformly on $[-M,M]$ by Bernstein polynomials, known to preserve monotonicity \citep{Lorentz1986}. Combining these steps establishes the desired approximation result.
\end{remark}

The quantile function $q_{x,i}$, representing the $i$th patient's pre-treatment HU values, is derived as the $\mathcal{L}^2([0,1])$ projection of $\widehat{q}_{x,i}$ onto the parametric space. This projection, for each $i\in \{1,\ldots,n\}$ and any integer degree $d$, is the solution to the optimization problem below:
\[
\begin{aligned}
\min_{q} \int_0^1\left(q(p)-\widehat{q}_{x,i}(p)\right)^2 \rd p
&= \min_{\bb{a}_{i,d}\in \Rnearrow} \int_0^1 \bigg(\sum_{j=0}^d a_{i,d,j} (\Phi^{-1})^j(p) - \widehat{q}_{x,i}(p)\bigg)^2 \rd p \\[-1mm]
&= \min_{\bb{a}_{i,d}\in \Rnearrow} \|P_{x,i,d} - \widehat{q}_{x,i} \circ \Phi\|_{\phi},
\end{aligned}
\]
where $P_{x,i,d}(z) = \sum_{j=0}^d a_{i,d,j} z^j$. The solution is
\[
\begin{aligned}
q_{x,i,d}^{\star} = \sum_{j=0}^d a_{i,d,j}^{\star} (\Phi^{-1})^j(\cdot), \qquad \bb{a}_{i,d}^{\star} = \argmin_{\bb{a}_{i,d}\in \Rnearrow} \|P_{x,i,d} - \widehat{q}_{x,i} \circ \Phi\|_{\phi}.
\end{aligned}
\]
This optimization procedure can be implemented numerically as detailed in Proposition~\ref{prop:q.opt} of Section~\ref{sec:explicit.minimizer}, where the explicit calculation of $\|P_{x,i,d} - \widehat{q}_{x,i} \circ \Phi\|_{\phi}$ is provided. The problem admits an exact solution using a semidefinite programming approach, as described in \cite{siem2008discrete}.

To select the optimal degree $d$ for each patient, the Bayesian Information Criterion (BIC) is used:
\[
d_{x,i}^{\star} \leqdef \underset{d\in \{1,\ldots,D\}}{\argmin}\left\{ \left\|P_{x,i,d}^{\star} - \widehat{q}_{x,i} \circ \Phi\right\|_{\phi} + d\log(N_i^x)\right\},
\]
where $D$ is a fixed positive integer. The resulting pre-treatment quantile function for the $i$th patient is then
\[
q_{x,i} = \sum_{j=0}^{d_{x,i}^{\star}} a_{i,d_{x,i}^{\star},j}^{\star} (\Phi^{-1})^j(\cdot).
\]

The same procedure is applied to the post-treatment HU values $y_{i,1} \leq \dots \leq y_{i,N_i^y}$ to obtain the post-treatment quantile function $q_{y,i}$. Thus, the dataset consists of $n=44$ pairs of parametric quantile functions, $\{(q_{x,i}, q_{y,i})\}_{i=1}^n$, representing pre- and post-treatment HU distributions for the study population.

In Sections~\ref{sec:regression} and~\ref{sec:application}, the linear regression model will be developed and applied under the constraint $d = D = 1$, simplifying the derivation of explicit estimators, confidence intervals, and the interpretability of the results. The case $D > 1$ is left for future research.

\section{A simple linear regression model for quantile functions}\label{sec:regression}

The proposed simple linear regression model to accommodate random quantile function data is
\begin{equation}\label{eq:not.our.model}
Q_{Y,i} = \beta_0 + \beta_1 q_{x,i} + E_i, \qquad i\in \{1,\ldots,n\},
\end{equation}
where the error terms $E_1,\ldots,E_n$ are assumed to form an independent and identically distributed (iid) sequence of random quantile functions, whose common distribution will be specified later. Unlike classical linear regression, where the response variable is a scalar, the response $Q_{Y,i}$ in this model is a random quantile function. This distinction is emphasized through the use of capital letters. The explanatory quantile function $q_{x,i}$, being observed, is represented with lowercase letters.

While this model was originally motivated by the specific problem of assessing asthma treatment response as described in Sections~\ref{sec:intro}~and~\ref{sec:medical.problem}, it is designed to be more broadly applicable to statistical problems involving quantile function datasets.

For added flexibility, consider the parametric regression model
\begin{equation}\label{eq:our.model}
Q_{Y,i} = \beta_0 + \beta_1\mu_{q_{x,i}} + \beta_2 q_{x,i}^c + E_i, \qquad i\in \{1,\ldots,n\},
\end{equation}
where the centered explanatory quantile functions $q_{x,i}^c$ are defined as
\[
q_{x,i}^c \leqdef q_{x,i}-\mu_{q_{x,i}}, \qquad \mu_{q_{x,i}} \leqdef \int_0^1 q_{x,i}(p) \rd p\in\mathbb{R}.
\]

The mathematical foundation of the model is detailed in the next section, where the concept of random quantile function is formally introduced, and a rigorous definition of the error terms $E_i$ is provided.

\begin{remark}
Compared to \eqref{eq:not.our.model}, model \eqref{eq:our.model} introduces two degrees of freedom in terms of shift and tilt. This design aligns with the functional data regression model by \cite{Verde2010}, which also examines shift and tilt effects. However, the current approach differs in two key aspects. Specifically, a parametric form is assumed for the quantile functions $q_{x,i}$ and $Q_{Y,i}$, enabling the derivation of interpretable regression coefficients, and distributional assumptions are imposed on the errors $E_i$, facilitating statistical inference.
\end{remark}

\subsection{A mathematical intermezzo}\label{sec:math}

For any $p\geq 1$, let $\mathcal{L}^p([0,1])$ be the space (of equivalence classes of almost everywhere equal) Borel measurable functions having a finite $p$-norm, defined by $\|f\|_p = \smash{\int_0^1 |f(x)|^p \rd x}$. This is well known to be a real Banach space under the usual operations of scalar multiplication and addition of functions.

Now, consider the closed metric subspace of $\mathcal{L}^p([0,1])$, denoted by $\mathcal{Q}^p$, that contains the quantile functions as in \eqref{eq:space.quantile}. Since quantile functions map one-to-one with Borel probability measures on $\R$, there is a bijective isometry between $(\mathcal{Q}^p, \|\cdot\|_p)$ and the space of Borel probability measures on $\R$ endowed with the $p$-Wasserstein metric.

As mentioned in Section~\ref{sec:raw.to.quantile}, the metric space $\mathcal{Q}^p$ is too large for modelling purposes. Therefore, the scope is confined to the closed metric subspace encompassing the nondecreasing polynomials of degree $d$ of the standard normal quantile function $\Phi^{-1}$, viz.,
\[
\mathcal{Q}_d^p = \left\{q\in \mathcal{Q}^p : \exists \bb{a}\in \Rnearrow ~\text{such that } q(u) = \sum_{i=0}^d a_i (\Phi^{-1})^i(u) ~~\forall u\in [0,1]\right\}.
\]
For $p=2$ and any $q(\cdot) = \sum_{i=0}^d a_i (\Phi^{-1})^i(\cdot) \in \mathcal{Q}_d^2$, the change of variable $u = \Phi(x)$ yields
\[
\|q\|_2 = \int_0^1 |q(u)|^2 \rd u = \sum_{i,j=0}^d a_i a_j \int_0^1 (\Phi^{-1}(u))^{i+j} \rd u = \sum_{i,j=0}^d a_i a_j \int_{\R} x^{i+j} \phi(x) \rd x.
\]
Since odd Gaussian moments are equal to zero, it follows that
\[
\|q\|_2 = \sum_{i=0}^d \sum_{k= \lceil i/2 \rceil}^{\lfloor (d+i)/2 \rfloor} a_i a_{2k-i} \int_{\R} x^{2k} \phi(x) \rd x = \sum_{i=0}^d \sum_{k= \lceil i/2 \rceil}^{\lfloor (d+i)/2 \rfloor} a_i a_{2k-i} \frac{(2k)!}{2^k k!}.
\]
Using the last expression to define the norm $\|\bb{a}\|_{\raisebox{0.2ex}{\scalebox{0.5}{$\nearrow$}}}$ of any vector $\bb{a}\in \Rnearrow$ as defined in \eqref{eq:R.inc}, it induces a bijective isometry $\Delta: \smash{(\Rnearrow, \|\cdot\|_{\raisebox{0.2ex}{\scalebox{0.5}{$\nearrow$}}})\to (\mathcal{Q}_d^2, \|\cdot\|_2)}$.

Since all norms are equivalent in finite dimension, note that the Borel $\sigma$-algebra on $\smash{\Rnearrow}$ coincides with the restriction of the Borel $\sigma$-algebra on $\R^{d+1}$. Explicitly, one has
\[
\mathcal{B}(\smash{\Rnearrow}) = \{B \in \mathcal{B}(\R^{d+1}) : B \subseteq \smash{\Rnearrow}\}.
\]
Thus, a random element in $\smash{(\smash{\Rnearrow},\mathcal{B}(\smash{\Rnearrow}))}$ is simply a random vector in $\smash{(\R^{d+1},\mathcal{B}(\R^{d+1}))}$ with support restricted to $\smash{\Rnearrow}$.

\begin{definition}\label{def:random.quantile}
A $\smash{\mathcal{Q}_d^2}$-valued random quantile function $Q$, or equivalently a random element in $\mathcal{Q}_d^2$, is a measurable map from any given probability space to $(\mathcal{Q}_d^2,\mathcal{B}(\mathcal{Q}_d^2))$.
\end{definition}

\begin{remark}\label{rem:construction.E.i.general}
A simple way to construct a $\smash{\mathcal{Q}_d^2}$-valued random quantile function $Q$ is to compose a random vector $\bb{A} = (A_0,\ldots,A_d): (\Omega,\mathcal{F})\to (\smash{\Rnearrow},\mathcal{B}(\smash{\Rnearrow}))$ with the bijective isometry $\Delta$, i.e., $Q = \Delta(\bb{A})$. Given an iid sequence of such random vectors, say $\bb{A}_1,\ldots,\bb{A}_n$, the error term $E_i$ in model~\eqref{eq:our.model} is defined, for every $i\in \{1,\ldots,n\}$, by
\[
E_i = \Delta(\bb{A}_i) = \sum_{j=0}^{d} A_{i,j} (\Phi^{-1})^{j},
\]
where $A_{i,j}$ denotes the $j$th component of $\bb{A}_i$.
\end{remark}

\begin{definition}\label{def:quantile.expectation}
For any $\smash{\mathcal{Q}_d^2}$-valued random quantile function $Q$ as in Definition~\ref{def:random.quantile}, the expectation of $Q$ is defined by
\[
\EE(Q) = \sum_{j=0}^d \EE(A_j) (\Phi^{-1})^j(\cdot).
\]
Note that $\EE(Q)$ itself is a (deterministic) quantile function because $\Rnearrow$ is closed and convex.
\end{definition}

In the special case $d=1$, the quantile errors defined in  Remark~\ref{rem:construction.E.i.general} have the simple form
\[
E_i = \Delta(\bb{A}_i) = A_{i,0} + A_{i,1} \Phi^{-1},
\]
for some iid random vectors $(A_{1,0},A_{1,1}),\ldots,(A_{n,0},A_{n,1})$ taking values in $\Rnearrowtwo \equiv \R\times (0,\infty)$.

For example, if $A_{i,0}$ and $A_{i,1}$ are chosen to be independent and satisfy
\[
\begin{aligned}
&A_{i,0} \sim \mathcal{N}(\mu,\sigma^2), \quad f_{A_{i,0}}(x) = \frac{1}{\sqrt{2\pi\sigma^2}} \exp\bigg(-\frac{(x-\mu)^2}{2\sigma^2}\bigg), \quad x\in \R, \\
&A_{i,1} \sim \mathcal{E}\mathrm{xp}(\beta,\delta), \quad f_{A_{i,1}}(y) = \frac{1}{\beta} \exp\bigg(-\frac{(y-\delta)}{\beta}\bigg) \ind_{[\delta,\infty)}(y), \quad y\in \R,
\end{aligned}
\]
for some parameters $(\mu,\sigma^2,\beta,\delta)\in \R\times (0,\infty)^3$, then the law of each error $E_i$ is defined as a quantile normal-exponential distribution, denoted by $E_i\sim \mathcal{QNE}(\mu,\sigma^2,\beta,\delta)$.

Moreover, for all $(a_0,a_1)\in \R\times (0,\infty)$, note that
\[
\Delta(a_0,a_1) = a_0 + a_1 \Phi^{-1} \equiv q_{\mathcal{N}(a_0,a_1^2)}
\]
is the quantile function of the $\mathcal{N}(a_0,a_1^2)$ distribution. Hence, $E_1,\ldots,E_n\stackrel{\mathrm{iid}}{\sim} \mathcal{QNE}(\mu,\sigma^2,\beta,\delta)$ is the same as stating that each $E_i$ is equal to the random quantile function of the normal distribution $\mathcal{N}(A_{i,0},A_{i,1}^2)$, where $\smash{(A_{1,0},A_{1,1}),\ldots,(A_{n,0},A_{n,1})\stackrel{\mathrm{iid}}{\sim} \mathcal{N}(\mu,\sigma^2) \otimes \mathcal{E}\mathrm{xp}(\beta,\delta)}$, and $\otimes$ stands for the product of two measures. See Figure~\ref{fig:QNE} for an illustration.

\begin{figure}[!t]
\includegraphics[width=0.70\textwidth]{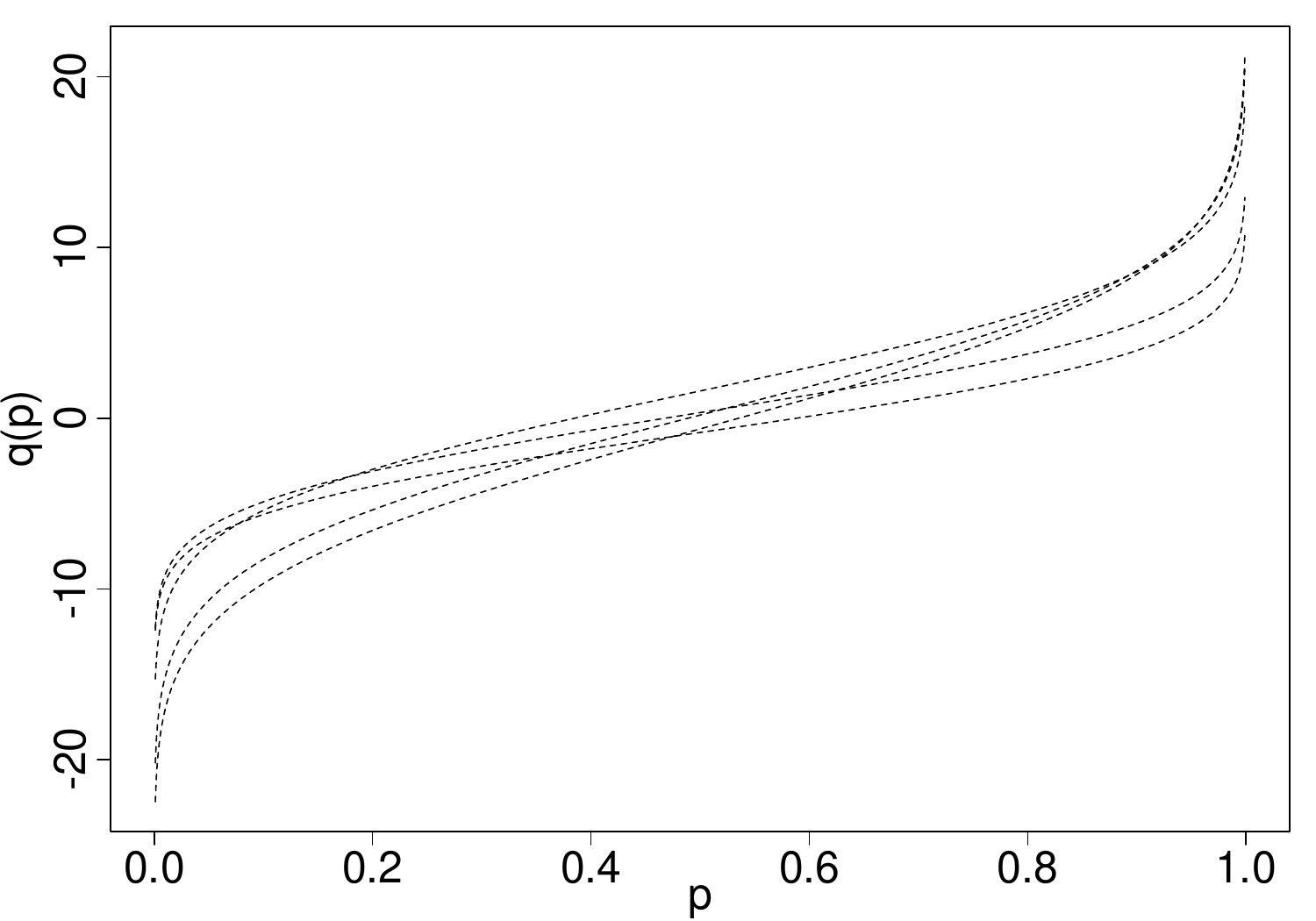}
\caption{A random sample of $n=5$ realizations of the $\mathcal{QNE}(0, 1, 2, 3)$ distribution.}
\label{fig:QNE}
\end{figure}

\subsection{Statistical inference for the simple linear regression model}\label{sec:slr}

As a proof-of-concept for the applicability of the new methodology on the real dataset described in Section~\ref{sec:medical.problem}, this section focuses on the case $d=1$. This choice simplifies the derivation of explicit regression coefficients and confidence intervals. Working with $d=1$ also enhances the interpretability of the results, allowing the findings to be more readily understood by those familiar with traditional simple linear regression techniques. The possibility of extending the model to cases where $d>1$, as well as incorporating multiple explanatory quantile functions (i.e., multiple predictors), remains an opportunity for future research.

The simple linear regression model of interest is
\begin{equation}\label{eq:model}
Q_{Y,i} = \beta_0 + \beta_1 \mu_{q_{x,i}} + \beta_2 q_{x,i}^c + E_i, \quad i\in \{1,\ldots,n\},
\end{equation}
where $E_1,\ldots,E_n\stackrel{\mathrm{iid}}{\sim} \mathcal{QNE}(0,\sigma^2,\beta,0)$ as defined in Section~\ref{sec:math}. This specific distribution of the errors was chosen for its suitability for the lung dataset analyzed in Section~\ref{sec:application}. Additionally, it results in explicit estimators for all parameters in the model (Section~\ref{sec:estimation}), explicit confidence intervals (Section~\ref{sec:CI}), and an explicit density function for the residual quantile functions (Section~\ref{sec:residuals}), facilitating the implementation in computing software at every step.

Since the quantile functions (including the errors $E_i$) are all $\mathcal{Q}_1^2$-valued when $d=1$, then for each $i\in \{1,\ldots,n\}$, one has
\[
Q_{Y,i} = \mu_{Q_{Y,i}} + \sigma_{Q_{Y,i}} \Phi^{-1}, \qquad q_{x,i} = \mu_{q_{x,i}} + \sigma_{q_{x,i}} \Phi^{-1}, \qquad E_i = \mu_{E_i} + \sigma_{E_i} \Phi^{-1}.
\]
Therefore, the model equivalently says that, for each $i\in \{1,\ldots,n\}$,
\begin{equation}\label{eq:model.reformulation}
\begin{aligned}
Q_{Y,i}
&= (\beta_0 + \beta_1 \mu_{q_{x,i}} + \mu_{E_i}) + (\beta_2 \sigma_{q_{x,i}} + \sigma_{E_i}) \Phi^{-1} \\
&\equiv \mu_{Q_{Y_,i}} + \sigma_{Q_{Y,i}} \Phi^{-1},
\end{aligned}
\end{equation}
where $(\mu_{E_1},\sigma_{E_1}), \ldots, (\mu_{E_n},\sigma_{E_n}) \stackrel{\mathrm{iid}}{\sim} \mathcal{N}(0,\sigma^2) \otimes \mathcal{E}\mathrm{xp}(\beta,0)$.

Here, $(\mu_{Q_{Y,i}},\sigma_{Q_{Y,i}})$ is a random vector taking values in $\Rnearrowtwo\equiv \R\times (0,\infty)$, so that $Q_{Y,i}$ is a random quantile function for the $\smash{\mathcal{N}(\mu_{Q_{Y,i}},\sigma_{Q_{Y,i}}^2)}$ distribution. To simplify the exposition of the theory, as often done in classical linear regression, the vector $(\mu_{q_{x,i}},\sigma_{q_{x,i}})$ is assumed to be non-random (deterministic), and $q_{x,i}$ is the corresponding deterministic quantile function for the $\smash{\mathcal{N}(\mu_{q_{x,i}},\sigma_{q_{x,i}}^2)}$ distribution. As before, the centered explanatory quantile functions $q_{x,i}^c$ are defined as
\[
q_{x,i}^c = q_{x,i}-\mu_{q_{x,i}}, \qquad \mu_{q_{x,i}} = \smash{\int_0^1 q_{x,i}(p) \rd p}.
\]
In particular, note that $\mu_{q_{x,i}^c} = 0$ and $\sigma_{q_{x,i}^c} = \sigma_{q_{x,i}}$, so that $q_{x,i}^c$ is the quantile function of the $\smash{\mathcal{N}(0,\sigma_{q_{x,i}}^2)}$ distribution.

The proposition below states the specific distribution of $(\mu_{Q_{Y,i}},\sigma_{Q_{Y,i}})$, which is an immediate consequence of \eqref{eq:model.reformulation}.

\begin{proposition}\label{prop:law.mu.sigma}
The random variables $\mu_{Q_{Y,i}}$ and $\sigma_{Q_{Y,i}}$ are independent, and satisfy
\[
\mu_{Q_{Y,i}} \sim \mathcal{N}(\beta_0 +\beta_1 \mu_{q_{x,i}}, \sigma^{2}),
\qquad
\sigma_{Q_{Y,i}} \sim \mathcal{E}\mathrm{xp}(\beta, \beta_2\sigma_{q_{x,i}}).
\]
Equivalently, one has $Q_{Y,i}\sim \mathcal{QNE}(\beta_0+\beta_1 \mu_{q_{x,i}}, \sigma^{2}, \beta, \beta_2\sigma_{q_{x,i}})$ as per Section~\ref{sec:math}.
\end{proposition}

\begin{remark}
Regardless of their pre-treatment mean values $\mu_{q_{x,i}}$, all post-treatment quantile functions $Q_{Y,i}$ experience a common global shift of $\beta_0\in \R$. The coefficient $\beta_1\in \R$ allows for a patient-specific shift of $\beta_1 \mu_{q_{x,i}}$, which can account for a patient's medical history and health status. As shown above, $\mu_{Q_{Y,i}} = \beta_0 + \beta_1 \mu_{q_{x,i}} + \mu_{E_i}$, so the values of $\beta_0$ and $\beta_1$ will determine if the curve $Q_{Y,i}$ is shifted vertically (upward or downward) with respect to $q_{x,i}$, which corresponds equivalently to a horizontal shift (right or left) of the associated post-treatment normal density. This is illustrated on the left-hand side of Figure~\ref{fig:model.interpretation}.
\end{remark}

\begin{remark}
The coefficient $\beta_2\in (0,\infty)$ is analogous to the slope parameter in classical simple linear regression. Multiplying the centered explanatory quantile function $q_{x,i}^c$ by a coefficient $0<\beta_2<1$ (resp., $\beta_2>1$) induces a shrinkage (resp., stretching) of the curve $q_{x,i}^c$ followed by a clockwise (resp., counterclockwise) rotation around the point $\smash{((q_{x,i})^{-1}(\mu_{q_{x,i}}),0)}$. Equivalently, since $\sigma_{Q_{Y,i}} = \beta_2 \sigma_{q_{x,i}} + \sigma_{E_i}$, the normal density associated with the post-treatment quantile function $Q_{Y,i}$ will be more concentrated around its mean if $\beta_2 < 1$ and more spread out if $\beta_2 > 1$. This is illustrated on the right-hand side of Figure~\ref{fig:model.interpretation}.
\end{remark}

\begin{figure}[!ht]
\centering
\begin{subfigure}{0.49\textwidth}
\includegraphics[width=\textwidth]{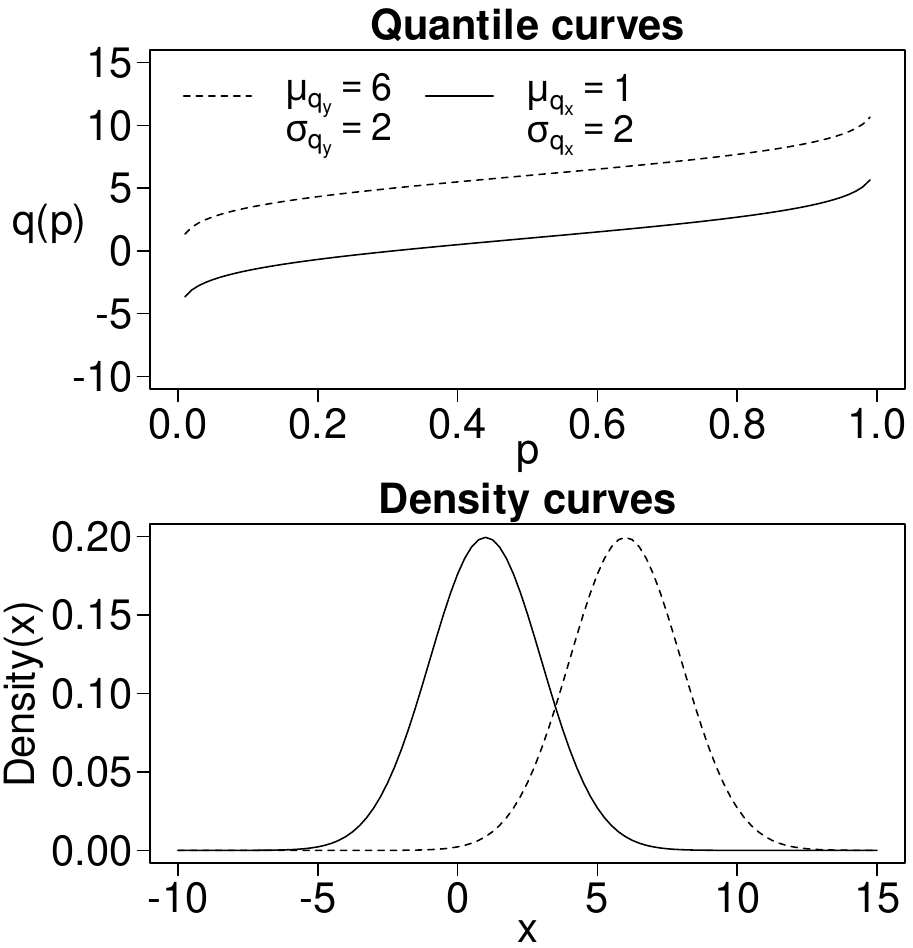}\vspace{-1mm}
\caption{$\beta_0 = 5, \beta_1=1, \beta_2=1$}\label{figure1}
\end{subfigure}
\begin{subfigure}{0.49\textwidth}
\includegraphics[width=\textwidth]{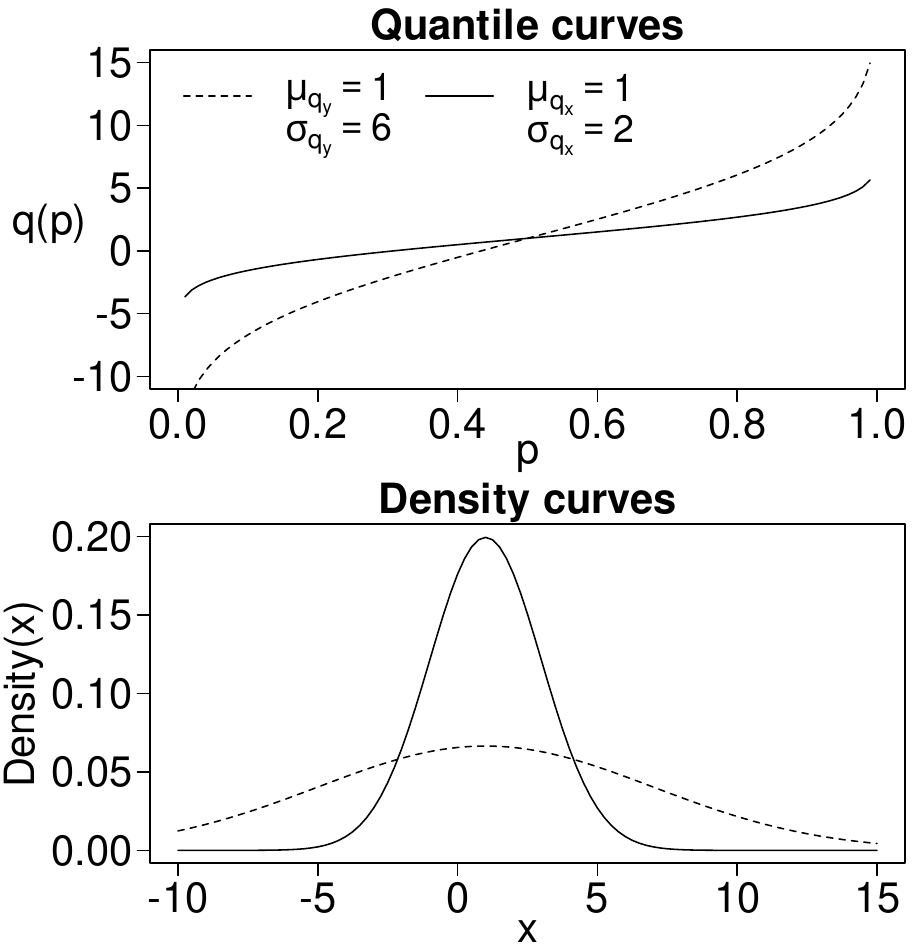}\vspace{-1mm}
\caption{$\beta_0 = 0, \beta_1 = 1, \beta_2 = 3$}\label{figure2}
\end{subfigure}
\caption{Illustration of the impact of the coefficients $\beta_0$, $\beta_1$ and $\beta_2$ on the post-treatment quantile curve (top row) and the associated normal density function (bottom row). For simplicity, the effect of the noise term is ignored. (a) The two subfigures on the left show the effect of $\beta_0$ and $\beta_1$ on the shift $\mu_{q_y} - \mu_{q_x} = \beta_0 + (\beta_1 - 1)\mu_{q_x}$ when the scale remains unchanged ($\beta_2 = 1$). (b) The two subfigures on the right show the effect of $\beta_2$ on the scale $\sigma_{q_y} = \beta_2 \sigma_{q_x}$  when the shift remains unchanged ($\beta_0 = 0$, $\beta_1 = 1$).}
\label{fig:model.interpretation}
\end{figure}

\subsubsection{Estimation of the parameters}\label{sec:estimation}

The next task is to derive explicit expressions for the maximum likelihood (ML) estimators of the model's unknown parameters.

\begin{proposition}[ML estimators and their distributions]\label{prop:MLE}
Assume that the model~\eqref{eq:model} holds true. Define the sample means
\[
\bar{\mu}_{Q_Y} = \frac{1}{n} \sum_{i=1}^n \mu_{Q_{Y,i}}, \quad
\bar{\sigma}_{Q_Y} = \frac{1}{n} \sum_{i=1}^n \sigma_{Q_{Y,i}}, \quad
\bar{\mu}_{q_x} = \frac{1}{n} \sum_{i=1}^n \mu_{q_{x,i}}, \quad
\bar{\sigma}_{q_x} = \frac{1}{n} \sum_{i=1}^n \sigma_{q_{x,i}},
\]
and the sample variance $w = n^{-1} \sum_{i=1}^n \mu_{q_{x,i}}^2 - \bar{\mu}_{q_{x}}^2$. Then, the ML estimators of $\beta_0$, $\beta_1$, $\sigma^2$, $\beta_2$, and $\beta$, are, respectively,
\begin{align*}
\widehat{\beta}_0
&\equiv \widehat{\beta}_{0,\mathrm{ML}}
= \bar{\mu}_{Q_Y} - \widehat{\beta}_1 \bar{\mu}_{q_x}
\sim \mathcal{N}\bigg(\beta_0, \frac{\sigma^2}{n}\bigg(1+\frac{\bar{\mu}_{q_x}^2}{w}\bigg)\bigg) \\[-1mm]
%%%
\widehat{\beta}_1
&\equiv \widehat{\beta}_{1,\mathrm{ML}}
= w^{-1} \bigg(n^{-1}\sum_{i=1}^n \mu_{Q_{Y,i}}\mu_{q_{x,i}} - \bar{\mu}_{Q_Y}\bar{\mu}_{q_x}\bigg)
\sim \mathcal{N}\left(\beta_1,\frac{\sigma^2}{nw}\right), \\[-2mm]
%%%
\widehat{\sigma^2}_{\mathrm{ML}}
&= \frac{1}{n} \sum_{i=1}^n (\mu_{Q_{Y,i}}-\widehat{\beta}_0-\widehat{\beta}_1\mu_{q_{x,i}})^2,
\qquad \frac{n}{\sigma^2} \widehat{\sigma^2}_{\mathrm{ML}} \sim \chi^2(n-2), \\[-1mm]
%%%
\widehat{\beta}_{2,\mathrm{ML}}
&= \min_{1 \leq i \leq n}\left(\frac{\sigma_{Q_{Y,i}}}{\sigma_{q_{x,i}}}\right)
\sim \mathcal{E}\mathrm{xp}\left(\frac{\beta}{n \bar{\sigma}_{q_x}}, \beta_2 \right), \\
%%%
\widehat{\beta}_{\mathrm{ML}}
&= \bar{\sigma}_{Q_Y} - \widehat{\beta}_{2,\mathrm{ML}} \bar{\sigma}_{q_x}
\sim \mathcal{G}\mathrm{amma}\left(n-1,\frac{\beta}{n}\right),
\end{align*}
where $\mathcal{E}\mathrm{xp}$ and $\mathcal{G}\mathrm{amma}$ have scale-shift and shape-scale parameterizations, respectively.
\end{proposition}

\begin{remark}\label{eq:bias}
First, since the expectation of a chi-square random variable corresponds to its degree-of-freedom parameter, Proposition~\ref{prop:MLE} implies that
\[
\EE(\widehat{\sigma^2}_{\mathrm{ML}}) = \frac{n-2}{n} \sigma^2,
\]
and thus $\widehat{\sigma^2}_{\mathrm{ML}}$ is biased. Second, the expectation of a shifted exponential distribution, $\mathcal{E}\mathrm{xp}(\lambda,\delta)$, is equal to $\lambda + \delta$. Hence, by Proposition~\ref{prop:MLE}, one has
\[
\EE(\widehat{\beta}_{2,\mathrm{ML}}) = \beta_2 + \frac{\beta}{n \bar{\sigma}_{q_x}},
\]
which shows that $\widehat{\beta}_{2,\mathrm{ML}}$ is biased. Third, since $(\sigma_{Q_{Y,1}} - \beta_2\sigma_{q_{x,1}}),\ldots,(\sigma_{Q_{Y,n}} - \beta_2\sigma_{q_{x,n}}) \overset{\mathrm{iid}}{\sim} \mathcal{E}\mathrm{xp}(\beta)$ in the scale parametrization, one also has
\[
\bar{\sigma}_{Q_Y} - \beta_2 \bar{\sigma}_{q_x} = \frac{1}{n} \sum_{i=1}^n (\sigma_{Q_{Y,i}} - \beta_2\sigma_{q_{x,i}}) \sim \mathcal{G}\mathrm{amma}\bigg(n, \frac{\beta}{n}\bigg),
\]
so that
\[
\EE(\widehat{\beta}_{\mathrm{ML}}) = \EE(\bar{\sigma}_{Q_Y} - \beta_2 \bar{\sigma}_{q_x}) - \EE((\widehat{\beta}_{2,\mathrm{ML}} - \beta_2) \bar{\sigma}_{q_x}) = \beta - \frac{\beta}{n} = \frac{n-1}{n} \beta,
\]
which shows that $\widehat{\beta}_{\mathrm{ML}}$ is biased.
\end{remark}

The biases of $\widehat{\sigma^2}_{\mathrm{ML}}$, $\widehat{\beta}_{2,\mathrm{ML}}$, and $\widehat{\beta}_{\mathrm{ML}}$, which were highlighted in Remark~\ref{eq:bias}, motivate the following definition.
\begin{definition}\label{def:unbiased.estimators}
Consider the estimators of $\sigma^2$, $\beta_2$, and $\beta$, defined as
\[
\widehat{\sigma^2}
\leqdef \frac{n}{n-2} \widehat{\sigma^2}_{\mathrm{ML}}, \qquad
\widehat{\beta}_2
\leqdef \frac{n}{n-1} \widehat{\beta}_{2,\mathrm{ML}} - \frac{\bar{\sigma}_{Q_Y}}{(n-1) \bar{\sigma}_{q_x}}, \qquad
\widehat{\beta}
\leqdef \frac{n}{n-1} \widehat{\beta}_{\mathrm{ML}}.
\]
\end{definition}

\begin{proposition}\label{prop:unbiased.estimators}
The estimators $\widehat{\sigma^2}$, $\widehat{\beta}_2$, and $\widehat{\beta}$, are unbiased.
\end{proposition}

The next result complements Proposition~\ref{prop:MLE} and Definition~\ref{def:unbiased.estimators} by presenting the marginal distributions of the pivot quantities $\smash{(\widehat{\beta}_2 - \beta_2) / \widehat{\beta}}$ and $\smash{\widehat{\beta}/\beta}$. It will be leveraged to obtain confidence intervals on $\beta_2$ and $\beta$ in Corollary~\ref{cor:CI} below.

\begin{proposition}[Pivot distributions]\label{prop:pivot.distributions}
Using the above notations, one has
\[
\begin{aligned}
\frac{\widehat{\beta}}{\beta}
&\sim \mathcal{G}\mathrm{amma}\bigg(n-1, \frac{1}{n-1}\bigg), \\[-2mm]
\frac{\widehat{\beta}_2 - \beta_2}{\widehat{\beta}} + \frac{1}{n \bar{\sigma}_{q_x}}
&\sim \mathcal{P}\mathrm{areto}\text{-}\mathcal{T}\mathrm{ype}\text{-}\mathrm{II}\bigg(n-1, \frac{n-1}{n \bar{\sigma}_{q_x}}\bigg).
\end{aligned}
\]
\end{proposition}

\begin{remark}\label{rem:hat.beta.2.positive}
Given that the variance of $\smash{\widehat{\beta}_2}$ is asymptotically negligible, Chebyshev's inequality implies that the probability of the unbiased estimator $\smash{\widehat{\beta}_2}$ being negative is also asymptotically negligible. In fact, if $\beta,\beta_2 > 0$, then for any $\e > 0$,
\[
\PP(\widehat{\beta}_2 < \beta_2 - \e) \leq \frac{\beta^2}{n (n-1) \e^2 \bar{\sigma}_{q_x}^2}.
\]
This claim is proved in Section~\ref{sec:proofs}.
\end{remark}

\subsubsection{Confidence intervals and regions}\label{sec:CI}

Given Proposition~\ref{prop:MLE} and Definition~\ref{def:unbiased.estimators}, natural estimators for the variance of $\widehat{\beta}_0$ and $\widehat{\beta}_1$ are
\[
\widehat{\sigma_{\widehat{\beta}_0}^2} \leqdef \frac{\widehat{\sigma^2}}{n} \bigg(1+\frac{\bar{\mu}_{q_x}^2}{w}\bigg), \qquad \widehat{\sigma_{\widehat{\beta}_1}^2} \leqdef \frac{\widehat{\sigma^2}}{nw},
\]
where recall $w = n^{-1} \sum_{i=1}^n \mu_{q_{x,i}}^2 - \bar{\mu}_{q_{x}}^2$. Then, it is easily seen that
\begin{equation}\label{eq:Student}
\frac{\widehat{\beta}_1 - \beta_1}{\widehat{\sigma_{\widehat{\beta}_1}}} \sim \mathcal{T}(n-2), \qquad
\frac{\widehat{\beta}_0 - \beta_0}{\widehat{\sigma_{\widehat{\beta}_0}}} \sim \mathcal{T}(n-2),
\end{equation}
where $\mathcal{T}(n-2)$ denotes Student's $t$ distribution with $n-2$ degrees of freedom.

Combining \eqref{eq:Student} with the results of Propositions~\ref{prop:MLE}~and~\ref{prop:pivot.distributions} allows for the derivation of two-sided confidence intervals for all five parameters $\beta_0$, $\beta_1$, $\sigma^2$, $\beta_2$, and $\beta$, in the model~\eqref{eq:model}, as detailed in the corollary below.

\begin{corollary}\label{cor:CI}
Let $t_{n-2,\alpha}$, $\smash{\chi_{n-2,\alpha}^2}$, $G_{\alpha}$ and $P_{\alpha}$, denote the quantiles at a given significance level $\alpha\in (0,1)$ for the $\mathcal{T}(n-2)$, $\chi_{n-2}^2$, $\mathcal{G}\mathrm{amma}(n-1,1/(n-1))$ and $\mathcal{P}\mathrm{areto}\text{-}\mathcal{T}\mathrm{ype}\text{-}\mathrm{II}(n-1, (1-1/n)/\bar{\sigma}_{q_x})$ distributions, respectively. The following two-sided confidence intervals for $\beta_0$, $\beta_1$, $\sigma^2$, $\beta_2$, and $\beta$, are valid with confidence level $1-\alpha$:
\[
\begin{aligned}
\mathrm{C\, I}_{1-\alpha}(\beta_0)
&= \left[\widehat{\beta}_0 - t_{n-2,1-\alpha/2} \widehat{\sigma_{\widehat{\beta}_0}}, \widehat{\beta}_0 +t_{n-2,1-\alpha/2} \widehat{\sigma_{\widehat{\beta}_0}}\right], \\
%%%
\mathrm{C\, I}_{1-\alpha}(\beta_1)
&= \left[\widehat{\beta}_1 - t_{n-2,1-\alpha/2} \widehat{\sigma_{\widehat{\beta}_1}}, \widehat{\beta}_1 +t_{n-2,1-\alpha/2} \widehat{\sigma_{\widehat{\beta}_1}}\right], \\
%%%
\mathrm{C\, I}_{1-\alpha}(\sigma^2)
&= \left[(n-2) \widehat{\sigma^2} / \chi_{n-2,1-\alpha/2}^2, (n-2) \widehat{\sigma^2} / \chi_{n-2,\alpha/2}^2\right], \\
%%%
\mathrm{C\, I}_{1-\alpha}(\beta_2)
&= \left[\widehat{\beta}_2 - \widehat{\beta} (P_{1-\alpha/2} - 1/(n \bar{\sigma}_{q_x})), \widehat{\beta}_2 - \widehat{\beta} (P_{\alpha/2} - 1/(n \bar{\sigma}_{q_x}))\right], \\
%%%
\mathrm{C\, I}_{1-\alpha}(\beta)
&= \left[\widehat{\beta} / G_{1-\alpha/2}, \widehat{\beta} / G_{\alpha/2}\right].
\end{aligned}
\]
\end{corollary}

Now, assume a new pre-treatment quantile function observation $q_{x,n+1}$ is added to the sample and one wants to estimate the corresponding post-treatment mean response $\EE(Q_{Y,n+1})$. Since the pair $(q_{x,n+1},Q_{Y,n+1})$ is assumed to follow model~\eqref{eq:model}, one can write
\[
Q_{Y,n+1} = \beta_0 + \beta_1 \mu_{q_{x,n+1}} + \beta_2 q_{x,n+1}^c + E_{n+1},
\]
where $E_{n+1}\sim \mathcal{QNE}(0,\sigma^2,\beta,0)$. Using Definition~\ref{def:quantile.expectation} to take the expectation of the quantile functions on both sides of the above equation yields
\[
\begin{aligned}
\EE(Q_{Y,n+1})
&= \beta_0 + \beta_1 \mu_{q_{x,n+1}} + \beta_2 (0 + \sigma_{q_{x,n+1}} \Phi^{-1}) + (0 + \beta \Phi^{-1}) \\
&= (\beta_0 + \beta_1 \mu_{q_{x,n+1}}) + (\beta_2 \sigma_{q_{x,n+1}} + \beta) \Phi^{-1} \\
&\equiv \mu_{\EE(Q_{Y,n+1})} + \sigma_{\EE(Q_{Y,n+1})} \Phi^{-1}.
\end{aligned}
\]
Then the new mean response $\EE(Q_{Y,n+1})$ is estimated by
\[
\widehat{Q}_{Y,n+1} \equiv \widehat{\EE}(Q_{Y,n+1}) \equiv \mu_{\widehat{Q}_{Y,n+1}} + \sigma_{\widehat{Q}_{Y,n+1}} \Phi^{-1},
\]
where
\[
\mu_{\widehat{Q}_{Y,n+1}} \leqdef \widehat{\beta}_0 + \widehat{\beta}_1 \mu_{q_{x,n+1}}, \qquad
\sigma_{\widehat{Q}_{Y,n+1}} \leqdef \widehat{\beta}_2 \sigma_{q_{x,n+1}} + \widehat{\beta}.
\]

The next proposition presents an approximate $(1-\alpha)$-level confidence region for the associated vector $(\mu_{\EE(Q_{Y,n+1})},\sigma_{\EE(Q_{Y,n+1})})$, which characterizes the new mean response $\EE(Q_{Y,n+1})$. Because of the bijective isometry $\Delta$ constructed in Section~\ref{sec:math}, this automatically yields an approximate $(1-\alpha)$-level confidence region for $\EE(Q_{Y,n+1})$ in $\mathcal{Q}_1^2$.

\begin{proposition}\label{prop:CR.new.mean.response}
An approximate density function for $\smash{(\mu_{\widehat{Q}_{Y,n+1}},\sigma_{\widehat{Q}_{Y,n+1}})}$ is defined, for all $(s,t)\in \R \times \smash{(\widehat{\beta}_2 \sigma_{q_{x,n+1}},\infty)}$, by
\[
\begin{aligned}
&\widehat{f}_{\widehat{Q}_{Y,n+1}}(s,t)
\equiv \widehat{f}_{\mu_{\widehat{Q}_{Y,n+1}},\sigma_{\widehat{Q}_{Y,n+1}}}(s,t) \\
&\leqdef \frac{1}{\sqrt{2\pi \frac{\widehat{\sigma^2}}{n} \left(1 + \frac{(\mu_{q_{x,n+1}} - \bar{\mu}_{q_x})^2}{w}\right)}} \exp\left(- \frac{(s - (\widehat{\beta}_0 + \mu_{q_{x,n+1}} \widehat{\beta}_1))^2}{2 \frac{\widehat{\sigma^2}}{n} \left(1 + \frac{(\mu_{q_{x,n+1}} - \bar{\mu}_{q_x})^2}{w}\right)}\right) \\
&~\times \frac{\frac{n \bar{\sigma}_{q_x}}{\widehat{\beta} \sigma_{q_{x,n+1}}} \exp\left(\frac{- n \bar{\sigma}_{q_x}}{\widehat{\beta} \sigma_{q_{x,n+1}}} (t - \widehat{\beta}_2 \sigma_{q_{x,n+1}})\right)}{\left(\frac{n}{n-1} \left(1 - \frac{\bar{\sigma}_{q_x}}{\sigma_{q_{x,n+1}}}\right)\right)^{n-1}} \, \overline{\gamma}\left(n-1, n \frac{\left(1 - \frac{\bar{\sigma}_{q_x}}{\sigma_{q_{x,n+1}}}\right)}{\left(1 - \frac{\sigma_{q_{x,n+1}}}{n \bar{\sigma}_{q_x}}\right)} \frac{(t - \widehat{\beta}_2 \sigma_{q_{x,n+1}})}{\widehat{\beta}}\right),
\end{aligned}
\]
(here, $\widehat{\beta}_0$, $\widehat{\beta}_1$, $\widehat{\sigma^2}$, $\widehat{\beta}_2$, and $\widehat{\beta}$, are the observed values of the statistics) with $\overline{\gamma}$ denoting the regularized lower incomplete gamma function
\[
\overline{\gamma}(a,b) = \frac{1}{\Gamma(a)} \int_0^b t^{a-1} e^{-t} \rd t, \quad (a,b)\in (0,\infty) \times \R.
\]
Under the restriction $\smash{\sigma_{q_{x,n+1}} < n \bar{\sigma}_{q_x}}$, and for any significance level $\alpha\in (0,1)$, an approximate $(1-\alpha)$-level confidence region for $(\mu_{\EE(Q_{Y,n+1})},\sigma_{\EE(Q_{Y,n+1})})$ is
\[
\mathrm{CR}_{1-\alpha}(\mu_{\EE(Q_{Y,n+1})},\sigma_{\EE(Q_{Y,n+1})}) = \{(s,t)\in \R \times (\widehat{\beta}_2 \sigma_{q_{x,n+1}},\infty) : \widehat{f}_{\widehat{Q}_{Y,n+1}}(s,t) \geq L_{\alpha}\},
\]
where the threshold $L_{\alpha}\in (0,\infty)$ solves the equation
\[
1 - \alpha = \int_{(s,t)\in \R \times (\widehat{\beta}_2 \sigma_{q_{x,n+1}},\infty) : \widehat{f}_{\widehat{Q}_{Y,n+1}}(s,t) \geq L_{\alpha}} \widehat{f}_{\widehat{Q}_{Y,n+1}}(s,t) \rd s \rd t.
\]
\end{proposition}

The approximate density function $\smash{\widehat{f}_{\widehat{Q}_{Y,n+1}}}$ defining the confidence region for the vector $\smash{(\mu_{\EE(Q_{Y,n+1})},\sigma_{\EE(Q_{Y,n+1})})}$ in Proposition~\ref{prop:CR.new.mean.response} is illustrated in Figure~\ref{fig:CR.new.mean.response} below. It was verified numerically that $L_{0.01} \approx 0.000033$, $L_{0.05} \approx 0.000164$, and $L_{0.10} \approx 0.000328$.

\begin{figure}[!ht]
\centering
\includegraphics[width=0.80\textwidth]{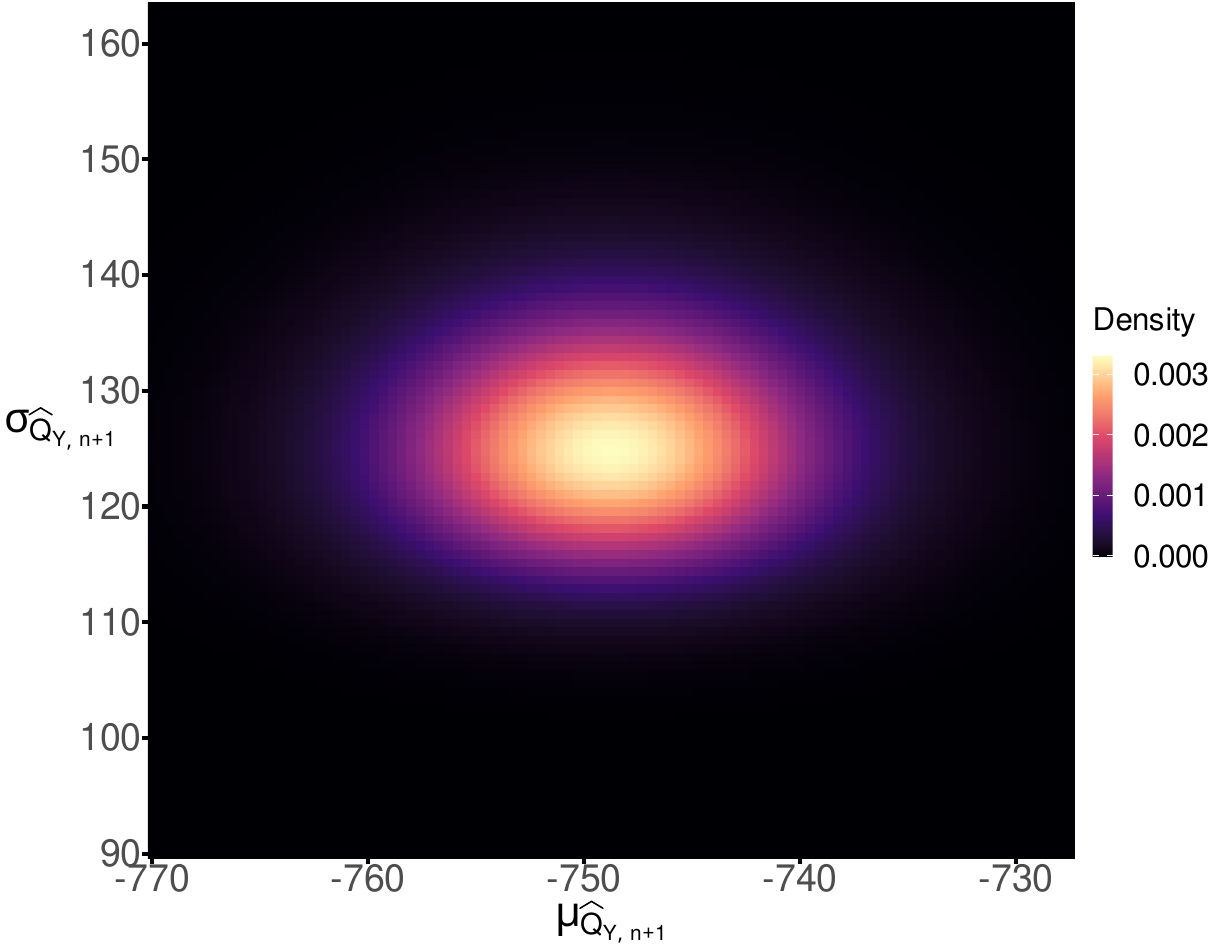}
\caption{Density plot for the approximate density function $\smash{\widehat{f}_{\widehat{Q}_{Y,n+1}}}$, using the observations of the lung dataset with arbitrarily chosen mean $\mu_{q_{x,n+1}} = -750$ and scale $\sigma_{q_{x,n+1}} = 120$. An interactive 3d widget is available at \href{https://biostatisticien.eu/Qlm/fig-3.3.html}{biostatisticien.eu/Qlm/fig-3.3.html}.}
\label{fig:CR.new.mean.response}
\end{figure}

\subsection{Residual quantile functions}\label{sec:residuals}

Using Definition~\ref{def:quantile.expectation} to take the expectation of the quantile functions on both sides of \eqref{eq:model.reformulation} yields
\[
\begin{aligned}
\EE(Q_{Y,i})
= (\beta_0 + \beta_1 \mu_{q_{x,i}} + 0) + (\beta_2 \sigma_{q_{x,i}} + \beta) \Phi^{-1}
\equiv \mu_{\EE(Q_{Y,i})} + \sigma_{\EE(Q_{Y,i})} \Phi^{-1}.
\end{aligned}
\]
The model-adjusted (i.e., fitted) value for this last expectation is thus
\[
\widehat{Q}_{Y,i} \equiv \widehat{\EE}(Q_{Y,i}) \equiv \mu_{\widehat{Q}_{Y,i}} + \sigma_{\widehat{Q}_{Y,i}} \Phi^{-1},
\]
\vspace{-2mm}
where
\[
\mu_{\widehat{Q}_{Y,i}} \leqdef \widehat{\beta}_0 + \widehat{\beta}_1 \mu_{q_{x,i}}, \qquad
\sigma_{\widehat{Q}_{Y,i}} \leqdef \widehat{\beta}_2 \sigma_{q_{x,i}} + \widehat{\beta}.
\]

It follows that the residual quantile functions $\widehat{E}_1,\ldots,\widehat{E}_n$ of the simple linear regression model~\eqref{eq:model} on quantile functions are defined, for every $i\in \{1,\ldots,n\}$, by
\[
\begin{aligned}
\widehat{E}_i
&\leqdef Q_{Y,i} - \widehat{Q}_{Y,i} + \widehat{\beta} \Phi^{-1} \equiv \mu_{\widehat{E}_i} + \sigma_{\widehat{E}_i} \Phi^{-1},
\end{aligned}
\]
where
\begin{equation}\label{eq:mu.sigma.residuals}
\begin{aligned}
\mu_{\widehat{E}_i}
&= \mu_{Q_{Y,i}} - \mu_{\widehat{Q}_{Y,i}} = (\beta_0 + \beta_1 \mu_{q_{x,i}} +\mu_{E_i}) - (\widehat{\beta}_0 + \widehat{\beta}_1 \mu_{q_{x,i}}), \\ \sigma_{\widehat{E}_i}
&= \sigma_{Q_{Y,i}} - \sigma_{\widehat{Q}_{Y,i}} + \widehat{\beta} = (\beta_2 \sigma_{q_{x,i}} + \sigma_{E_i}) - \widehat{\beta}_2 \sigma_{q_{x,i}}.
\end{aligned}
\end{equation}
Since $\widehat{\beta}_0$, $\widehat{\beta}_1$, and $\widehat{\beta}_2$, are unbiased estimators, it follows that $\smash{(\mu_{\widehat{E}_i},\sigma_{\widehat{E}_i})}$ is an unbiased estimator of $(\mu_{E_i},\sigma_{E_i})$. The distribution of $\smash{(\mu_{\widehat{E}_i},\sigma_{\widehat{E}_i})}$ is presented in the next proposition.

\begin{proposition}\label{prop:residuals}
For every $i\in \{1,\ldots,n\}$, the random variables $\mu_{\widehat{E}_i}$ and $\sigma_{\widehat{E}_i}$ are independent. Furthermore,
\vspace{-1mm}
\[
\mu_{\widehat{E}_i}\sim \mathcal{N}\left(0, \sigma^2 \xi_n\right), \quad \xi_n = \left(1 - \frac{1}{n} - \frac{(\mu_{q_{x,i}} - \bar{\mu}_{q_x})^2}{n w}\right),
\]
where recall $w = n^{-1} \sum_{i=1}^n \mu_{q_{x,i}}^2 - \bar{\mu}_{q_{x}}^2$, and the density function of $\sigma_{\widehat{E}_i}$ is, for all $t\in (0,\infty)$,
\[
f_{\sigma_{\widehat{E}_i}}(t) = \frac{\sigma_{q_{x,i}}}{n \bar{\sigma}_{q_x}} f_{\mathcal{G}\mathrm{amma}(n-1,\nu_{i,2})}(t) + \left(1 - \frac{\sigma_{q_{x,i}}}{n \bar{\sigma}_{q_x}}\right) \frac{e^{-\nu_{i,1}^{-1}t}}{\nu_{i,1}} \left(\frac{\theta}{\nu_{i,2}}\right)^{n-2} F_{\mathcal{G}\mathrm{amma}(n-2, \theta)}(t),
\]
where
\vspace{-2mm}
\[
\nu_{i,1} \leqdef \beta + \nu_{i,2}, \qquad \nu_{i,2} \leqdef \frac{\beta \sigma_{q_{x,i}}}{n(n-1) \bar{\sigma}_{q_x}}, \qquad \theta \leqdef \left(\frac{1}{\nu_{i,2}} - \frac{1}{\nu_{i,1}}\right)^{-1}.
\]
An approximate density function for $(\mu_{\widehat{E}_i},\sigma_{\widehat{E}_i})$ is defined, for all $(s,t)\in \R \times (0,\infty)$, by
\[
\begin{aligned}
&\widehat{f}_{\widehat{E}_i}(s,t)
\equiv \widehat{f}_{\mu_{\widehat{E}_i},\sigma_{\widehat{E}_i}}(s,t) \leqdef \frac{1}{\sqrt{2\pi \widehat{\sigma^2} \xi_n}} \exp\left(- \frac{s^2}{2 \widehat{\sigma^2} \xi_n}\right) \\[-2mm]
&\quad\times \left[\frac{\sigma_{q_{x,i}}}{n \bar{\sigma}_{q_x}} f_{\mathcal{G}\mathrm{amma}(n-1,\widehat{\nu}_{i,2})}(t) + \left(1 - \frac{\sigma_{q_{x,i}}}{n \bar{\sigma}_{q_x}}\right) \frac{e^{-\widehat{\nu}_{i,1}^{-1}t}}{\widehat{\nu}_{i,1}} \left(\frac{\widehat{\theta}}{\widehat{\nu}_{i,2}}\right)^{n-2} F_{\mathcal{G}\mathrm{amma}(n-2, \widehat{\theta})}(t)\right],
\end{aligned}
\]
(here, $\widehat{\sigma^2}$ and $\widehat{\beta}$ are the observed values of the statistics) where
\[
\widehat{\nu}_{i,1} \leqdef \widehat{\beta} + \widehat{\nu}_{i,2}, \qquad \widehat{\nu}_{i,2} \leqdef \frac{\widehat{\beta} \sigma_{q_{x,i}}}{n(n-1) \bar{\sigma}_{q_x}}, \qquad \widehat{\theta} \leqdef \left(\frac{1}{\widehat{\nu}_{i,2}} - \frac{1}{\widehat{\nu}_{i,1}}\right)^{-1}.
\]
See Figure~\ref{fig:CR.residuals} for an illustration of the density plot of $\smash{\widehat{f}_{\widehat{E}_{10}}}$.
\end{proposition}

\begin{figure}[!ht]
\centering
\includegraphics[width=0.80\textwidth]{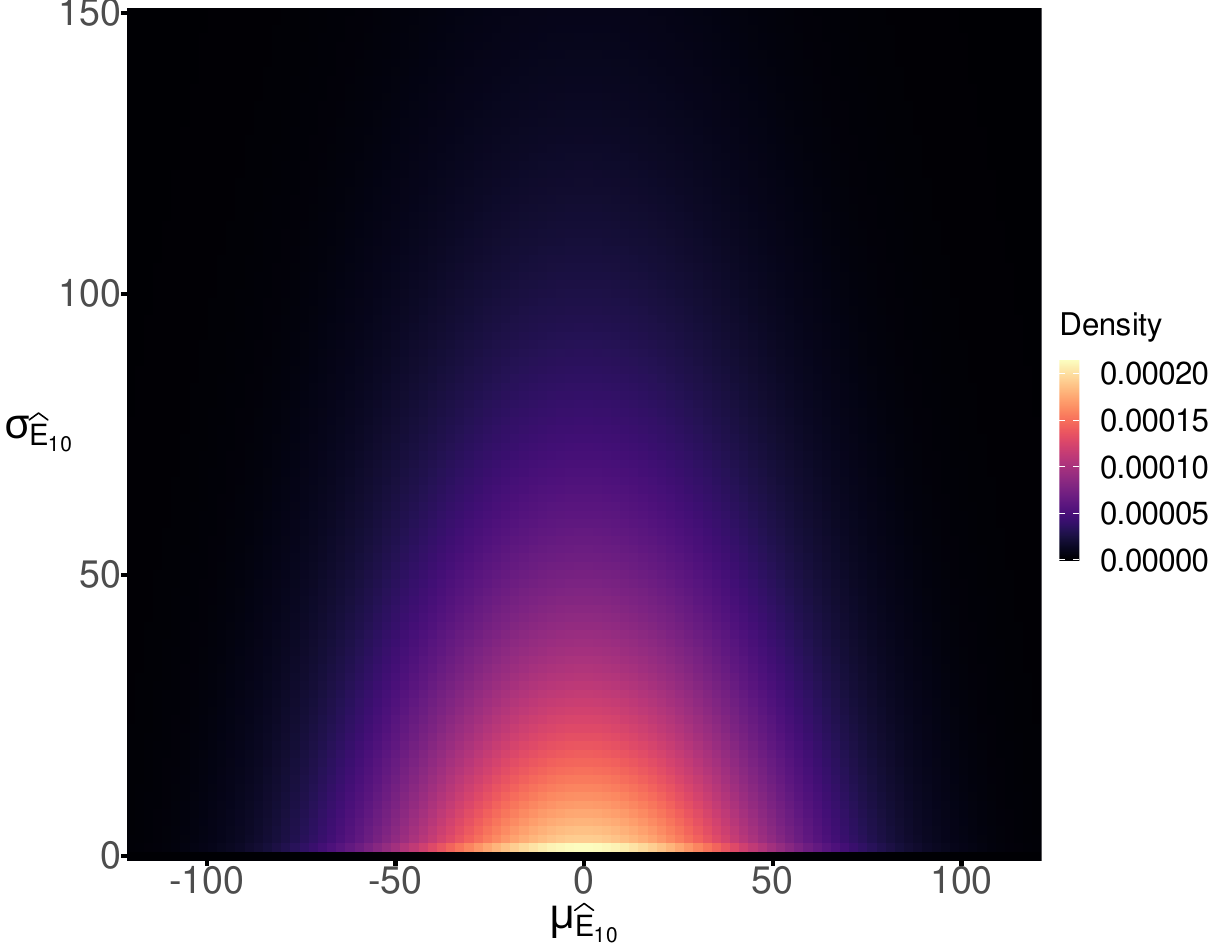}
\caption{Density plot for the approximate density function $\widehat{f}_{\widehat{E}_{10}}$, using the observations of the lung dataset. An interactive 3d widget is available at \href{https://biostatisticien.eu/Qlm/fig-3.4.html}{biostatisticien.eu/Qlm/fig-3.4.html}.}
\label{fig:CR.residuals}
\end{figure}

\section{Application to assessing treatment response in asthmatic patients}\label{sec:application}

The code and dataset used to produce the results in this section are available online in the GitHub repository of \citet{BeclinMicheauxMolinariOuimet_github_2024}, ensuring full reproducibility of the method.

This study aims to analyze segmented expiration-phase lung CT scans from 44 asthmatic patients treated with Benralizumab. The general hypothesis is that Benralizumab treatment is capable of reducing air-trapping and partially reversing bronchial remodeling as detected on CT scans.

The initial step in the analysis involves creating the $n=44$ pairs of (pre, post)-treatment Gaussian quantile functions $\{(q_{x,i}, q_{y,i})\}_{i=1}^{44}$ as outlined in Section~\ref{sec:raw.to.quantile}, where
\[
q_{x,i} = \mu_{q_{x,i}} + \sigma_{q_{x,i}}
\Phi^{-1}, \qquad q_{y,i} = \mu_{q_{y,i}} + \sigma_{q_{y,i}}
\Phi^{-1}.
\]
Recall that for each patient $i\in \{1,\ldots,44\}$, the (pre, post)-treatment CT scans are segmented, the empirical quantile functions $(\widehat{q}_{x,i},\widehat{q}_{y,i})$ are calculated using the HU values of the segmented images, and then each quantile function is orthogonally projected onto the space $\mathcal{Q}_1^2$ of Gaussian quantile functions to obtain $(q_{x,i}, q_{y,i})$. The results are displayed in Table~\ref{tab:Gaussian.quantile.functions} below for convenience.

\begin{table}[ht]
\caption{Parameters of the Gaussian quantile functions for each patient before and after treatment.}
\label{tab:Gaussian.quantile.functions}
\centering
\setlength{\tabcolsep}{2pt} % Set global column separation to x points
\newcommand{\colskipextra}{0.25cm} % Extra horizontal skip between before and after treatment
\renewcommand{\arraystretch}{1} % Reduce vertical spacing
\begin{tabular}{|c|D{.}{.}{4.4}|D{.}{.}{3.4}|D{.}{.}{4.4}|D{.}{.}{3.4}||c|D{.}{.}{4.4}|D{.}{.}{3.4}|D{.}{.}{4.4}|D{.}{.}{3.4}|}
\hline
\rule{0pt}{2.3ex} % Adjust height for next row
\!\raisebox{-0.5mm}{Patient} & \multicolumn{2}{@{\hskip \colskipextra}c|}{\!\!\!Before treatment} & \multicolumn{2}{@{\hskip \colskipextra}c||}{\!\!\!\!After treatment} & \raisebox{-0.5mm}{Patient} & \multicolumn{2}{@{\hskip \colskipextra}c|}{\!\!\!Before treatment} & \multicolumn{2}{@{\hskip \colskipextra}c|}{\!\!\!\!After treatment} \\[0.2mm]
\hline
\raisebox{-0.5mm}{$i$} & \multicolumn{1}{c}{$\mu_{q_{x,i}}$} \vline & \multicolumn{1}{c}{$\sigma_{q_{x,i}}$} \vline & \multicolumn{1}{c}{$\mu_{q_{y,i}}$} \vline & \multicolumn{1}{c||}{$\sigma_{q_{y,i}}$} &
\raisebox{-0.5mm}{$i$} & \multicolumn{1}{c}{$\mu_{q_{x,i}}$} \vline & \multicolumn{1}{c}{$\sigma_{q_{x,i}}$} \vline & \multicolumn{1}{c}{$\mu_{q_{y,i}}$} \vline & \multicolumn{1}{c|}{$\sigma_{q_{y,i}}$} \\[0.4mm]
\hline
\rule{0pt}{2ex} % Adjust height for next row
\!\!1 & -691.4931 & 140.8254 & -739.0630 & 141.1091 & 23 & -797.1403 & 114.8312 & -801.1290 & 114.5843 \\
2 & -678.9295 & 132.6967 & -685.7835 & 169.9798 & 24 & -769.5303 & 117.6721 & -742.4835 & 125.0238 \\
3 & -708.7794 & 132.5881 & -698.7831 & 140.4667 & 25 & -729.5424 & 126.2841 & -761.4966 & 123.0142 \\
4 & -711.9193 & 110.9741 & -742.1448 & 109.8616 & 26 & -769.6909 & 156.7329 & -810.1636 & 132.2282 \\
5 & -622.1860 & 126.5870 & -603.8744 & 132.1156 & 27 & -675.9667 & 120.6983 & -661.6918 & 121.5206 \\
6 & -740.1834 & 125.1759 & -675.4633 & 139.0663 & 28 & -725.9413 & 127.7271 & -746.8861 & 124.6374 \\
7 & -734.4035 & 115.5035 & -750.5532 & 157.8889 & 29 & -691.2101 & 129.3730 & -687.5860 & 123.6381 \\
8 & -773.4940 & 133.1161 & -742.1335 & 167.1660 & 30 & -635.2004 & 160.1849 & -592.4990 & 124.1066 \\
9 & -640.7417 & 121.3690 & -534.6137 & 125.0477 & 31 & -860.3746 & 90.8649 & -845.2094 & 120.2888 \\
10 & -767.5906 & 121.9286 & -742.6850 & 122.7488 & 32 & -730.0569 & 156.3903 & -702.8203 & 168.2426 \\
11 & -671.7848 & 142.4396 & -814.0957 & 104.1052 & 33 & -724.0669 & 142.6923 & -735.2466 & 147.5985 \\
12 & -698.9207 & 129.0829 & -706.4403 & 113.1154 & 34 & -751.6976 & 154.1182 & -736.8656 & 150.1756 \\
13 & -840.6441 & 115.6835 & -831.9079 & 122.9021 & 35 & -644.5101 & 146.0220 & -666.7298 & 149.3380 \\
14 & -780.6131 & 121.8097 & -732.3012 & 144.7658 & 36 & -676.1773 & 120.8362 & -657.8814 & 132.1409 \\
15 & -577.8501 & 136.5012 & -591.8307 & 199.5379 & 37 & -672.8262 & 144.2112 & -704.3792 & 131.1094 \\
16 & -676.0956 & 115.6735 & -682.8546 & 115.6614 & 38 & -755.1644 & 134.9401 & -746.4215 & 139.3344 \\
17 & -681.6340 & 162.6760 & -704.4922 & 159.7900 & 39 & -750.2510 & 108.4672 & -743.3522 & 107.7744 \\
18 & -854.7864 & 118.6292 & -869.9233 & 113.3971 & 40 & -812.7209 & 112.6021 & -788.3978 & 123.8007 \\
19 & -705.9140 & 126.1935 & -736.1574 & 115.9091 & 41 & -680.6981 & 146.7967 & -714.1882 & 136.9068 \\
20 & -774.8398 & 116.1520 & -785.5845 & 111.3199 & 42 & -708.0478 & 130.8564 & -715.3282 & 129.8226 \\
21 & -787.0903 & 107.2140 & -736.5424 & 141.5978 & 43 & -791.2589 & 99.2491 & -750.1276 & 107.5248 \\
22 & -706.0180 & 156.3191 & -795.2887 & 101.1520 & 44 & -721.4120 & 126.1214 & -758.5488 & 114.5121 \\
\hline
\end{tabular}
\end{table}

After transforming the data, the next step involves applying the newly developed simple linear regression method from Section~\ref{sec:slr} to the $n=44$ pairs of Gaussian quantile functions. This method is implemented in the \texttt{R} function \texttt{qlm()}, which offers features analogous to those of the classical \texttt{lm()} function, as demonstrated below.

\begin{quote}
\begin{verbatim}
> res.qlm <- qlm(qy ~ qx, data = list(qx = qx, qy = qy))
> coef(res.qlm)
    beta0hat     beta1hat     beta2hat    sigma2hat      betahat
 -85.9158466    0.8837675    0.6383911 1616.8402929   49.3636853
> confint(res.qlm)
                 2.5 %       97.5 %
beta0hat  -232.1414243   60.3097310
beta1hat     0.6827804    1.0847547
beta2hat     0.6135934    0.6468666
sigma2hat 1099.2369447 2611.9533531
betahat     37.3889592   68.2096823
> summary(res.qlm)

Call:
qlm(formula = qy ~ qx, data = list(qx = qx, qy = qy))


Coefficients:
          Estimate Std. Error t value Pr(>|t|)
beta0hat -85.91585   72.45771  -1.186    0.242
beta1hat   0.88377    0.09959   8.874  3.5e-11 ***
           Estimate Std. Error stat value P-val (H0: par >= 1)
beta2hat  6.384e-01  8.796e-03     -0.007               <2e-16 ***
sigma2hat 1.617e+03  3.528e+02  67907.292                    1
betahat   4.936e+01  7.528e+00     49.364                    1
---
Signif. codes:  0 ‘***’ 0.001 ‘**’ 0.01 ‘*’ 0.05 ‘.’ 0.1 ‘ ’ 1
\end{verbatim}

\begin{verbatim}
> res.qlm$sigmaqxbar
129.0184
\end{verbatim}

\begin{verbatim}
> res.qlm$w
3704.718
\end{verbatim}

\begin{verbatim}
> res.qlm$muqxbar
-724.9863
\end{verbatim}
\end{quote}

Before interpreting the results, the residual quantile functions
\[
\widehat{e}_i = q_{y,i} - \widehat{q}_{y,i} + \widehat{\beta} \Phi^{-1} \equiv \mu_{\widehat{e}_i} + \sigma_{\widehat{e}_i} \Phi^{-1}, \quad i\in \{1,\ldots,44\},
\]
are computed using the \texttt{R} command \texttt{residuals(res.qlm)}. The pairs $(\mu_{\widehat{e}_i},\sigma_{\widehat{e}_i})_{i=1}^{44}$ are listed in Table~\ref{tab:residuals} and plotted in Figure~\ref{fig:residuals.scatterplot}.

\begin{table}[!t]
\caption{Values of the pairs $(\mu_{\widehat{e}_i}, \sigma_{\widehat{e}_i})$ for the residual quantile functions $\widehat{e}_1, \ldots, \widehat{e}_{44}$.}
\label{tab:residuals}
\centering
\setlength{\tabcolsep}{3pt} % Set global column separation
\renewcommand{\arraystretch}{1} % Reduce vertical spacing
\begin{tabular}{|c|D{.}{.}{4.2}|D{.}{.}{2.2}|c|D{.}{.}{3.2}|D{.}{.}{3.2}|c|D{.}{.}{3.2}|D{.}{.}{2.2}|c|D{.}{.}{3.2}|D{.}{.}{2.2}|}
\hline
\!\raisebox{-0.5mm}{Patient} & \multicolumn{1}{c}{$\mu_{\widehat{e}_i}$} \vline & \multicolumn{1}{c}{$\sigma_{\widehat{e}_i}$} \vline &
\raisebox{-0.5mm}{Patient} & \multicolumn{1}{c}{$\mu_{\widehat{e}_i}$} \vline & \multicolumn{1}{c}{$\sigma_{\widehat{e}_i}$} \vline &
\raisebox{-0.5mm}{Patient} & \multicolumn{1}{c}{$\mu_{\widehat{e}_i}$} \vline & \multicolumn{1}{c}{$\sigma_{\widehat{e}_i}$} \vline &
\raisebox{-0.5mm}{Patient} & \multicolumn{1}{c}{$\mu_{\widehat{e}_i}$} \vline & \multicolumn{1}{c}{$\sigma_{\widehat{e}_i}$} \vline \\
\hline
\rule{0pt}{2ex} % Adjust height for next row
\!1 & -42.03 & 51.21 & 12 & -2.84 & 30.71 & 23 & -10.73 & 41.28 & 34 & 13.38 & 51.79 \\
2 & 0.15 & 85.27 & 13 & -3.06 & 49.05 & 24 & 23.52 & 49.90 & 35 & -11.22 & 56.12 \\
3 & 13.53 & 55.82 & 14 & 43.50 & 67.00 & 25 & -30.83 & 42.40 & 36 & 25.62 & 55.00 \\
4 & -27.06 & 39.02 & 15 & 4.77 & 112.40 & 26 & -44.02 & 32.17 & 37 & -23.84 & 39.05 \\
5 & 31.91 & 51.30 & 16 & 0.57 & 41.82 & 27 & 21.62 & 44.47 & 38 & 6.88 & 53.19 \\
6 & 64.60 & 59.16 & 17 & -16.17 & 55.94 & 28 & -19.41 & 43.10 & 39 & 5.61 & 38.53 \\
7 & -15.60 & 84.15 & 18 & -28.57 & 37.67 & 29 & 9.20 & 41.05 & 40 & 15.77 & 51.92 \\
8 & 27.37 & 82.19 & 19 & -26.38 & 35.35 & 30 & 54.79 & 21.85 & 41 & -26.69 & 43.19 \\
9 & 117.57 & 47.57 & 20 & -14.89 & 37.17 & 31 & 1.08 & 62.28 & 42 & -3.66 & 46.28 \\
10 & 21.60 & 44.91 & 21 & 44.98 & 73.15 & 32 & 28.30 & 68.40 & 43 & 35.08 & 44.17 \\
11 & -134.48 & 13.17 & 22 & -85.42 & 1.36 & 33 & -9.42 & 56.50 & 44 & -35.07 & 34.00 \\
\hline
\end{tabular}
\end{table}

\begin{figure}[!b]
\centering
\includegraphics[width=0.85\textwidth]{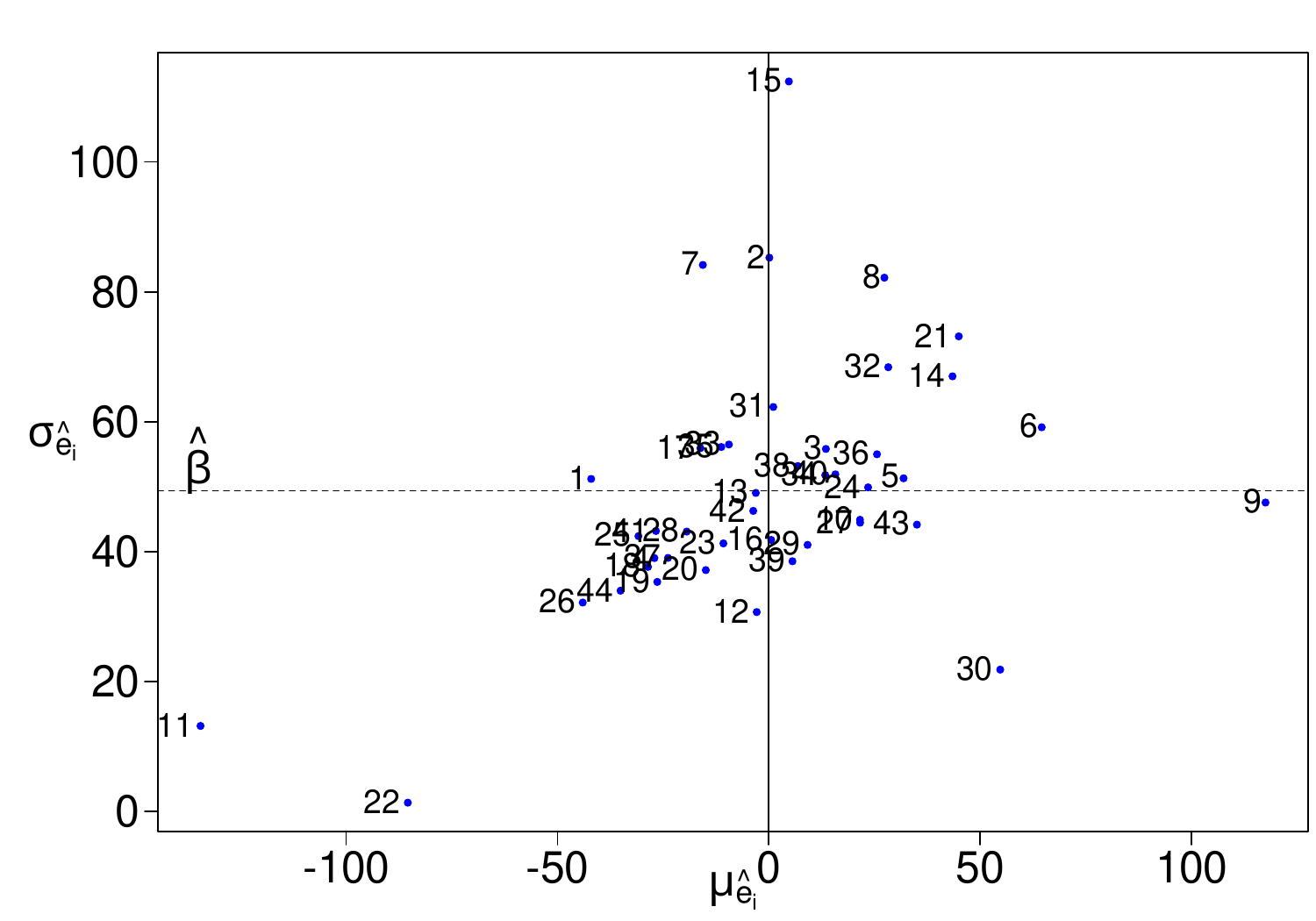}
\caption{Scatterplot of the residual pairs $(\mu_{\widehat{e}_i},\sigma_{\widehat{e}_i})$, ~$i\in \{1,\ldots,44\}$.}
\label{fig:residuals.scatterplot}
\end{figure}

To assess whether $\mu_{E_i}$ and $\sigma_{E_i}$ are independent, as assumed in \eqref{eq:model.reformulation}, the Hellinger test of correlation proposed by \citet{Geenens2022} was employed:
\begin{quote}
\begin{verbatim}
HellCor::HellCor(muresid, sigmaresid, pval.comp = TRUE)$p.value
[1] 0.0508
\end{verbatim}
\end{quote}
This yielded a $p$-value of 0.0508. Consequently, there is insufficient evidence to reject the assumption of independence between $\mu_{E_i}$ and $\sigma_{E_i}$.

Histograms (not shown) of the residual means and residual standard deviations were generated to assess the distributional assumptions. These histograms indicate that the residual means fit the normal distribution reasonably well, while the residual standard deviations fit the exponential distribution to a lesser extent. A slight concentration of mass in the center of the histogram of the residual standard deviations, at the expense of the left tail, was observed. Despite this misalignment with the exponential distribution, assuming an exponential distribution enables tractable likelihood-based inference with explicit formulas. The capacity to derive closed-form confidence intervals and $p$-values outweighs the observed discrepancy in distributional fit.

To identify potential outliers, the $p$-values of the $44$ residual pairs $(\mu_{\widehat{e}_i},\sigma_{\widehat{e}_i})$ were computed using the formula
\[
p\text{-value}(\widehat{e}_i) = \int_{(s,t)\in \R \times (0,\infty) : \widehat{f}_{\widehat{E}_i}(s,t) \leq \widehat{f}_{\widehat{E}_i}(\mu_{\widehat{e}_i}, \sigma_{\widehat{e}_i})} \widehat{f}_{\widehat{E}_i}(s,t) \rd s \rd t,
\]
where $\widehat{f}_{\widehat{E}_i}$ denotes the approximate density of the residual quantile functions from Proposition~\ref{prop:residuals}. The resulting $p$-values are provided in Table~\ref{tab:residuals.p.values}, and their distribution is illustrated in the barplot in Figure~\ref{fig:residuals.barplot}. Two observations, corresponding to patients $\#9$ and $\#11$, yielded $p$-values close to or below $1\%$, indicating the need for further examination, as discussed in Section~\ref{sec:outlier.analysis}.

\begin{table}[!b]
\caption{$p$-values of the residual pairs $(\mu_{\widehat{e}_i},\sigma_{\widehat{e}_i})$, ~$i\in \{1,\ldots,44\}$.}
\label{tab:residuals.p.values}
\centering
\setlength{\tabcolsep}{3pt} % Set global column separation
\renewcommand{\arraystretch}{1} % Reduce vertical spacing
\begin{tabular}{|c|c||c|c||c|c||c|c|}
\hline
\rule{0pt}{2.3ex} % Adjust height for next row
\!Patient & $p$-value & Patient & $p$-value & Patient & $p$-value & Patient & $p$-value \\[0.2mm]
\hline
\rule{0pt}{2ex} % Adjust height for next row
\!\!1 & 0.3605 & 12 & 0.7363 & 23 & 0.6235 & 34 & 0.5269 \\
2 & 0.3256 & 13 & 0.5714 & 24 & 0.4962 & 35 & 0.4995 \\
3 & 0.4957 & 14 & 0.2679 & 25 & 0.5067 & 36 & 0.4473 \\
4 & 0.5610 & 15 & 0.2058 & 26 & 0.4649 & 37 & 0.5801 \\
5 & 0.4272 & 16 & 0.6352 & 27 & 0.5492 & 38 & 0.5325 \\
6 & 0.1684 & 17 & 0.4849 & 28 & 0.5729 & 39 & 0.6606 \\
7 & 0.3119 & 18 & 0.5489 & 29 & 0.6298 & 40 & 0.5167 \\
8 & 0.2819 & 19 & 0.5962 & 30 & 0.4077 & 41 & 0.5280 \\
\cellcolor[gray]{0.9}9 & \cellcolor[gray]{0.9} \hspace{-1.5mm} 0.0075 & 20 & 0.6453 & 31 & 0.4698 & 42 & 0.5939 \\
10 & 0.5459 & 21 & 0.2328 & 32 & 0.3494 & 43 & 0.4579 \\
\cellcolor[gray]{0.9}11 & \cellcolor[gray]{0.9} \hspace{-1.5mm} 0.0068 & 22 & 0.5069 & 33 & 0.5016 & 44 & 0.5383 \\
\hline
\end{tabular}
\end{table}

\begin{figure}[!t]
\centering
\includegraphics[width=1\textwidth]{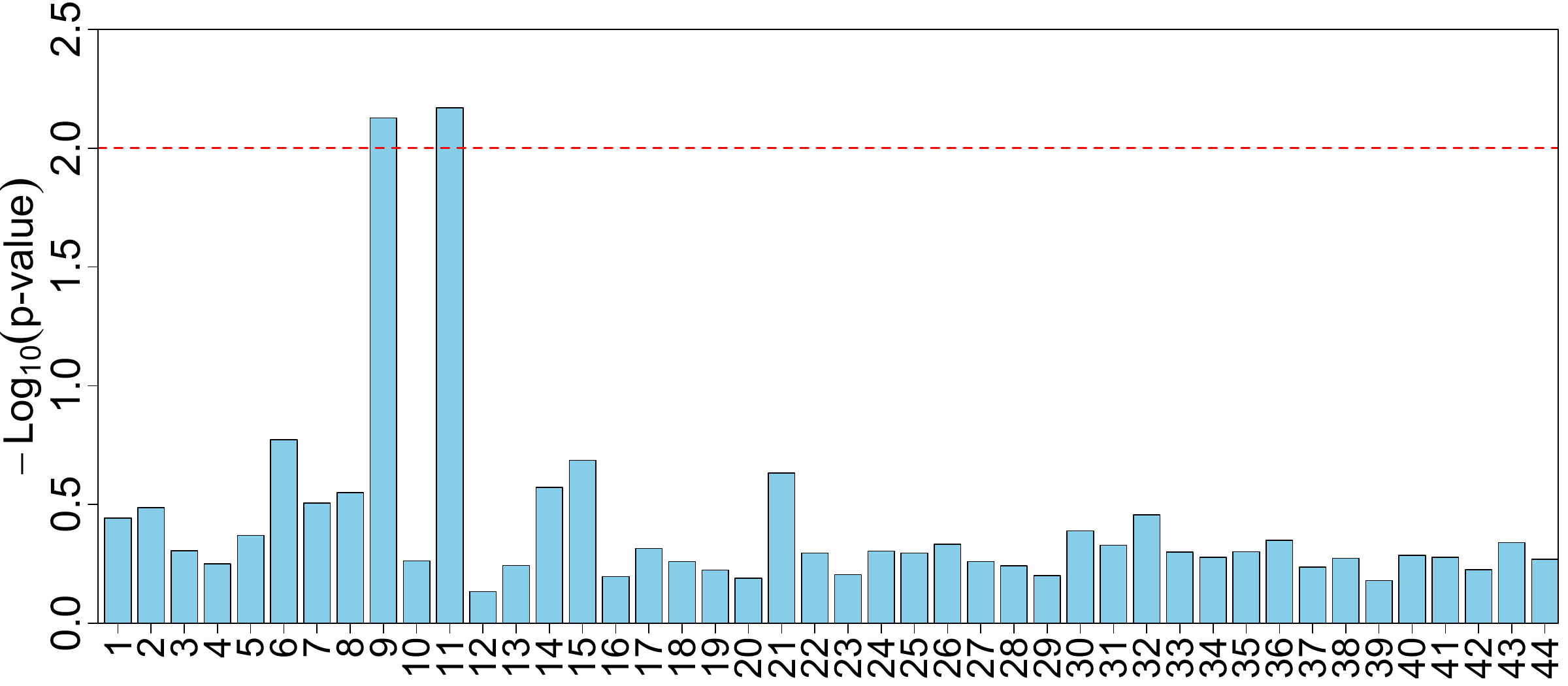}
\caption{Barplot for the $p$-values of the residual pairs $(\mu_{\widehat{e}_i},\sigma_{\widehat{e}_i})$, ~$i\in \{1,\ldots,44\}$. The dashed line denotes a $p$-value of $0.01$, or equivalently, a $-\log_{10}(\text{$p$-value})$ of $2$.}
\label{fig:residuals.barplot}
\end{figure}

\subsection{Interpretation of the results}\label{sec:interpretation.results}

The model parameter estimates and $p$-values (at a significance level of $\alpha = 5\%$) associated with the null hypotheses $\beta_0 = 0$ and $\beta_1 = 0$ are:
\[
\begin{aligned}
\widehat{\beta}_0 = -85.916, \quad \textrm{$p$-value} = 0.242, \quad \textrm{CI}_{0.95}(\beta_0) = [-232.14, 60.31], \\[-2mm]
\widehat{\beta}_1 = 0.884, \quad \textrm{$p$-value} = 3.5 \times 10^{-11}, \quad \textrm{CI}_{0.95}(\beta_1) = [0.683, 1.085].
\end{aligned}
\]
This indicates that the null hypothesis $H_0: \beta_0 = 0$ cannot be rejected at the 5\% significance level, whereas the null hypothesis $H_0: \beta_1 = 0$ is rejected.
Morevover, approximately 36\% of patients exhibited positive effects, as reflected by a positive value of the model estimate $\smash{\mu_{\widehat{Q}_{Y,i}} - \mu_{q_{x,i}} = \widehat{\beta}_0 + (\widehat{\beta}_1 - 1) \mu_{q_{x,i}}}$. Among patients with a more severe initial condition ($\mu_{q_{x,i}} \leq -746.4~\textrm{HU}$), the positive response rate increased to 75\%. Here, the threshold value of $-746.4$ HU is the solution to the affine equation $\smash{\widehat{\beta}_0 + (\widehat{\beta}_1 - 1) \mu_{q_x}} = 0$, highlighting that the greatest improvement in lung function is observed in individuals starting from a more severe baseline condition.

The estimate $\smash{\widehat{\beta}_2}=0.638$ (with $p$-value $< 2 \times 10^{-16}$ and $\textrm{CI}_{0.95}(\beta_2)=[0.614, 0.647]$) strongly supports $\beta_2 < 1$. This implies a reduction in the variability of $Q_{Y,1},\ldots,Q_{Y,n}$ around their mean compared to $q_{x,1},\ldots,q_{x,n}$, consistent with improved air-trapping patterns.

In Figure~\ref{quantile}, the difference between the average quantiles before and after treatment, $\bar{q}_{Y} - \bar{q}_{x}$, is shown for the $n_w=16$ patients with the worst initial conditions, defined by $\mu_{q_{x,i}} < -746.4$. The averages are computed as:
\[
\bar{q}_{x}(t) = \frac{1}{n_w} \sum_{i=1}^{n_w} q_{x,i}(t) \ind_{\{\mu_{q_{x,i}} < -746.4\}}, \qquad \bar{q}_Y(t) = \frac{1}{n_w} \sum_{i=1}^{n_w} \widehat{Q}_{Y,i}(t) \ind_{\{\mu_{q_{x,i}} < -746.4\}}.
\]

\begin{figure}[!b]
\centering
\includegraphics[width=0.70\textwidth]{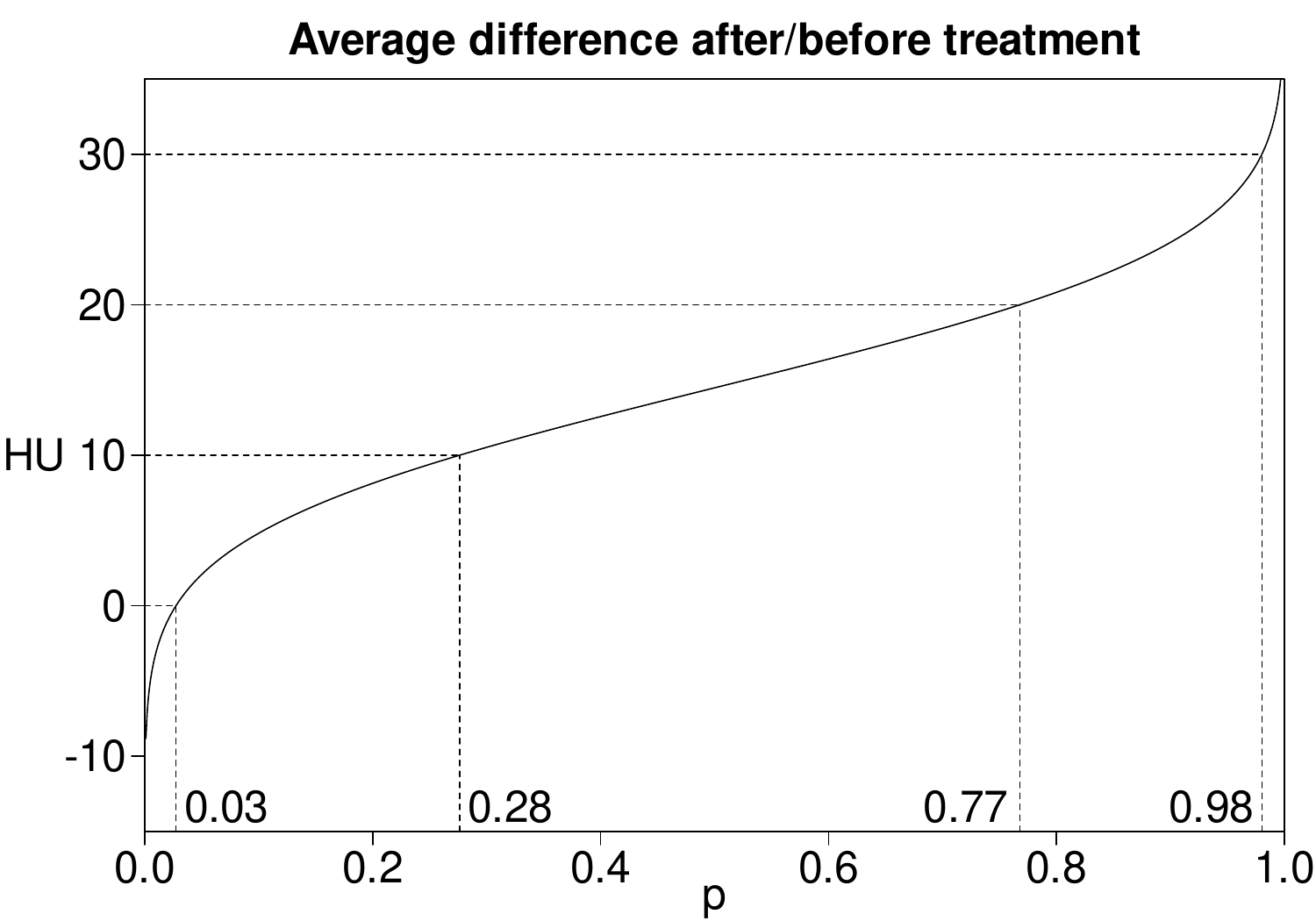}
\caption{Average difference in quantiles, $\bar{q}_{Y} - \bar{q}_{x}$, for the $n_w=16$ patients with the worst initial conditions ($\mu_{q_{x,i}} < -746.4$). The figure illustrates a significant shift in the distribution of lung voxel values, with approximately 25\% of voxels exhibiting an average displacement between $0$ and $10$ HU, 49\% between $10$ and $20$ HU, and 21\% between $20$ and $30$ HU.}\label{quantile}
\end{figure}

Patient \#13, with an initial baseline condition $\mu_{q_{x,13}} = -840.6 < -746.4$, is presented as an illustrative example to highlight the observed trends. Figure~\ref{fig4.4} shows a clear rightward shift in the post-treatment density curve, reflecting reduced air trapping and clinical improvement. Figure~\ref{fig4.5} highlights the corresponding confidence region for the mean response quantile parameters $(\mu_{\EE(Q_{Y,13})}, \sigma_{\EE(Q_{Y,13})})$, showing the rightward shift in the post-treatment mean $\mu_{q_{y,13}}$ and the increased variability in HU values through $\sigma_{q_{y,13}}$.

\begin{figure}[H]
\centering
% Top figure occupying full text width
%\includegraphics[width=0.75\textwidth]{fig4.4-top.pdf}
\includegraphics[width=0.80\textwidth]{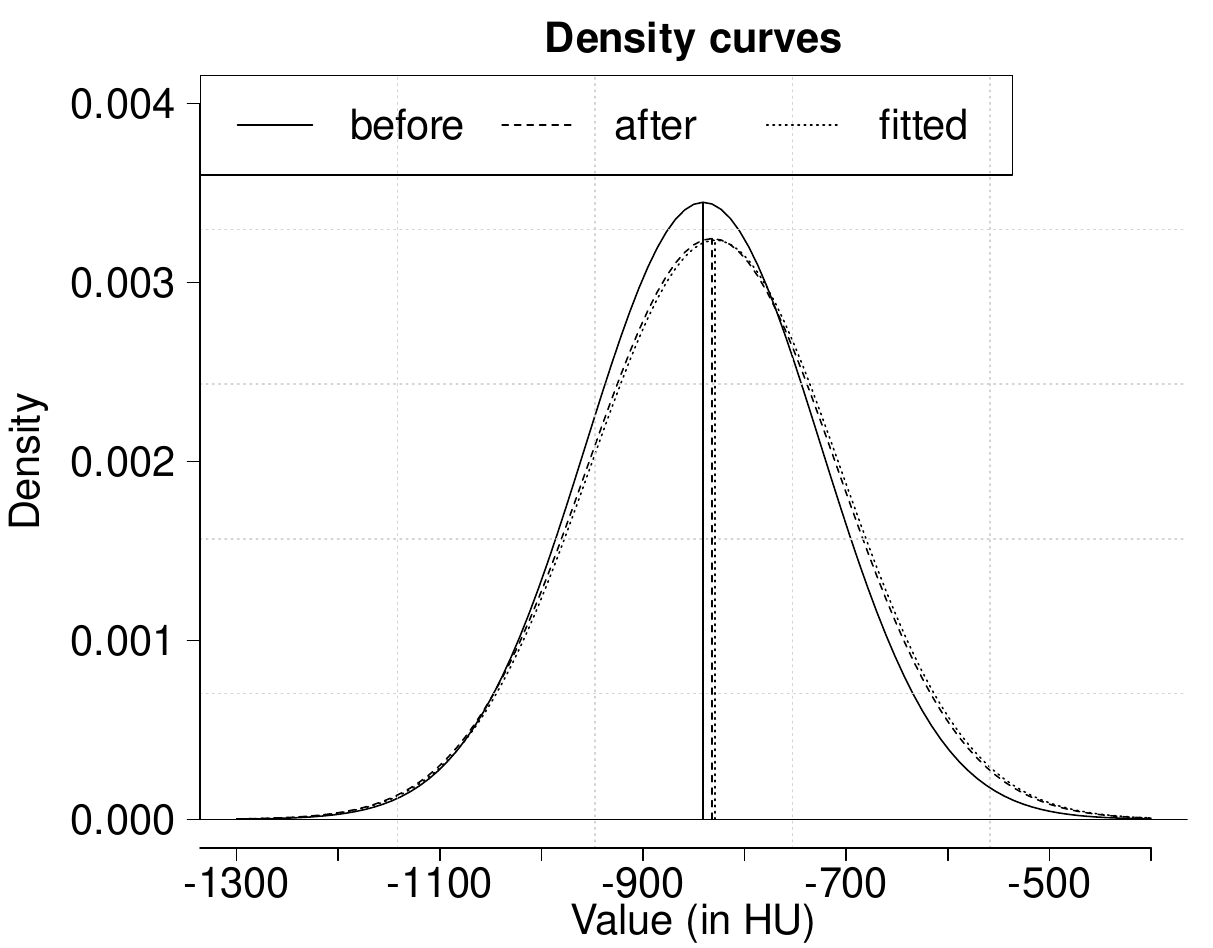}
% Common caption for all figures
\caption{Density curves associated with the fitted, pre- and post-treament quantile curves for patient \#13, who started in the initial condition $\mu_{q_{x,13}} = -840.6 < -746.4$. A clear rightward shift is present.}
\label{fig4.4}
\end{figure}

\vspace{-5mm}
\begin{figure}[H]
\centering
\includegraphics[width=0.80\textwidth]{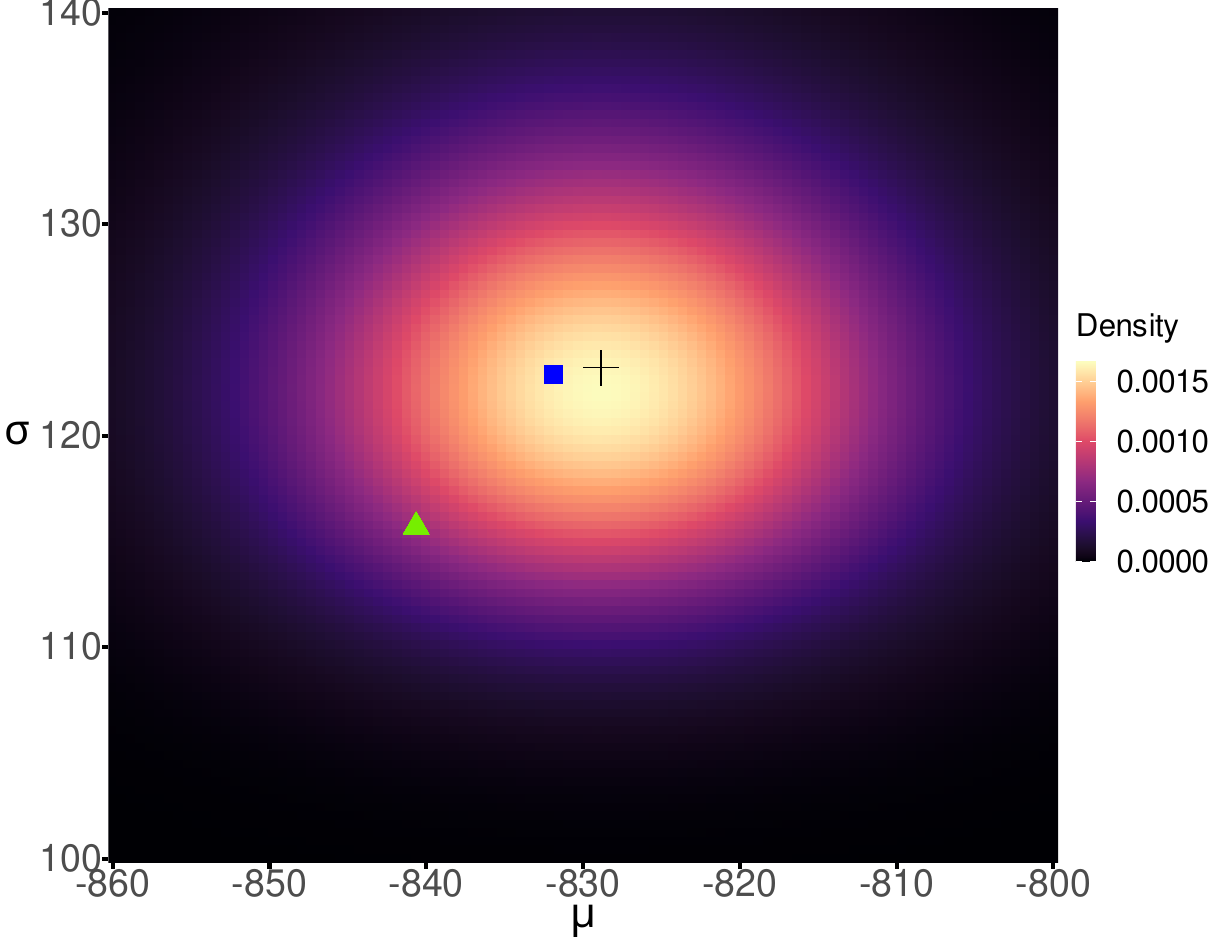}
\caption{Approximate confidence region for the vector $(\mu_{\EE(Q_{Y,13})}, \sigma_{\EE(Q_{Y,13})})$ defining the mean response quantile function $\smash{\EE(Q_{Y,13}) = \mu_{\EE(Q_{Y,13})} + \sigma_{\EE(Q_{Y,13})} \Phi^{-1}}$. The triangle (resp., square) represents the pre- (resp., post-) treatment quantile curve $\smash{\mu_{q_{x,13}}+\sigma_{q_{x,13}}\Phi^{-1}}$ (resp., $\smash{\mu_{q_{y,13}}+\sigma_{q_{y,13}}\Phi^{-1}}$). The cross represents the fitted quantile curve $\smash{\mu_{\widehat{Q}_{Y,13}} + \sigma_{\widehat{Q}_{Y,13}} \Phi^{-1}}$.}
\label{fig4.5}
\end{figure}

\subsection{Analysis of the two outliers}\label{sec:outlier.analysis}

The two outliers, patients \#9 and \#11, warrant closer examination. As seen in Table~\ref{tab:Gaussian.quantile.functions}, these individuals display the most pronounced treatment responses, marked by the largest values of $|\mu_{q_y}-\mu_{q_x}|$, albeit in opposite directions. While patient \#9 %(C1.P32)
showed notable improvement, patient \#11 %(C1.P18)
experienced a decline in condition.

Discussions with the attending medical doctors provided the following insights:
\begin{itemize}
\item For patient \#9, a review of medical records identified a young individual with both a high initial eosinophil count (a type of white blood cell often elevated in allergic reactions and asthma) and a high baseline forced vital capacity (a measure of the maximum amount of air a person can exhale forcefully after taking a deep breath). These characteristics made the patient particularly favorable for a strong positive response to the treatment.
\item For patient \#11, the post-treatment CT scan revealed poor image quality due to insufficient expiration during acquisition, rendering the second scan unrepresentative of the patient's true condition.
\end{itemize}

\section{Discussion}\label{sec:discussion}

The parametric linear regression model for quantile functions introduced herein is simple to understand and provides clear and actionable statistical inference tools such as explicit estimators, $p$-values, and confidence regions for all parameters. It also allows practionners to work within a familiar framework, namely linear regression. The simplicity and interpretability of the model are significant advantages, which broaden its appeal beyond medical applications.

However, the convenience provided by this parametric model comes with certain limitations. Firstly, focusing exclusively on quantile functions results in the loss of spatial information. For instance, in the lung application (Section~\ref{sec:application}), it is assumed that the effectiveness of the treatment is reflected in the distribution of voxel intensities, not their position. This assumption deserves further investigation. To refine and test it, the model could be applied to CT scans segmented by lung lobes. For each lobe, the regression model could be applied independently, and the resulting estimators compared. This approach could capture the heterogeneity across different lung regions and provide a more granular understanding of treatment response. Secondly, although the model provides easy-to-interpret estimates, the distributional assumption on the quantile errors could be made more flexible to enhance the fit of $\sigma_{Q_{Y,i}} \sim \mathcal{E}\mathrm{xp}(\beta, \beta_2\sigma_{q_{x,i}})$, as derived in Proposition~\ref{prop:law.mu.sigma}, to the post-treatment data. The exponential distribution was chosen to enable the derivation of explicit estimators and confidence regions, albeit at the expense of some accuracy. An in-depth study is needed to determine whether a wider family of distributions can be imposed on $\sigma_{Q_{Y,i}}$.

A promising direction for future research is extending the method to a multivariate setting. In the context of the lung application, the linear regression model could be adapted to analyze bivariate quantile functions \citep{MR1934981}, with one component for expiration and one for inspiration. These bivariate quantile functions would be derived from CT scans by aligning expiration and inspiration images using B-spline registration \cite{rohr,sengupta_survey_2022}. The registration of CT scans during both inspiration and expiration is already done in parametric response map (PRM) studies; see, e.g., \cite{galban_computed_2012,POMPE201748}. This new model would leverage information from both phases, enhancing the analysis beyond current approaches. Although the formalism of the regression model seems to permit such an extension, deriving the corresponding estimators and confidence regions may be more challenging. Overcoming these challenges could significantly advance the methodology.

\section*{Data availability statement}\label{sec:data.availability}

The computing codes and datasets that generated the figures throughout the paper and the lung application results in Section~\ref{sec:application} are available online in the Github repository of \citet{BeclinMicheauxMolinariOuimet_github_2024}.

\begin{acks}[Acknowledgments]
Grateful acknowledgment is extended to Dr. Arnaud Bourdin (MD, PhD; Pulmonary Medicine; Montpellier, Occitanie, France) for invaluable assistance in the outlier analysis (Section~\ref{sec:outlier.analysis}) and for providing key medical insights. %The authors also extend their appreciation to the anonymous reviewers for their thoughtful feedback, which has significantly enhanced the clarity and depth of this paper.
\end{acks}

\begin{funding}
F.\ Ouimet's was previously funded through the Canada Research Chairs Program (Grant 950-231937 to Christian Genest) and the Natural Sciences and Engineering Research Council of Canada (Grant RGPIN-2024-04088 to Christian Genest). F.\ Ouimet’s current postdoctoral fellowship is funded through the Natural Science and Engineering Research Council of Canada (Grant RGPIN-2024-05794 to Anne MacKay).
\end{funding}

%% References

\bibliographystyle{imsart-nameyear} % bst file
\bibliography{main_bib.bib} % main bib file

\begin{thebibliography}{50}
% BibTex style file: imsart-nameyear.bst, 2017-11-03
% Default style options (sort=1,type=nameyear).
% Used options (sort=1,type=nameyear).

\bibitem[\protect\citeauthoryear{B\'eclin}{2024}]{Beclin2024}
\begin{bphdthesis}[author]
\bauthor{\bsnm{B\'eclin},~\bfnm{M.~F.}\binits{M.~F.}}
(\byear{2024}).
\btitle{\'Elaboration de mod\`eles intelligents \`a partir des donn\'ees
  d'imagerie scanner de patients sous traitement au Benralizumab},
\btype{PhD thesis},
\bpublisher{Universit\'e de Montpellier},
\baddress{Montpellier, France}.
\end{bphdthesis}
\endbibitem

\bibitem[\protect\citeauthoryear{B\'eclin
  et~al.}{2024}]{BeclinMicheauxMolinariOuimet_github_2024}
\begin{bmisc}[author]
\bauthor{\bsnm{B\'eclin},~\bfnm{M.~F.}\binits{M.~F.}},
  \bauthor{\bparticle{Lafaye~de} \bsnm{Micheaux},~\bfnm{P.}\binits{P.}},
  \bauthor{\bsnm{Molinari},~\bfnm{N.}\binits{N.}} \AND
  \bauthor{\bsnm{Ouimet},~\bfnm{F.}\binits{F.}}
(\byear{2024}).
\btitle{Quantile{F}unction{R}egression}.
\bnote{Available online at
  \href{https://github.com/FredericOuimetMcGill/QuantileFunctionRegression}{https://github.com/FredericOuimetMcGill/QuantileFunctionRegression}}.
\end{bmisc}
\endbibitem

\bibitem[\protect\citeauthoryear{Beyaztas, Shang and Alin}{2022}]{Beyaztas2022}
\begin{barticle}[author]
\bauthor{\bsnm{Beyaztas},~\bfnm{U.}\binits{U.}},
  \bauthor{\bsnm{Shang},~\bfnm{H.~L.}\binits{H.~L.}} \AND
  \bauthor{\bsnm{Alin},~\bfnm{A.}\binits{A.}}
(\byear{2022}).
\btitle{Function-on-function partial quantile regression}.
\bjournal{Journal of Agricultural, Biological and Environmental Statistics}
\bvolume{27}
\bpages{149--174}.
\bdoi{10.1007/s13253-021-00477-9}
\end{barticle}
\endbibitem

\bibitem[\protect\citeauthoryear{Beyaztas, Shang and
  Saricam}{2024}]{BeyaztasShangSaricam2024}
\begin{barticle}[author]
\bauthor{\bsnm{Beyaztas},~\bfnm{U.}\binits{U.}},
  \bauthor{\bsnm{Shang},~\bfnm{H.~L.}\binits{H.~L.}} \AND
  \bauthor{\bsnm{Saricam},~\bfnm{S.}\binits{S.}}
(\byear{2024}).
\btitle{Penalized function-on-function linear quantile regression}.
\bjournal{Computational Statistics}
\bpages{29 pp.}
\bdoi{10.1007/s00180-024-01494-1}
\end{barticle}
\endbibitem

\bibitem[\protect\citeauthoryear{Beyaztas, Tez and Shang}{2024}]{Beyaztas2024}
\begin{barticle}[author]
\bauthor{\bsnm{Beyaztas},~\bfnm{U.}\binits{U.}},
  \bauthor{\bsnm{Tez},~\bfnm{M.}\binits{M.}} \AND
  \bauthor{\bsnm{Shang},~\bfnm{H.~L.}\binits{H.~L.}}
(\byear{2024}).
\btitle{Robust scalar-on-function partial quantile regression}.
\bjournal{Journal of Applied Statistics}
\bvolume{51}
\bpages{1359--1377}.
\bdoi{10.1080/02664763.2023.2202464}
\bmrnumber{4753760}
\end{barticle}
\endbibitem

\bibitem[\protect\citeauthoryear{Blanke and Bosq}{2021}]{Blanke2021_supp}
\begin{bincollection}[author]
\bauthor{\bsnm{Blanke},~\bfnm{D.}\binits{D.}} \AND
  \bauthor{\bsnm{Bosq},~\bfnm{D.}\binits{D.}}
(\byear{2021}).
\btitle{Piecewise linear continuous estimators of the quantile function}.
In \bbooktitle{Advances in contemporary statistics and
  econometrics---{F}estschrift in honor of {C}hristine {T}homas-{A}gnan}
(\beditor{\bfnm{A.}\binits{A.}~\bsnm{Daouia}} \AND
  \beditor{\bfnm{A.}\binits{A.}~\bsnm{Ruiz-Gazen}}, eds.)
\bchapter{9},
\bpages{161--175}.
\bpublisher{Springer, Cham}.
\bdoi{10.1007/978-3-030-73249-3_9}
\bmrnumber{4299279}
\end{bincollection}
\endbibitem

\bibitem[\protect\citeauthoryear{Bonneel, Peyr\'e and
  Cuturi}{2016}]{bonneel_Wasserstein_2016}
\begin{barticle}[author]
\bauthor{\bsnm{Bonneel},~\bfnm{N.}\binits{N.}},
  \bauthor{\bsnm{Peyr\'e},~\bfnm{G.}\binits{G.}} \AND
  \bauthor{\bsnm{Cuturi},~\bfnm{M.}\binits{M.}}
(\byear{2016}).
\btitle{{W}asserstein barycentric coordinates: histogram regression using
  optimal transport}.
\bjournal{ACM Transactions on Graphics}
\bpages{Article No. 71, 1--10.}
\bdoi{10.1145/2897824.2925918}
\end{barticle}
\endbibitem

\bibitem[\protect\citeauthoryear{Cabon et~al.}{2019}]{cabon2019k}
\begin{barticle}[author]
\bauthor{\bsnm{Cabon},~\bfnm{Y.}\binits{Y.}},
  \bauthor{\bsnm{Suehs},~\bfnm{C.}\binits{C.}},
  \bauthor{\bsnm{Bommart},~\bfnm{S.}\binits{S.}},
  \bauthor{\bsnm{Vachier},~\bfnm{I.}\binits{I.}},
  \bauthor{\bsnm{Marin},~\bfnm{G.}\binits{G.}},
  \bauthor{\bsnm{Bourdin},~\bfnm{A.}\binits{A.}} \AND
  \bauthor{\bsnm{Molinari},~\bfnm{N.}\binits{N.}}
(\byear{2019}).
\btitle{k-nearest neighbor curves in imaging data classification}.
\bjournal{Frontiers in Applied Mathematics and Statistics}
\bvolume{5}
\bpages{11 pp.}
\bdoi{10.3389/fams.2019.00022}
\end{barticle}
\endbibitem

\bibitem[\protect\citeauthoryear{Cardot, Crambes and Sarda}{2005}]{Cardot2005}
\begin{barticle}[author]
\bauthor{\bsnm{Cardot},~\bfnm{H.}\binits{H.}},
  \bauthor{\bsnm{Crambes},~\bfnm{C.}\binits{C.}} \AND
  \bauthor{\bsnm{Sarda},~\bfnm{P.}\binits{P.}}
(\byear{2005}).
\btitle{Quantile regression when the covariates are functions}.
\bjournal{Journal of Nonparametric Statistics}
\bvolume{17}
\bpages{841--856}.
\bdoi{10.1080/10485250500303015}
\bmrnumber{2180369}
\end{barticle}
\endbibitem

\bibitem[\protect\citeauthoryear{Chen, Lin and
  M\"{u}ller}{2023}]{chen2023wasserstein}
\begin{barticle}[author]
\bauthor{\bsnm{Chen},~\bfnm{Y.}\binits{Y.}},
  \bauthor{\bsnm{Lin},~\bfnm{Z.}\binits{Z.}} \AND
  \bauthor{\bsnm{M\"{u}ller},~\bfnm{H.~G.}\binits{H.~G.}}
(\byear{2023}).
\btitle{{W}asserstein regression}.
\bjournal{Journal of the American Statistical Association}
\bvolume{118}
\bpages{869--882}.
\bdoi{10.1080/01621459.2021.1956937}
\bmrnumber{4595462}
\end{barticle}
\endbibitem

\bibitem[\protect\citeauthoryear{Chen and M\"uller}{2012}]{Chen2012}
\begin{barticle}[author]
\bauthor{\bsnm{Chen},~\bfnm{K.}\binits{K.}} \AND
  \bauthor{\bsnm{M\"uller},~\bfnm{H.~G.}\binits{H.~G.}}
(\byear{2012}).
\btitle{Conditional quantile analysis when covariates are functions, with
  application to growth data}.
\bjournal{Journal of the Royal Statistical Society: Series B (Statistical
  Methodology)}
\bvolume{74}
\bpages{67--89}.
\bdoi{10.1111/j.1467-9868.2011.01008.x}
\bmrnumber{2885840}
\end{barticle}
\endbibitem

\bibitem[\protect\citeauthoryear{Chen and Welsh}{2002}]{MR1934981}
\begin{barticle}[author]
\bauthor{\bsnm{Chen},~\bfnm{L.~A.}\binits{L.~A.}} \AND
  \bauthor{\bsnm{Welsh},~\bfnm{A.~H.}\binits{A.~H.}}
(\byear{2002}).
\btitle{Distribution-function-based bivariate quantiles}.
\bjournal{Journal of Multivariate Analysis}
\bvolume{83}
\bpages{208--231}.
\bdoi{10.1006/jmva.2001.2043}
\bmrnumber{1934981}
\end{barticle}
\endbibitem

\bibitem[\protect\citeauthoryear{Cleveland}{2004}]{Cleveland2004}
\begin{bbook}[author]
\bauthor{\bsnm{Cleveland},~\bfnm{W.~S.}\binits{W.~S.}}
(\byear{2004}).
\btitle{The {E}lements of {G}raphing {D}ata},
\bedition{2} ed.
\bpublisher{Hobart Press}, \baddress{Summit, NJ}.
\end{bbook}
\endbibitem

\bibitem[\protect\citeauthoryear{Dias and Brito}{2015}]{dias2015linear}
\begin{barticle}[author]
\bauthor{\bsnm{Dias},~\bfnm{S.}\binits{S.}} \AND
  \bauthor{\bsnm{Brito},~\bfnm{P.}\binits{P.}}
(\byear{2015}).
\btitle{Linear regression model with histogram-valued variables}.
\bjournal{Statistical Analysis and Data Mining}
\bvolume{8}
\bpages{75--113}.
\bdoi{10.1002/sam.11260}
\bmrnumber{3342982}
\end{barticle}
\endbibitem

\bibitem[\protect\citeauthoryear{Galb\'an et~al.}{2012}]{galban_computed_2012}
\begin{barticle}[author]
\bauthor{\bsnm{Galb\'an},~\bfnm{C.~J.}\binits{C.~J.}},
  \bauthor{\bsnm{Han},~\bfnm{M.~K.}\binits{M.~K.}},
  \bauthor{\bsnm{Boes},~\bfnm{J.~L.}\binits{J.~L.}},
  \bauthor{\bsnm{Chughtai},~\bfnm{K.~A.}\binits{K.~A.}},
  \bauthor{\bsnm{Meyer},~\bfnm{C.~R.}\binits{C.~R.}},
  \bauthor{\bsnm{Johnson},~\bfnm{T.~D.}\binits{T.~D.}},
  \bauthor{\bsnm{Galb\'an},~\bfnm{S.}\binits{S.}},
  \bauthor{\bsnm{Rehemtulla},~\bfnm{A.}\binits{A.}},
  \bauthor{\bsnm{Kazerooni},~\bfnm{E.~A.}\binits{E.~A.}},
  \bauthor{\bsnm{Martinez},~\bfnm{F.~J.}\binits{F.~J.}} \AND
  \bauthor{\bsnm{Ross},~\bfnm{B.~D.}\binits{B.~D.}}
(\byear{2012}).
\btitle{Computed tomography-based biomarker provides unique signature for
  diagnosis of {COPD} phenotypes and disease progression}.
\bjournal{Nature Medicine}
\bvolume{18}
\bpages{1711--1715}.
\bdoi{10.1038/nm.2971}
\end{barticle}
\endbibitem

\bibitem[\protect\citeauthoryear{Geenens and Lafaye~de
  Micheaux}{2022}]{Geenens2022}
\begin{barticle}[author]
\bauthor{\bsnm{Geenens},~\bfnm{G.}\binits{G.}} \AND
  \bauthor{\bparticle{Lafaye~de} \bsnm{Micheaux},~\bfnm{P.}\binits{P.}}
(\byear{2022}).
\btitle{The {H}ellinger correlation}.
\bjournal{Journal of the American Statistical Association}
\bvolume{117}
\bpages{639--653}.
\bdoi{10.1080/01621459.2020.1791132}
\bmrnumber{4436302}
\end{barticle}
\endbibitem

\bibitem[\protect\citeauthoryear{Genofre et~al.}{2023}]{genofre_effects_2023}
\begin{barticle}[author]
\bauthor{\bsnm{Genofre},~\bfnm{E.}\binits{E.}},
  \bauthor{\bsnm{Carstens},~\bfnm{D.}\binits{D.}},
  \bauthor{\bsnm{DeBacker},~\bfnm{W.}\binits{W.}},
  \bauthor{\bsnm{Muchmore},~\bfnm{P.}\binits{P.}},
  \bauthor{\bsnm{Panettieri~Jr},~\bfnm{R.~A.}\binits{R.~A.}},
  \bauthor{\bsnm{Rhodes},~\bfnm{K.}\binits{K.}},
  \bauthor{\bsnm{Shih},~\bfnm{V.~H.}\binits{V.~H.}} \AND
  \bauthor{\bsnm{Trudo},~\bfnm{F.}\binits{F.}}
(\byear{2023}).
\btitle{The effects of benralizumab on airway geometry and dynamics in severe
  eosinophilic asthma: a single-arm study design exploring a functional
  respiratory imaging approach}.
\bjournal{Respiratory Research}
\bvolume{24}
\bpages{121, 13 pp.}
\bdoi{10.1186/s12931-023-02415-4}
\end{barticle}
\endbibitem

\bibitem[\protect\citeauthoryear{Ghodrati and
  Panaretos}{2022}]{ghodrati2022distribution}
\begin{barticle}[author]
\bauthor{\bsnm{Ghodrati},~\bfnm{L.}\binits{L.}} \AND
  \bauthor{\bsnm{Panaretos},~\bfnm{V.~M.}\binits{V.~M.}}
(\byear{2022}).
\btitle{Distribution-on-distribution regression via optimal transport maps}.
\bjournal{Biometrika}
\bvolume{109}
\bpages{957--974}.
\bdoi{10.1093/biomet/asac005}
\bmrnumber{4519110}
\end{barticle}
\endbibitem

\bibitem[\protect\citeauthoryear{Ghodrati and
  Panaretos}{2024}]{ghodrati2024transportation}
\begin{barticle}[author]
\bauthor{\bsnm{Ghodrati},~\bfnm{L.}\binits{L.}} \AND
  \bauthor{\bsnm{Panaretos},~\bfnm{V.~M.}\binits{V.~M.}}
(\byear{2024}).
\btitle{Transportation of measure regression in higher dimensions}.
\bjournal{Preprint}
\bpages{28 pp.}
\bdoi{10.48550/arXiv.2305.17503}
\end{barticle}
\endbibitem

\bibitem[\protect\citeauthoryear{Ghosal et~al.}{2023}]{Ghosal2023}
\begin{barticle}[author]
\bauthor{\bsnm{Ghosal},~\bfnm{R.}\binits{R.}},
  \bauthor{\bsnm{Varma},~\bfnm{V.~R.}\binits{V.~R.}},
  \bauthor{\bsnm{Volfson},~\bfnm{D.}\binits{D.}},
  \bauthor{\bsnm{Hillel},~\bfnm{I.}\binits{I.}},
  \bauthor{\bsnm{Urbanek},~\bfnm{J.}\binits{J.}},
  \bauthor{\bsnm{Hausdorff},~\bfnm{J.~M.}\binits{J.~M.}},
  \bauthor{\bsnm{Watts},~\bfnm{A.}\binits{A.}} \AND
  \bauthor{\bsnm{Zipunnikov},~\bfnm{V.}\binits{V.}}
(\byear{2023}).
\btitle{Distributional data analysis via quantile functions and its application
  to modeling digital biomarkers of gait in {Alzheimer}'s disease}.
\bjournal{Biostatistics}
\bvolume{24}
\bpages{539--561}.
\bdoi{10.1093/biostatistics/kxab041}
\bmrnumber{4615240}
\end{barticle}
\endbibitem

\bibitem[\protect\citeauthoryear{Hartley et~al.}{2016}]{hartley_et_al_2016}
\begin{barticle}[author]
\bauthor{\bsnm{Hartley},~\bfnm{R.~A.}\binits{R.~A.}},
  \bauthor{\bsnm{Barker},~\bfnm{B.~L.}\binits{B.~L.}},
  \bauthor{\bsnm{Newby},~\bfnm{C.}\binits{C.}},
  \bauthor{\bsnm{Pakkal},~\bfnm{M.}\binits{M.}},
  \bauthor{\bsnm{Baldi},~\bfnm{S.}\binits{S.}},
  \bauthor{\bsnm{Kajekar},~\bfnm{R.}\binits{R.}},
  \bauthor{\bsnm{Kay},~\bfnm{R.}\binits{R.}},
  \bauthor{\bsnm{Laurencin},~\bfnm{M.}\binits{M.}},
  \bauthor{\bsnm{Marshall},~\bfnm{R.~P.}\binits{R.~P.}},
  \bauthor{\bsnm{Sousa},~\bfnm{A.~R.}\binits{A.~R.}},
  \bauthor{\bsnm{Parmar},~\bfnm{H.}\binits{H.}},
  \bauthor{\bsnm{Siddiqui},~\bfnm{S.}\binits{S.}},
  \bauthor{\bsnm{Gupta},~\bfnm{S.}\binits{S.}} \AND
  \bauthor{\bsnm{Brightling},~\bfnm{C.~E.}\binits{C.~E.}}
(\byear{2016}).
\btitle{Relationship between lung function and quantitative computed
  tomographic parameters of airway remodeling, air trapping, and emphysema in
  patients with asthma and chronic obstructive pulmonary disease: a
  single-center study}.
\bjournal{Journal of Allergy and Clinical Immunology}
\bvolume{137}
\bpages{1413--1422}.
\bdoi{10.1016/j.jaci.2016.02.001}
\end{barticle}
\endbibitem

\bibitem[\protect\citeauthoryear{Heuberger, Geissb{\"u}hler and
  M{\"u}ller}{2005}]{heuberg_lung_2005}
\begin{binproceedings}[author]
\bauthor{\bsnm{Heuberger},~\bfnm{J.}\binits{J.}},
  \bauthor{\bsnm{Geissb{\"u}hler},~\bfnm{A.}\binits{A.}} \AND
  \bauthor{\bsnm{M{\"u}ller},~\bfnm{H.}\binits{H.}}
(\byear{2005}).
\btitle{Lung CT segmentation for image retrieval using the Insight Toolkit
  (ITK)}.
In \bbooktitle{Medical Imaging and Telemedicine (MIT)}
\bpages{57--62}.
\end{binproceedings}
\endbibitem

\bibitem[\protect\citeauthoryear{Irpino}{2024}]{Irpino2024}
\begin{bmisc}[author]
\bauthor{\bsnm{Irpino},~\bfnm{A.}\binits{A.}}
(\byear{2024}).
\btitle{Hist{DAW}ass: {H}istogram-{V}alued {D}ata {A}nalysis}.
\bnote{\texttt{R} package version 1.0.8}.
\bdoi{10.32614/CRAN.package.HistDAWass}
\end{bmisc}
\endbibitem

\bibitem[\protect\citeauthoryear{Irpino and Verde}{2015}]{IrpinoVerde2015}
\begin{barticle}[author]
\bauthor{\bsnm{Irpino},~\bfnm{A.}\binits{A.}} \AND
  \bauthor{\bsnm{Verde},~\bfnm{R.}\binits{R.}}
(\byear{2015}).
\btitle{Linear regression for numeric symbolic variables: a least squares
  approach based on {W}asserstein distance}.
\bjournal{Advances in Data Analysis and Classification}
\bvolume{9}
\bpages{81--106}.
\bdoi{10.1007/s11634-015-0197-7}
\bmrnumber{3317894}
\end{barticle}
\endbibitem

\bibitem[\protect\citeauthoryear{Kato}{2012}]{Kato2012}
\begin{barticle}[author]
\bauthor{\bsnm{Kato},~\bfnm{K.}\binits{K.}}
(\byear{2012}).
\btitle{Estimation in functional linear quantile regression}.
\bjournal{The Annals of Statistics}
\bvolume{40}
\bpages{3108--3136}.
\bdoi{10.1214/12-AOS1066}
\bmrnumber{3097971}
\end{barticle}
\endbibitem

\bibitem[\protect\citeauthoryear{Koenker}{2005}]{Koenker2005}
\begin{bbook}[author]
\bauthor{\bsnm{Koenker},~\bfnm{R.}\binits{R.}}
(\byear{2005}).
\btitle{Quantile {R}egression}.
\bseries{Econometric Society Monographs}
\bvolume{38}.
\bpublisher{Cambridge University Press, Cambridge}.
\bdoi{10.1017/CBO9780511754098}
\bmrnumber{2268657}
\end{bbook}
\endbibitem

\bibitem[\protect\citeauthoryear{Li et~al.}{2022}]{LiWangMaityStaicu2022}
\begin{barticle}[author]
\bauthor{\bsnm{Li},~\bfnm{M.}\binits{M.}},
  \bauthor{\bsnm{Wang},~\bfnm{K.}\binits{K.}},
  \bauthor{\bsnm{Maity},~\bfnm{A.}\binits{A.}} \AND
  \bauthor{\bsnm{Staicu},~\bfnm{A.~M.}\binits{A.~M.}}
(\byear{2022}).
\btitle{Inference in functional linear quantile regression}.
\bjournal{Journal of Multivariate Analysis}
\bvolume{190}
\bpages{Paper No. 104985, 19 pp.}
\bdoi{10.1016/j.jmva.2022.104985}
\bmrnumber{4402003}
\end{barticle}
\endbibitem

\bibitem[\protect\citeauthoryear{Liu, Li and Morris}{2020}]{Liu2020}
\begin{barticle}[author]
\bauthor{\bsnm{Liu},~\bfnm{Y.}\binits{Y.}},
  \bauthor{\bsnm{Li},~\bfnm{M.}\binits{M.}} \AND
  \bauthor{\bsnm{Morris},~\bfnm{J.~S.}\binits{J.~S.}}
(\byear{2020}).
\btitle{Function-on-scalar quantile regression with application to mass
  spectrometry proteomics data}.
\bjournal{The Annals of Applied Statistics}
\bvolume{14}
\bpages{521--541}.
\bdoi{10.1214/19-AOAS1319}
\bmrnumber{4117818}
\end{barticle}
\endbibitem

\bibitem[\protect\citeauthoryear{Liu, Li and Morris}{2023}]{Liu2023}
\begin{barticle}[author]
\bauthor{\bsnm{Liu},~\bfnm{Y.}\binits{Y.}},
  \bauthor{\bsnm{Li},~\bfnm{M.}\binits{M.}} \AND
  \bauthor{\bsnm{Morris},~\bfnm{J.~S.}\binits{J.~S.}}
(\byear{2023}).
\btitle{Scalable function-on-scalar quantile regression for densely sampled
  functional data}.
\bjournal{Preprint}.
\bdoi{10.48550/arXiv.2002.03355}
\end{barticle}
\endbibitem

\bibitem[\protect\citeauthoryear{Lorentz}{1986}]{Lorentz1986}
\begin{bbook}[author]
\bauthor{\bsnm{Lorentz},~\bfnm{G.~G.}\binits{G.~G.}}
(\byear{1986}).
\btitle{Bernstein {P}olynomials},
\bedition{Second} ed.
\bpublisher{Chelsea Publishing Co., New York}.
\bmrnumber{864976}
\end{bbook}
\endbibitem

\bibitem[\protect\citeauthoryear{Marin et~al.}{2016}]{marin2016fractal}
\begin{bincollection}[author]
\bauthor{\bsnm{Marin},~\bfnm{G.}\binits{G.}},
  \bauthor{\bsnm{Bommart},~\bfnm{S.}\binits{S.}},
  \bauthor{\bsnm{Molinari},~\bfnm{N.}\binits{N.}},
  \bauthor{\bsnm{Knabe},~\bfnm{L.}\binits{L.}},
  \bauthor{\bsnm{Cabon},~\bfnm{Y.}\binits{Y.}},
  \bauthor{\bsnm{Vachier},~\bfnm{I.}\binits{I.}},
  \bauthor{\bsnm{Devautour},~\bfnm{C.}\binits{C.}} \AND
  \bauthor{\bsnm{Bourdin},~\bfnm{A.}\binits{A.}}
(\byear{2016}).
\btitle{Fractal dimension as a global index of air trapping}.
In \bbooktitle{A38. Diagnostic Markers of Asthma and COPD}
\bpages{A1448--A1448}.
\bpublisher{American Thoracic Society}.
\end{bincollection}
\endbibitem

\bibitem[\protect\citeauthoryear{Mitsunobu and
  Tanizaki}{2005}]{mitsunobu_use_2005}
\begin{barticle}[author]
\bauthor{\bsnm{Mitsunobu},~\bfnm{F.}\binits{F.}} \AND
  \bauthor{\bsnm{Tanizaki},~\bfnm{Y.}\binits{Y.}}
(\byear{2005}).
\btitle{The use of computed tomography to assess asthma severity}.
\bjournal{Current Opinion in Allergy and Clinical Immunology}
\bvolume{5}
\bpages{85--90}.
\bdoi{10.1097/00130832-200502000-00015}
\end{barticle}
\endbibitem

\bibitem[\protect\citeauthoryear{Mutis et~al.}{2024}]{Mutis2024}
\begin{barticle}[author]
\bauthor{\bsnm{Mutis},~\bfnm{M.}\binits{M.}},
  \bauthor{\bsnm{Beyaztas},~\bfnm{U.}\binits{U.}},
  \bauthor{\bsnm{Karaman},~\bfnm{F.}\binits{F.}} \AND
  \bauthor{\bsnm{Shang},~\bfnm{H.~L.}\binits{H.~L.}}
(\byear{2024}).
\btitle{On function-on-function linear quantile regression}.
\bjournal{Journal of Applied Statistics}.
\bdoi{10.1080/02664763.2024.2395960}
\end{barticle}
\endbibitem

\bibitem[\protect\citeauthoryear{Okano and Imaizumi}{2024}]{OkanoImaizumi2024}
\begin{barticle}[author]
\bauthor{\bsnm{Okano},~\bfnm{R.}\binits{R.}} \AND
  \bauthor{\bsnm{Imaizumi},~\bfnm{M.}\binits{M.}}
(\byear{2024}).
\btitle{Distribution-on-distribution regression with {W}asserstein metric:
  multivariate {G}aussian case}.
\bjournal{Journal of Multivariate Analysis}
\bvolume{203}
\bpages{Paper No. 105334, 20 pp.}
\bdoi{10.1016/j.jmva.2024.105334}
\bmrnumber{4752834}
\end{barticle}
\endbibitem

\bibitem[\protect\citeauthoryear{Panaretos and Zemel}{2019}]{Panaretos2019}
\begin{barticle}[author]
\bauthor{\bsnm{Panaretos},~\bfnm{V.~M.}\binits{V.~M.}} \AND
  \bauthor{\bsnm{Zemel},~\bfnm{Y.}\binits{Y.}}
(\byear{2019}).
\btitle{Statistical aspects of {W}asserstein distances}.
\bjournal{Annual Review of Statistics and Its Application}
\bvolume{6}
\bpages{405--431}.
\bdoi{10.1146/annurev-statistics-030718-104938}
\bmrnumber{3939527}
\end{barticle}
\endbibitem

\bibitem[\protect\citeauthoryear{Parzen}{1983}]{Parzen1983}
\begin{btechreport}[author]
\bauthor{\bsnm{Parzen},~\bfnm{E.}\binits{E.}}
(\byear{1983}).
\btitle{Informative Quantile Functions and Identification of Probability
  Distribution Types}
\btype{Technical Report} No. \bnumber{A-26},
\bpublisher{Texas A\&M University}.
\end{btechreport}
\endbibitem

\bibitem[\protect\citeauthoryear{Pompe et~al.}{2017}]{POMPE201748}
\begin{barticle}[author]
\bauthor{\bsnm{Pompe},~\bfnm{E.}\binits{E.}},
  \bauthor{\bsnm{Galb\'an},~\bfnm{C.~J.}\binits{C.~J.}},
  \bauthor{\bsnm{Ross},~\bfnm{B.~D.}\binits{B.~D.}},
  \bauthor{\bsnm{Koenderman},~\bfnm{L.}\binits{L.}}, \bauthor{\bparticle{ten}
  \bsnm{Hacken},~\bfnm{N.~H.~T.}\binits{N.~H.~T.}},
  \bauthor{\bsnm{Postma},~\bfnm{D.~S.}\binits{D.~S.}},
  \bauthor{\bparticle{van~den} \bsnm{Berge},~\bfnm{M.}\binits{M.}},
  \bauthor{\bparticle{de} \bsnm{Jong},~\bfnm{P.~A.}\binits{P.~A.}},
  \bauthor{\bsnm{Lammers},~\bfnm{J.~W.~J.}\binits{J.~W.~J.}} \AND
  \bauthor{\bsnm{Mohamed~Hoesein},~\bfnm{F.~A.~A.}\binits{F.~A.~A.}}
(\byear{2017}).
\btitle{Parametric response mapping on chest computed tomography associates
  with clinical and functional parameters in chronic obstructive pulmonary
  disease}.
\bjournal{Respiratory Medicine}
\bvolume{123}
\bpages{48--55}.
\bdoi{10.1016/j.rmed.2016.11.021}
\end{barticle}
\endbibitem

\bibitem[\protect\citeauthoryear{Pompe
  et~al.}{2023}]{pompe_imaging-derived_2023}
\begin{barticle}[author]
\bauthor{\bsnm{Pompe},~\bfnm{E.}\binits{E.}}, \bauthor{\bsnm{Kwee},~\bfnm{A.~K.
  A.~L.}\binits{A.~K. A.~L.}}, \bauthor{\bsnm{Tejwani},~\bfnm{V.}\binits{V.}},
  \bauthor{\bsnm{Siddharthan},~\bfnm{T.}\binits{T.}} \AND
  \bauthor{\bsnm{Mohamed~Hoesein},~\bfnm{F.~A.~A.}\binits{F.~A.~A.}}
(\byear{2023}).
\btitle{Imaging-derived biomarkers in Asthma: Current status and future
  perspectives}.
\bjournal{Respiratory Medicine}
\bvolume{208}
\bpages{107130, 6 pp.}
\bdoi{10.1016/j.rmed.2023.107130}
\end{barticle}
\endbibitem

\bibitem[\protect\citeauthoryear{Rohr et~al.}{2001}]{rohr}
\begin{barticle}[author]
\bauthor{\bsnm{Rohr},~\bfnm{K.}\binits{K.}},
  \bauthor{\bsnm{Stiehl},~\bfnm{H.~S.}\binits{H.~S.}},
  \bauthor{\bsnm{Sprengel},~\bfnm{R.}\binits{R.}},
  \bauthor{\bsnm{Buzug},~\bfnm{T.~M.}\binits{T.~M.}},
  \bauthor{\bsnm{Weese},~\bfnm{J.}\binits{J.}} \AND
  \bauthor{\bsnm{Kuhn},~\bfnm{M.~H.}\binits{M.~H.}}
(\byear{2001}).
\btitle{Landmark-based elastic registration using approximating thin-plate
  splines}.
\bjournal{IEEE Transactions on Medical Imaging}
\bvolume{20}
\bpages{526--534}.
\bdoi{10.1109/42.929618}
\end{barticle}
\endbibitem

\bibitem[\protect\citeauthoryear{Sengupta, Gupta and
  Biswas}{2022}]{sengupta_survey_2022}
\begin{barticle}[author]
\bauthor{\bsnm{Sengupta},~\bfnm{D.}\binits{D.}},
  \bauthor{\bsnm{Gupta},~\bfnm{P.}\binits{P.}} \AND
  \bauthor{\bsnm{Biswas},~\bfnm{A.}\binits{A.}}
(\byear{2022}).
\btitle{A survey on mutual information based medical image registration
  algorithms}.
\bjournal{Neurocomputing}
\bvolume{486}
\bpages{174--188}.
\bdoi{10.1016/j.neucom.2021.11.023}
\end{barticle}
\endbibitem

\bibitem[\protect\citeauthoryear{Siem, de~Klerk and
  Den~Hertog}{2008}]{siem2008discrete}
\begin{barticle}[author]
\bauthor{\bsnm{Siem},~\bfnm{A.~Y.~D.}\binits{A.~Y.~D.}},
  \bauthor{\bparticle{de} \bsnm{Klerk},~\bfnm{E.}\binits{E.}} \AND
  \bauthor{\bsnm{Den~Hertog},~\bfnm{D.}\binits{D.}}
(\byear{2008}).
\btitle{Discrete least-norm approximation by nonnegative (trigonometric)
  polynomials and rational functions}.
\bjournal{Structural and Multidisciplinary Optimization}
\bvolume{35}
\bpages{327--339}.
\bdoi{10.1007/s00158-007-0135-1}
\bmrnumber{2396498}
\end{barticle}
\endbibitem

\bibitem[\protect\citeauthoryear{Sumikawa et~al.}{2009}]{sumikawa2009computed}
\begin{barticle}[author]
\bauthor{\bsnm{Sumikawa},~\bfnm{H.}\binits{H.}},
  \bauthor{\bsnm{Johkoh},~\bfnm{T.}\binits{T.}},
  \bauthor{\bsnm{Yamamoto},~\bfnm{S.}\binits{S.}},
  \bauthor{\bsnm{Yanagawa},~\bfnm{M.}\binits{M.}},
  \bauthor{\bsnm{Inoue},~\bfnm{A.}\binits{A.}},
  \bauthor{\bsnm{Honda},~\bfnm{O.}\binits{O.}},
  \bauthor{\bsnm{Yoshida},~\bfnm{S.}\binits{S.}},
  \bauthor{\bsnm{Tomiyama},~\bfnm{N.}\binits{N.}} \AND
  \bauthor{\bsnm{Nakamura},~\bfnm{H.}\binits{H.}}
(\byear{2009}).
\btitle{Computed tomography values calculation and volume histogram analysis
  for various computed tomographic patterns of diffuse lung diseases}.
\bjournal{Journal of Computer Assisted Tomography}
\bvolume{33}
\bpages{731--738}.
\bdoi{10.1097/RCT.0b013e31818da65c}
\end{barticle}
\endbibitem

\bibitem[\protect\citeauthoryear{Trivedi et~al.}{2017}]{trivedi}
\begin{barticle}[author]
\bauthor{\bsnm{Trivedi},~\bfnm{A.}\binits{A.}},
  \bauthor{\bsnm{Hall},~\bfnm{C.}\binits{C.}},
  \bauthor{\bsnm{Hoffman},~\bfnm{E.~A.}\binits{E.~A.}},
  \bauthor{\bsnm{Woods},~\bfnm{J.~C.}\binits{J.~C.}},
  \bauthor{\bsnm{Gierada},~\bfnm{D.~S.}\binits{D.~S.}} \AND
  \bauthor{\bsnm{Castro},~\bfnm{M.}\binits{M.}}
(\byear{2017}).
\btitle{Using imaging as a biomarker for asthma}.
\bjournal{Journal of Allergy and Clinical Immunology}
\bvolume{139}
\bpages{1--10}.
\bdoi{10.1016/j.jaci.2016.11.009}
\end{barticle}
\endbibitem

\bibitem[\protect\citeauthoryear{Verde and Irpino}{2010}]{Verde2010}
\begin{binproceedings}[author]
\bauthor{\bsnm{Verde},~\bfnm{R.}\binits{R.}} \AND
  \bauthor{\bsnm{Irpino},~\bfnm{A.}\binits{A.}}
(\byear{2010}).
\btitle{Ordinary Least Squares for Histogram Data Based on {W}asserstein
  Distance}.
In \bbooktitle{Proceedings of COMPSTAT'2010}
(\beditor{\bfnm{Yves}\binits{Y.}~\bsnm{Lechevallier}} \AND
  \beditor{\bfnm{Gilbert}\binits{G.}~\bsnm{Saporta}}, eds.)
\bpages{581--588}.
\bpublisher{Physica-Verlag HD}, \baddress{Heidelberg}.
\bdoi{10.1007/978-3-7908-2604-3_60}
\end{binproceedings}
\endbibitem

\bibitem[\protect\citeauthoryear{Yan, Li and Niu}{2023}]{Yan2023}
\begin{barticle}[author]
\bauthor{\bsnm{Yan},~\bfnm{Q.}\binits{Q.}},
  \bauthor{\bsnm{Li},~\bfnm{H.}\binits{H.}} \AND
  \bauthor{\bsnm{Niu},~\bfnm{C.}\binits{C.}}
(\byear{2023}).
\btitle{Optimal subsampling for functional quantile regression}.
\bjournal{Statistical Papers}
\bvolume{64}
\bpages{1943--1968}.
\bdoi{10.1007/s00362-022-01367-z}
\bmrnumber{4672800}
\end{barticle}
\endbibitem

\bibitem[\protect\citeauthoryear{Yang}{2020}]{Yang2020a}
\begin{barticle}[author]
\bauthor{\bsnm{Yang},~\bfnm{H.}\binits{H.}}
(\byear{2020}).
\btitle{Random distributional response model based on spline method}.
\bjournal{Journal of Statistical Planning and Inference}
\bvolume{207}
\bpages{27--44}.
\bdoi{10.1016/j.jspi.2019.10.005}
\bmrnumber{4066120}
\end{barticle}
\endbibitem

\bibitem[\protect\citeauthoryear{Yang et~al.}{2020}]{Yang2020b}
\begin{barticle}[author]
\bauthor{\bsnm{Yang},~\bfnm{H.}\binits{H.}},
  \bauthor{\bsnm{Baladandayuthapani},~\bfnm{V.}\binits{V.}},
  \bauthor{\bsnm{Rao},~\bfnm{A.~U.~K.}\binits{A.~U.~K.}} \AND
  \bauthor{\bsnm{Morris},~\bfnm{J.~S.}\binits{J.~S.}}
(\byear{2020}).
\btitle{Quantile function on scalar regression analysis for distributional
  data}.
\bjournal{Journal of the American Statistical Association}
\bvolume{115}
\bpages{90--106}.
\bdoi{10.1080/01621459.2019.1609969}
\bmrnumber{4078447}
\end{barticle}
\endbibitem

\bibitem[\protect\citeauthoryear{Zarei et~al.}{2024}]{zarei2024quantitative}
\begin{barticle}[author]
\bauthor{\bsnm{Zarei},~\bfnm{F.}\binits{F.}},
  \bauthor{\bsnm{Jannatdoust},~\bfnm{P.}\binits{P.}},
  \bauthor{\bsnm{Malekpour},~\bfnm{S.}\binits{S.}},
  \bauthor{\bsnm{Razaghi},~\bfnm{M.}\binits{M.}},
  \bauthor{\bsnm{Chatterjee},~\bfnm{S.}\binits{S.}},
  \bauthor{\bsnm{Varadhan~Chatterjee},~\bfnm{V.}\binits{V.}},
  \bauthor{\bsnm{Abbasi},~\bfnm{A.}\binits{A.}} \AND
  \bauthor{\bsnm{Haghighi},~\bfnm{R.~R.}\binits{R.~R.}}
(\byear{2024}).
\btitle{Quantitative analysis of lung lesions using unenhanced chest computed
  tomography images}.
\bjournal{The Clinical Respiratory Journal}
\bvolume{18}
\bpages{e13759, 8 pp.}
\bdoi{10.1111/crj.13759}
\end{barticle}
\endbibitem

\bibitem[\protect\citeauthoryear{Zhang et~al.}{2022}]{Zhang2022}
\begin{barticle}[author]
\bauthor{\bsnm{Zhang},~\bfnm{Z.}\binits{Z.}},
  \bauthor{\bsnm{Wang},~\bfnm{X.}\binits{X.}},
  \bauthor{\bsnm{Kong},~\bfnm{L.}\binits{L.}} \AND
  \bauthor{\bsnm{Zhu},~\bfnm{H.}\binits{H.}}
(\byear{2022}).
\btitle{High-dimensional spatial quantile function-on-scalar regression}.
\bjournal{Journal of the American Statistical Association}
\bvolume{117}
\bpages{1563--1578}.
\bdoi{10.1080/01621459.2020.1870984}
\bmrnumber{4480732}
\end{barticle}
\endbibitem

\bibitem[\protect\citeauthoryear{Zhu et~al.}{2023}]{Zhu2023}
\begin{barticle}[author]
\bauthor{\bsnm{Zhu},~\bfnm{H.}\binits{H.}},
  \bauthor{\bsnm{Zhang},~\bfnm{Y.}\binits{Y.}},
  \bauthor{\bsnm{Li},~\bfnm{Y.}\binits{Y.}} \AND
  \bauthor{\bsnm{Lian},~\bfnm{H.}\binits{H.}}
(\byear{2023}).
\btitle{Semiparametric function-on-function quantile regression model with
  dynamic single-index interactions}.
\bjournal{Computational Statistics \& Data Analysis}
\bvolume{182}
\bpages{Paper No. 107727, 22 pp.}
\bdoi{10.1016/j.csda.2023.107727}
\bmrnumber{4550799}
\end{barticle}
\endbibitem

\end{thebibliography}

\section{Explicit expression for the minimizer \texorpdfstring{$q_{x,i,d}^{\star}$}{q\_x,i,d\_star}}\label{sec:explicit.minimizer}

\begin{proposition}\label{prop:q.opt}
Let $x_{i,1} \leq \dots \leq x_{i,N_i^x}$ be a sample of pre-treatment HU values for the $i$th patient, sorted in nondecreasing order. Denote by $\Gamma(\cdot)$ the Euler gamma function, $\gamma(\cdot, \cdot)$ the lower incomplete gamma function, and $\Gamma(\cdot, \cdot)$ the upper incomplete gamma function. Recall the parametric space of quantile functions, $\mathcal{Q}_d^2$, from Section~\ref{sec:math}, and recall that the empirical quantile function is defined as $\widehat{q}_{x,i}(p) = \smash{\sum_{\ell=1}^{N_i^x} x_{i,\ell} \, \ind_{\{(\ell-1)/N_i^x < p \leq \ell/N_i^x\}}}$. Then, the solution $q_{x,i,d}^{\star}$ to the minimization problem,
\[
\min_{q\in \mathcal{Q}_d^2} \int_0^1 \left(q(p)-\widehat{q}_{x,i}(p)\right)^2 \rd p,
\]
is $q_{x,i,d}^{\star} = \sum_{j=0}^d a_{i,d,j}^{\star} (\Phi^{-1})^j(\cdot)$, where
\[
\begin{aligned}
\bb{a}_{i,d}^{\star} \equiv (a_{i,d,0}^{\star},\ldots,a_{i,d,d}^{\star})
&= \argmin_{\bb{a}_{i,d}\in \Rnearrow} \left\{\sum_{j=0}^d a_{i,d,j} \left(\sum_{\ell=j}^{j+d} a_{i,d,\ell-j} \Phi_{\ell}(1) - 2 \psi_{i,j}\right)\right\}, \\
%%%
\psi_{i,j}
&= \sum_{\ell=1}^{N_i^x} x_{i,\ell} \left[\Phi_j(\ell/N_i^x)-\Phi_j((\ell-1)/N_i^x)\right], \quad j\in \N_0,
\end{aligned}
\]
and where the truncated Gaussian moment $\Phi_j(b)$ satisfies, for all $j\in \N_0$ and all $b\in (0,1]$,
\[
\begin{aligned}
&\Phi_j(b)
\leqdef \int_{-\infty}^{\Phi^{-1}(b)} x^j \phi(x) \rd x \\
&~~= \frac{2^{j/2-1}}{\sqrt{\pi}} \times
\begin{cases}
\Gamma(j/2+1/2) + \mathrm{sgn}(b - 1/2) \gamma(j/2 + 1/2, |\Phi^{-1}(b)|^2/2), &\mbox{if $j$ is even}, \\
- \Gamma(j/2+1/2, |\Phi^{-1}(b)|^2/2), &\mbox{if $j$ is odd}.
\end{cases}
\end{aligned}
\]
In the limit as $b \to 1$, Legendre's duplication formula ensures that $\Phi_j(b)$ reduces to the well-known expressions for the non-truncated Gaussian moments:
\[
\Phi_j(1) =
\begin{cases}
\frac{j!}{2^{j/2}(j/2)!}, &\mbox{if $j$ is even}, \\
0, &\mbox{if $j$ is odd}.
\end{cases}
\]
In the special case $d=1$, the minimizer is $q_{x,i,1}^{\star} = a_{i,1,0}^{\star} + a_{i,1,1}^{\star} \Phi^{-1}$ with
\begin{equation}\label{eq:special.case}
a_{i,1,0}^{\star} = \psi_{i,0}, \qquad a_{i,1,1}^{\star} = \psi_{i,1}.
\end{equation}
\end{proposition}

\begin{remark}
When $N_i^x$ is large, the optimization procedure can be very slow. However, if the observations are drawn from a finite set of possible values and include repetitions, the process can be significantly accelerated. Specifically, suppose there are $M_i$ distinct sorted values, denoted $\tilde{x}_{i,1} \leq \dots \leq \tilde{x}_{i,M_i}$, where each value $\tilde{x}_{i,\ell}$ occurs $c_{i,\ell}$ times. In this case, the total number of observations satisfies $N_i^x = c_{i,1} + \dots + c_{i,M_i}$. This structure enables optimization by leveraging the reduced dimensionality of the unique values and their frequencies. Consequently, one can redefine $\psi_{i,j}$ as
\[
\psi_{i,j} = \sum_{\ell=1}^K \tilde{x}_{i,\ell} \left[\Phi_j(n_{i,\ell}/N_i^x)-\Phi_j(n_{i,\ell-1}/N_i^x)\right], \quad j\in \N_0,
\]
where $n_{i,\ell} = c_{i,1} + \dots + c_{i,\ell}$ for all $\ell\in \{1,\ldots,M_i\}$.
\end{remark}

\begin{remark}
In Proposition~\ref{prop:q.opt}, the empirical quantile function $\widehat{q}_{x,i}(p)$ can be substituted with the refined expression,
\[
\begin{aligned}
\widetilde{q}_{x,i}(p)
&= \big[(N_i^x p-1/2) (x_{i,2} - x_{i,1}) + x_{i,1}\big] \ind_{\{0\leq p\leq \frac{1}{2N_i^x}\}} \\[-1mm]
&\quad+ \sum_{j=1}^{N_i^x-1} \big[(N_i^x p - j + 1/2) (x_{i,j+1} - x_{i,j}) + x_{i,j}\big] \ind_{\{\frac{2j-1}{2N_i^x}< p\leqslant\frac{2j+1}{2N_i^x}\}} \\
&\quad+ \big[(N_i^x p - N_i^x + 1/2) (x_{i,N_i^x} - x_{i,N_i^x-1}) + x_{i,N_i^x}\big] \ind_{\{1-\frac{1}{2N_i^x} < p\leq 1\}},
\end{aligned}
\]
to derive an alternative explicit formula for $q_{x,i,d}^{\star}$ by mimicking the proof presented below. This refinement was introduced by \citet[p.~163]{Blanke2021_supp}, who demonstrated that it results in a mean integrated squared error which is strictly lower than that of the classical empirical quantile function.
\end{remark}

\begin{proof}[Proof of Proposition~\ref{prop:q.opt}]
Recall the definition of $\Rnearrow$ from \eqref{eq:R.inc}. One has
\[
\begin{aligned}
q_{x,i,d}^{\star}
&= \argmin_{q\in \mathcal{Q}_d^2} \int_0^1 \left(q(p)-\widehat{q}_{x,i}(p)\right)^2 \rd p \\
&= \argmin_{q\in \mathcal{Q}_d^2} \left\{\int_0^1 q^2(p) \rd p - 2 \int_0^1 q(p) \widehat{q}_{x,i}(p) \rd p + \int_0^1 \widehat{q}_{x,i}^2(p) \rd p \right\} \\
&= \argmin_{q\in \mathcal{Q}_d^2} \left\{\int_0^1q^2(p) \rd p - 2\int_0^1q(p)\widehat{q}_{x,i}(p)\rd p\right\}.
\end{aligned}
\]
In particular, if $q_{x,i,d}^{\star} = \sum_{j=0}^d a_{i,d,j}^{\star} (\Phi^{-1})^j(\cdot)$, then
\[
\begin{aligned}
(a_{i,d,0}^{\star},\ldots,a_{i,d,d}^{\star})
&= \argmin_{\bb{a}_{i,d}\in \Rnearrow} \bigg\{\sum_{j=0}^d \sum_{k=0}^d a_{i,d,j} a_{i,d,k} \int_0^1\left(\Phi^{-1}\right)^{j+k}(p) \rd p \bigg. \\[-4mm]
&\hspace{40mm} \bigg.- 2 \sum_{j=0}^d a_{i,d,j} \int_0^1 (\Phi^{-1})^j(p) \widehat{q}_{x,i}(p) \rd p \bigg\} \\[-2mm]
&= \argmin_{\bb{a}_{i,d}\in \Rnearrow} \bigg\{\sum_{j=0}^d a_{i,d,j} \bigg(\sum_{\ell=j}^{j+d} a_{i,d,\ell-j} \Phi_{\ell}(1) - 2 \psi_{i,j}\bigg)\bigg\}.
\end{aligned}
\]

In the special case $d=1$, the above reduces to
\[
\begin{aligned}
q_{x,i,d}^{\star}
&= \argmin_{(a_{i,1,0},a_{i,1,1})\in \R \times (0,\infty)} \Big\{a_{i,1,0} (a_{i,1,0} \Phi_0(1) + a_{i,1,1} \Phi_1(1) - 2 \psi_{i,0}) \Big. \\[-4mm]
&\hspace{40mm} \Big. + a_{i,1,1} (a_{i,1,0} \Phi_1(1) + a_{i,1,1} \Phi_2(1) - 2 \psi_{i,1})\Big\}.
\end{aligned}
\]
After differentiating with respect to $a_{i,1,0}$ and $a_{i,1,1}$ to find the critical points, one has to solve the linear system:
\vspace{-2mm}
\[
\begin{bmatrix}
2 \Phi_0(1) ~~& 2 \Phi_1(1) \\[2mm]
2 \Phi_1(1) ~~& 2 \Phi_2(1)
\end{bmatrix}
\begin{bmatrix}
a_{i,1,0} \\[2mm]
a_{i,1,1}
\end{bmatrix}
=
\begin{bmatrix}
2 \psi_{i,0} \\[2mm]
2 \psi_{i,1}
\end{bmatrix}.
\]
Since $\Phi_0(1) = 1$, $\Phi_1(1) = 0$, and $\Phi_2(1) = 1$, the above yields
\[
\begin{bmatrix}
a_{i,1,0}^{\star} \\[2mm]
a_{i,1,1}^{\star}
\end{bmatrix}
=
\frac{1}{\Phi_0(1) \Phi_2(1) - \Phi_1(1)^2}
\begin{bmatrix}
\Phi_2(1) ~~& -\Phi_1(1) \\[2mm]
-\Phi_1(1) ~~& \Phi_0(1)
\end{bmatrix}
\begin{bmatrix}
\psi_{i,0} \\[2mm]
\psi_{i,1}
\end{bmatrix}
=
\begin{bmatrix}
\psi_{i,0} \\[2mm]
\psi_{i,1}
\end{bmatrix},
\]
which proves the claim \eqref{eq:special.case}.

To complete the proof, it remains to find a general expression for the truncated Gaussian moments $\Phi_j(b)$ that appear in $\psi_{i,j}$. For even $j$, the symmetry of the integrand and the change of variable $t = x^2/2$ yield the following expression, for all $b\in (0,1]$,
\[
\begin{aligned}
\Phi_j(b)
&= \int_{-\infty}^{\Phi^{-1}(b)} x^j \phi(x) \rd x
= \int_{-\infty}^0 x^j \phi(x) \rd x + \mathrm{sgn}(b - 1/2) \int_0^{|\Phi^{-1}(b)|} x^j \phi(x) \rd x \\[-1mm]
&= \int_0^{\infty} (2 t)^{j/2} \frac{\exp(-t)}{\sqrt{2\pi}} \frac{\rd t}{\sqrt{2 t}} + \mathrm{sgn}(b - 1/2) \int_0^{|\Phi^{-1}(b)|^2/2} (2 t)^{j/2} \frac{\exp(-t)}{\sqrt{2\pi}} \frac{\rd t}{\sqrt{2 t}} \\
&= \frac{2^{j/2-1}}{\sqrt{\pi}} \left[\Gamma(j/2+1/2) + \mathrm{sgn}(b - 1/2) \gamma(j/2 + 1/2, |\Phi^{-1}(b)|^2/2)\right].
\end{aligned}
\]
For odd $j$, the fact that $\Phi_j(1) = 0$ and the change of variable $t = x^2/2$ yield the following expression, for all $b\in (0,1]$,
\[
\begin{aligned}
\Phi_j(b)
&= \int_{-\infty}^{\Phi^{-1}(b)} x^j \phi(x) \rd x
= - \int_{|\Phi^{-1}(b)|}^{\infty} x^j \phi(x) \rd x \\[-1mm]
&= - \int_{|\Phi^{-1}(b)|^2/2}^{\infty} (2 t)^{j/2} \frac{\exp(-t)}{\sqrt{2\pi}} \frac{\rd t}{\sqrt{2 t}}
= - \frac{2^{j/2-1}}{\sqrt{\pi}} \Gamma(j/2 + 1/2, |\Phi^{-1}(b)|^2/2).
\end{aligned}
\]

\vspace{2mm}
\noindent
The expression for $\Phi_j(b)$ can also be verified numerically using the following \texttt{R} code:
\begin{quote}
\begin{verbatim}
## integrate() versus the explicit formula:
b <- 0.3 # some arbitrary cutoff
j <- 4 # case when j is even
integrate(function(p) (qnorm(p)) ^ j, 0, b)$value
2 ^ (j / 2 - 1) / sqrt(pi) *
    (gamma(j / 2 + 0.5) + sign(b - 0.5) *
        pgamma(qnorm(b) ^ 2 / 2, j / 2 + 0.5, lower.tail = TRUE) *
            gamma(j / 2 + 0.5))
j <- 3 # case when j is odd
integrate(function(p) (qnorm(p)) ^ j, 0, b)$value
-2 ^ (j / 2 - 1) / sqrt(pi) *
    pgamma(qnorm(b) ^ 2 / 2, j / 2 + 0.5, lower.tail = FALSE) *
        gamma(j / 2 + 0.5)
\end{verbatim}
\end{quote}
This concludes the proof.
\end{proof}

\section{Proofs of the results stated in Section~\ref{sec:regression}}\label{sec:proofs}

Two preliminary results are needed before proving Propositions~\ref{prop:MLE}~and~\ref{prop:pivot.distributions}.

\begin{lemma}\label{lem:Fred.1}
Let $\theta_1,\ldots,\theta_n\in (0,\infty)$ be given, and let $X_1,\ldots,X_n\stackrel{\mathrm{iid}}{\sim} \mathcal{E}\mathrm{xp}(1)$ for some integer $n\geq 2$. Then, the joint survival function of $\min_{1\leq i\leq n} \theta_i X_i$ and $\frac{1}{n} \sum_{i=1}^n X_i$ is, for all $(x,y)\in (0,\infty)^2$,
\[
\begin{aligned}
S(x,y)
&\equiv \PP\left(\min_{1\leq i\leq n} \theta_i X_i \geq x, \frac{1}{n} \sum_{i=1}^n X_i \geq y\right) \\
&=
\begin{cases}
\int_{n y}^{\infty} \left(1 - \frac{x}{t} \sum_{i=1}^n \frac{1}{\theta_i}\right)^{n-1} \frac{t^{n-1} \exp(-t)}{\Gamma(n)} \rd t, &\mbox{if } \sum_{i=1}^n \frac{x}{\theta_i} < n y, \\
\exp\left(-\sum_{i=1}^n \frac{x}{\theta_i}\right), &\mbox{if } \sum_{i=1}^n \frac{x}{\theta_i} \geq n y.
\end{cases}
\end{aligned}
\]
The corresponding joint density function is, for all $(x,y)\in \R^2$,
\[
\begin{aligned}
&f_{\min_{1\leq i\leq n} \theta_i X_i, \frac{1}{n} \sum_{i=1}^n X_i}(x,y) \\
&\quad=
\begin{cases}
n \left(\sum_{i=1}^n \frac{1}{\theta_i}\right) \frac{\left(n y - \sum_{i=1}^n \frac{x}{\theta_i}\right)^{n-2} \exp(-n y)}{\Gamma(n-1)}, &\mbox{if } \sum_{i=1}^n \frac{x}{\theta_i} < n y, \\
0, &\mbox{otherwise}.
\end{cases}
\end{aligned}
\]
\end{lemma}

\begin{proof}[Proof of Lemma~\ref{lem:Fred.1}]
Assume throughout the proof that $n\geq 2$, $\theta_1,\ldots,\theta_n\in (0,\infty)$, and $x,y\in (0,\infty)$, are given. If $\sum_{i=1}^n \frac{x}{\theta_i} \geq n y$, then the event $\min_{1\leq i\leq n} \theta_i X_i \geq x$ implies
\[
\frac{1}{n} \sum_{i=1}^n X_i = \frac{1}{n} \sum_{i=1}^n \frac{1}{\theta_i} (\theta_i X_i) \geq \frac{1}{n} \sum_{i=1}^n \frac{x}{\theta_i} \geq y.
\]
Therefore, the independence of $X_1,\ldots,X_n$ yields
\[
S(x,y) = \PP\left(\min_{1\leq i\leq n} \theta_i X_i \geq x\right) = \prod_{i=1}^n \PP\bigg(X_i \geq \frac{x}{\theta_i}\bigg) = \exp\bigg(-\sum_{i=1}^n \frac{x}{\theta_i}\bigg).
\]

Now, assume that $\sum_{i=1}^n \frac{x}{\theta_i} < n y$ holds for the remainder of the proof. One has
\[
\begin{aligned}
S(x,y) = \int_{x/\theta_1}^{\infty} \dots \int_{x/\theta_n}^{\infty} \mathds{1}_{\{\sum_{i=1}^n x_i \geq n y\}} \exp\left(-\sum_{i=1}^n x_i\right) \rd x_1 \dots \rd x_n
\end{aligned}
\]
Consider the change of variables
\[
x_1 = t s_1, \quad \dots, \quad x_{n-1} = t s_{n-1}, \quad x_n = t (1 - s_1 - \dots - s_{n-1}),
\]
so that $\sum_{i=1}^n x_i = t$. Using a cofactor expansion, the Jacobian determinant is equal to
\[
\left|\det
\begin{bmatrix}
\frac{\rd (x_1,\ldots,x_{n-1},x_n)}{\rd (s_1,\ldots,s_{n-1},t)}
\end{bmatrix}
\right|
=
\left|\det
\begin{bmatrix}
t ~~ & 0 ~~ & 0 ~~ & \dots ~~ & 0 ~~ & s_1 \\
0 ~~ & t ~~ & 0 ~~ & \dots ~~ & 0 ~~ & s_2 \\
0 ~~ & 0 ~~ & t ~~ & \dots ~~ & 0 ~~ & s_3 \\
\vdots ~~ & \vdots ~~ & \vdots ~~ & \ddots ~~ & \vdots ~~ & \vdots \\
0 ~~ & 0 ~~ & 0 ~~ & \dots ~~ & t ~~ & s_{n-1} \\
-t ~~ & -t ~~ & -t ~~ & \dots ~~ & -t ~~ & 1 - \sum_{i=1}^{n-1} s_i
\end{bmatrix}
\right|
=
t^{n-1}.
\]
Hence, the above integral reduces to
\[
S(x,y) = \int_{n y}^{\infty} \int_{\mathcal{S}_{n-1}^{\star}} \rd (s_1,\ldots,s_{n-1}) ~ t^{n-1} \exp(-t) \rd t,
\]
where $\mathcal{S}_{n-1}^{\star}$ denotes the $(n-1)$-dimensional simplex with a certain buffer near the boundary:
\[
\mathcal{S}_{n-1}^{\star} = \left\{(s_1,\ldots,s_{n-1})\in (0,1)^{n-1} : s_1 \geq \frac{x}{\theta_1 t}, \ldots, s_{n-1} \geq \frac{x}{\theta_{n-1} t}, 1 - \sum_{i=1}^{n-1} s_i \geq \frac{x}{\theta_n t}\right\}.
\]
A simple geometric argument shows that
\[
\int_{\mathcal{S}_{n-1}^{\star}} \rd (s_1,\ldots,s_{n-1}) = \frac{\big(1 - \sum_{i=1}^n \frac{x}{\theta_i t}\big)^{n-1}}{(n-1)!},
\]
which proves the claimed expression for the joint survival function.

To obtain the expression for the joint density function, observe that
\[
\begin{aligned}
&f_{\min_{1\leq i\leq n} \theta_i X_i, \frac{1}{n} \sum_{i=1}^n X_i}(x,y)
= \frac{\partial^2}{\partial x \partial y} S(x,y) \\
&\quad=
\begin{cases}
\frac{\partial}{\partial x} n \left(1 - \frac{1}{n y} \sum_{i=1}^n \frac{x}{\theta_i}\right)^{n-1} \frac{(n y)^{n-1} \exp(-n y)}{\Gamma(n)}, &\mbox{if } \sum_{i=1}^n \frac{x}{\theta_i} < n y, \\
0, &\mbox{otherwise},
\end{cases} \\
&\quad=
\begin{cases}
n (n-1) \left(1 - \frac{1}{n y} \sum_{i=1}^n \frac{x}{\theta_i}\right)^{n-2} \left(\frac{1}{n y} \sum_{i=1}^n \frac{1}{\theta_i}\right) \frac{(n y)^{n-1} \exp(-n y)}{\Gamma(n)}, &\mbox{if } \sum_{i=1}^n \frac{x}{\theta_i} < n y, \\
0, &\mbox{otherwise},
\end{cases} \\
&\quad=
\begin{cases}
n \left(\sum_{i=1}^n \frac{1}{\theta_i}\right) \frac{\left(n y - \sum_{i=1}^n \frac{x}{\theta_i}\right)^{n-2} \exp(-n y)}{\Gamma(n-1)}, &\mbox{if } \sum_{i=1}^n \frac{x}{\theta_i} < n y, \\
0, &\mbox{otherwise}.
\end{cases}
\end{aligned}
\]
This concludes the proof.
\end{proof}

\begin{corollary}\label{cor:density.afer.linear.transformation}
Let $\theta_1,\ldots,\theta_n\in (0,\infty)$ be given, and let $X_1,\ldots,X_n\stackrel{\mathrm{iid}}{\sim} \mathcal{E}\mathrm{xp}(1)$ for some integer $n\geq 2$. For an invertible matrix $A \equiv (a_{ij})_{1\leq i,j\leq 2}\in \R^{2 \times 2}$, define the linear change of variable
\[
\begin{bmatrix}
U \\[2mm]
V
\end{bmatrix}
=
\begin{bmatrix}
a_{11} & ~~ a_{12} \\[2mm]
a_{21} & ~~ a_{22}
\end{bmatrix}
\begin{bmatrix}
\min_{1\leq i\leq n} \theta_i X_i \\[2mm]
\frac{1}{n} \sum_{i=1}^n X_i
\end{bmatrix}.
\]
Let $(a^{ij})_{1\leq i,j \leq 2}$ denote the entries of $A^{-1}$:
\[
A^{-1}
=
\begin{bmatrix}
a^{11} & ~~ a^{12} \\[2mm]
a^{21} & ~~ a^{22}
\end{bmatrix}
=
\frac{1}{a_{11} a_{22} - a_{21} a_{12}}
\begin{bmatrix}
a_{22} & ~~ -a_{12} \\[2mm]
-a_{21} & ~~ a_{11}
\end{bmatrix}.
\]
Then, for all $(u,v)\in \R^2$ such that $a^{11} u + a^{12} v\in (0,\infty)$ and $a^{21} u + a^{22} v\in (0,\infty)$, the joint density function of $(U,V)$ is
\[
\begin{aligned}
f_{(U,V)}(u,v)
&= f_{(\min_{1\leq i\leq n} \theta_i X_i,\frac{1}{n} \sum_{i=1}^n X_i)}(a^{11} u + a^{12} v, a^{21} u + a^{22} v) \frac{1}{|\det(A)|} \\
&=
\left\{\hspace{-1mm}
\begin{array}{l}
\frac{n \left(\sum_{i=1}^n 1/\theta_i\right)}{|a_{11} a_{22} - a_{21} a_{12}|} \frac{\left(n (a^{21} u + a^{22} v) - (a^{11} u + a^{12} v) \sum_{i=1}^n \frac{1}{\theta_i}\right)^{n-2} \exp\left(-n (a^{21} u + a^{22} v)\right)}{\Gamma(n-1)}, \\[2mm]
\hspace{45mm} \mbox{if } (a^{11} u + a^{12} v) \sum_{i=1}^n \frac{1}{\theta_i} < n (a^{21} u + a^{22} v), \\[4mm]
0, \hspace{42mm} \mbox{otherwise}.
\end{array}
\right.
\end{aligned}
\]
\end{corollary}

\begin{proof}[Proof of Proposition~\ref{prop:MLE}]
Recall that $(\mu_{Q_{Y,i}},\sigma_{Q_{Y,i}})=\Delta^{-1}(Q_{Y,i})$ as per Section~\ref{sec:math}. The likelihood function of the model~\eqref{eq:model} is equal to
\[
\begin{aligned}
L(\beta_0,\beta_1,\beta_2,\sigma^2,\beta)
&= \prod_{i=1}^n f_{\beta_0,\beta_1,\beta_2,\sigma^2,\beta}(Q_{Y,i}) \\
&\equiv \prod_{i=1}^n f_{\mathcal{N}(\beta_0 + \beta_1 \mu_{q_{x,i}},\sigma^2)}(\mu_{Q_{Y,i}}) \times f_{\mathcal{E}\mathrm{xp}(\beta, \beta_2\sigma_{q_{x,i}})}(\sigma_{Q_{Y,i}}) \\
&= \frac{1}{(2\pi\sigma^2)^{n/2}} \exp\left\{-\sum_{i=1}^n \frac{(\mu_{Q_{Y,i}} -\beta_0 - \beta_1 \mu_{q_{x,i}})^2}{2\sigma^2}\right\} \\
&\quad\times \frac{1}{\beta^n } \exp\left\{\sum_{i=1}^n \frac{-(\sigma_{Q_{Y,i}} -\beta_2\sigma_{q_{x,i}})}{\beta} \right\} \ind_{\left\{\min\left(\frac{\sigma_{Q_{Y,i}}}{\sigma_{q_{x,i}}}\right) \geq \beta_2\right\}}.
\end{aligned}
\]
Differentiating the likelihood with respect to the variables $\beta_0$, $\beta_1$, $\sigma^2$ and $\beta$, respectively, for $\beta_2 \leq \min_{1\leq i\leq n}(\sigma_{Q_{Y,i}}/\sigma_{q_{x,i}})$, and setting each result to zero leads to the system of equations:
\[
\begin{aligned}
\sum_{i=1}^n (\mu_{Q_{Y,i}} -\beta_0 -\beta_1\mu_{q_{x,i}}) = 0, & \qquad
\sum_{i=1}^n \mu_{q_{x,i}}\bigl(\mu_{Q_{Y,i}} -\beta_0 -\beta_1\mu_{q_{x,i}}\bigr) = 0, & \\
-\frac{n}{2\sigma^2} + \sum_{i=1}^n \bigr(\mu_{Q_{Y,i}} - \beta_0-\beta_1\mu_{q_{x,i}}\bigl)^2 \frac{1}{2\sigma^4} = 0, & \qquad
-\frac{n}{\beta} + \frac{1}{\beta^2} \sum_{i=1}^n (\sigma_{Q_{Y,i}} -\beta_2\sigma_{q_{x,i}}) = 0.
\end{aligned}
\]
In turn, straightforward algebraic manipulations yield the ML estimators:
\[
\widehat{\beta}_{0,\mathrm{ML}} = \bar{\mu}_{Q_Y} - \widehat{\beta}_{1,\mathrm{ML}}\bar{\mu}_{q_x},\qquad\widehat{\beta}_{1,\mathrm{ML}} = \frac{\sum_{i=1}^n \mu_{Q_{Y,i}}\mu_{q_{x,i}} - n\bar{\mu}_{Q_Y}\bar{\mu}_{q_x} }{\sum_{i=1}^n \mu_{q_{x,i}}^2 - n\bar{\mu}_{q_x}^2},
\]
\[
\widehat{\beta}_{\mathrm{ML}} = \frac{1}{n} \sum_{i=1}^n (\sigma_{Q_{Y,i}}-\widehat{\beta}_{2,\mathrm{ML}}\sigma_{q_{x,i}})=\bar{\sigma}_{Q_Y}-\widehat{\beta}_{2,\mathrm{ML}}\bar{\sigma}_{q_x},
\]
and
\[
\widehat{\sigma^2}_{\mathrm{ML}} = \frac{1}{n}\sum_{i=1}^n \bigr(\mu_{Q_{Y,i}}-\widehat{\beta}_{0,\mathrm{ML}}-\widehat{\beta}_{1,\mathrm{ML}}\mu_{q_{x,i}}\bigl)^2.
\]
Moreover, given that the likelihood function is positive and increasing in $\beta_2$, it is clear that it maximizes when $\beta_2$ is set to
\[
\widehat{\beta}_{2,\mathrm{ML}} = \min_{1\leq i\leq n} \left(\frac{\sigma_{Q_{Y,i}}}{\sigma_{q_{x,i}}}\right).
\]

To finalize the proof, the distribution of each ML estimator must be derived. Since $\mu_{Q_{Y,i}} \stackrel{\mathrm{iid}}{\sim} \mathcal{N}(\beta_0+\beta_1\mu_{q_{x,i}},\sigma^2)$ by Proposition~\ref{prop:law.mu.sigma}, one can write, for every $i\in \{1,\ldots,n\}$,
\[
\mu_{Q_{Y,i}} = \beta_0 + \beta_1 \mu_{q_{x,i}} + \e_i,
\]
where $\e_1, \ldots, \e_n \overset{\mathrm{iid}}{\sim} \mathcal{N}(0, \sigma^2)$. This corresponds to a classical simple linear regression model, with Gaussian errors, for the $\mu_{Q_{Y,i}}$'s regressed on the $\mu_{q_{x,i}}$'s. Consequently,
\[
\widehat{\beta}_{0,\mathrm{ML}} \sim \mathcal{N}\bigg(\beta_0, \frac{\sigma^2}{n}\bigg(1+\frac{\bar{\mu}_{q_x}^2}{w}\bigg)\bigg), \quad
\widehat{\beta}_{1,\mathrm{ML}} \sim \mathcal{N}\bigg(\beta_1,\frac{\sigma^2}{nw}\bigg),
\]
where $w = n^{-1} \sum_{i=1}^n \mu_{q_{x,i}}^2 - \bar{\mu}_{q_{x}}^2$, and
\[
\frac{n}{\sigma^2} \widehat{\sigma^2}_{\mathrm{ML}} \sim \chi^2(n-2).
\]

Similarly, since $\sigma_{Q_{Y,i}} \stackrel{\mathrm{iid}}{\sim} \mathcal{E}\mathrm{xp}(\beta, \beta_2\sigma_{q_{x,i}})$, one can write, for every $i\in \{1,\ldots,n\}$,
\[
\sigma_{Q_{Y,i}} = \beta_2 \sigma_{q_{x,i}} + \e_i,
\]
where $\e_1,\ldots,\e_n \overset{\mathrm{iid}}{\sim} \mathcal{E}\mathrm{xp}(\beta)$. This represents a classical simple linear regression model, with exponential errors and no intercept, for the $\sigma_{Q_{Y,i}}$'s regressed on the $\sigma_{q_{x,i}}$'s. Given that this setting is not standard, the distributions of $\smash{\widehat{\beta}_{\mathrm{ML}}}$ and $\smash{\widehat{\beta}_{2,\mathrm{ML}}}$ are derived below.

From Definition~\ref{def:unbiased.estimators} and Proposition~\ref{prop:pivot.distributions}, it is known that \[
\widehat{\beta}_{\mathrm{ML}} = \frac{n-1}{n} \widehat{\beta}, \qquad \widehat{\beta}\sim \mathcal{G}\mathrm{amma}\left(n-1, \frac{\beta}{n-1}\right),
\]
in the shape-scale parametrization. Therefore, $\widehat{\beta}_{\mathrm{ML}}\sim \mathcal{G}\mathrm{amma}(n-1,\beta/n)$.

Finally, in the scale-shift parametrization, the survival function of a shifted exponential distribution, $\mathcal{E}\mathrm{xp}(\lambda,\delta)$, is $\smash{\exp(-\lambda^{-1}(x-\delta)) \ind_{\{x\geq\delta\}} + \ind_{\{x < \delta\}}}$. Since $\smash{\sigma_{Q_{Y,i}} \stackrel{\mathrm{iid}}{\sim} \mathcal{E}\mathrm{xp}(\beta, \beta_2\sigma_{q_{x,i}})}$, one has, for all $t\geq \beta_2$,
\[
\begin{aligned}
\PP\left(\min_{1\leq i\leq n}\left(\frac{\sigma_{Q_{Y,i}}}{\sigma_{q_{x,i}}}\right)>t\right)
&= \prod_{i=1}^n \PP(\sigma_{Q_{Y,i}} > t \sigma_{q_{x,i}}) \\[-0.5mm]
&= \prod_{i=1}^n \exp(-\beta^{-1}(t \sigma_{q_{x,i}} -\beta_2\sigma_{q_{x,i}})) \\[-1mm]
&= \exp\left(-\beta^{-1}\sum_{i=1}^n \sigma_{q_{x,i}}(t-\beta_2)\right),
\end{aligned}
\]
which shows that $\widehat{\beta}_{2,\mathrm{ML}} \sim \mathcal{E}\mathrm{xp}(\beta / (n\bar{\sigma}_{q_{x,i}}), \beta_2)$. This concludes the proof.
\end{proof}

\begin{proof}[Proof of Proposition~\ref{prop:unbiased.estimators}]
A quick verification using the calculations in Remark~\ref{eq:bias} shows that $\smash{\widehat{\sigma^2}}$, $\smash{\widehat{\beta}_2}$, and $\smash{\widehat{\beta}}$, as given in Definition~\ref{def:unbiased.estimators}, are indeed unbiased:
\[
\begin{aligned}
&\EE(\widehat{\sigma^2}) = \frac{n}{n-2} \EE(\widehat{\sigma^2}_{\mathrm{ML}}) = \frac{n}{n-2} \left(\frac{n-2}{n} \sigma^2\right) = \sigma^2, \\
&\EE(\widehat{\beta}_2) = \frac{n}{n-1} \left(\beta_2 + \frac{\beta}{n \bar{\sigma}_{q_x}}\right) - \frac{(\beta + \beta_2 \bar{\sigma}_{q_x})}{(n-1) \bar{\sigma}_{q_x}} = \beta_2, \\
&\EE(\widehat{\beta}) = \frac{n}{n-1} \EE(\widehat{\beta}) = \frac{n}{n-1} \left(\frac{n-1}{n} \beta\right) = \beta.
\end{aligned}
\]
This concludes the proof.
\end{proof}

\begin{proof}[Proof of Proposition~\ref{prop:pivot.distributions}]
From Proposition~\ref{prop:MLE} and Definition~\ref{def:unbiased.estimators}, one has
\[
\begin{bmatrix}
(\widehat{\beta}_2 - \beta_2) / \beta \\[1mm]
\widehat{\beta} / \beta
\end{bmatrix}
\stackrel{\mathrm{law}}{=}
\frac{n}{n-1}
\begin{bmatrix}
1 & ~~ \frac{-1}{n \bar{\sigma}_{q_x}} \\[2mm]
- \bar{\sigma}_{q_x} & ~~ 1
\end{bmatrix}
\begin{bmatrix}
\min_{1\leq i\leq n} \sigma_{q_{x,i}}^{-1} X_i \\[2mm]
\frac{1}{n} \sum_{i=1}^n X_i
\end{bmatrix},
\]
where $X_1,\ldots,X_n \stackrel{\mathrm{iid}}{\sim} \mathcal{E}\mathrm{xp}(1)$.

An application of Corollary~\ref{cor:density.afer.linear.transformation} with
\[
\theta_i = \sigma_{q_{x,i}}^{-1},
\qquad
A = \frac{n}{n-1}
\begin{bmatrix}
1 & ~~ \frac{-1}{n \bar{\sigma}_{q_x}} \\[1mm]
- \bar{\sigma}_{q_x} & ~~ 1
\end{bmatrix},
\qquad
A^{-1}
=
\begin{bmatrix}
1 & ~~ \frac{1}{n \bar{\sigma}_{q_x}} \\[2mm]
\bar{\sigma}_{q_x} & ~~ 1
\end{bmatrix},
\]
leads to expressions for the densities of $((\widehat{\beta}_2 - \beta_2)/\beta,\widehat{\beta}/\beta)$, $\widehat{\beta}/\beta$, and $(\widehat{\beta}_2 - \beta_2)/\beta$, viz.,
\[
\begin{aligned}
f_{((\widehat{\beta}_2 - \beta_2)/\beta,\widehat{\beta}/\beta)}(u,v)
&=
\begin{cases}
n \bar{\sigma}_{q_x} \exp(- n \bar{\sigma}_{q_x} u) \times \frac{(n-1)^{n-1} v^{n-2} \exp(-n v)}{\Gamma(n-1)}, &\mbox{if } u > \frac{-v}{n \bar{\sigma}_{q_x}}, ~ v > 0, \\
0, &\mbox{otherwise}.
\end{cases} \\
f_{\widehat{\beta}/\beta}(v)
&=
\begin{cases}
\frac{(n-1)^{n-1} v^{n-2} \exp(-(n-1) v)}{\Gamma(n-1)}, &\mbox{if } v > 0, \\
0, &\mbox{otherwise},
\end{cases} \\
f_{(\widehat{\beta}_2 - \beta_2)/\beta}(u)
&=
\begin{cases}
n \bar{\sigma}_{q_x} \exp(-n \bar{\sigma}_{q_x} u) \left(\frac{n-1}{n}\right)^{n-1} \overline{\Gamma}(n-1, -n^2 \bar{\sigma}_{q_x} u), &\mbox{if } u < 0 \\
n \bar{\sigma}_{q_x} \exp(-n \bar{\sigma}_{q_x} u) \left(\frac{n-1}{n}\right)^{n-1}, &\mbox{if } u\geq 0,
\end{cases}
\end{aligned}
\]
where $\overline{\Gamma}$ denotes the regularized upper incomplete gamma function, i.e.,
\[
\overline{\Gamma}(a,b) = \frac{1}{\Gamma(a)} \int_b^{\infty} t^{a-1} \exp(-t) \rd t, \quad (a,b)\in (0,\infty) \times \R.
\]
In particular, this shows that $\widehat{\beta}/\beta\sim \mathcal{G}\mathrm{amma}(n-1,1/(n-1))$.

Furthermore, for all $z\in (-1/(n \bar{\sigma}_{q_x}),\infty)$, one has
\[
\begin{aligned}
f_{(\widehat{\beta}_2 - \beta_2) / \widehat{\beta}}(z)
&= \int_0^{\infty} f_{((\widehat{\beta}_2 - \beta_2)/\beta,\widehat{\beta}/\beta)}(z v,v) v \rd v \\
&= \int_0^{\infty} n \bar{\sigma}_{q_x} \exp(-n \bar{\sigma}_{q_x} z v) \frac{(n-1)^{n-1} v^{n-2} \exp(-n v)}{\Gamma(n-1)} \, v \rd v \\
&= \frac{n (n-1)^n \bar{\sigma}_{q_x}}{(n (1 + \bar{\sigma}_{q_x} z))^n } \times \frac{1}{\Gamma(n)} \int_0^{\infty} (n (1 + \bar{\sigma}_{q_x} z))^n v^{n-1} \exp(-n (1 + \bar{\sigma}_{q_x} z) v) \rd v \\
&= \frac{n (1-1/n)^n \bar{\sigma}_{q_x}}{(1 + \bar{\sigma}_{q_x} z)^n } \times 1.
\end{aligned}
\]
It follows that, for all $z\in (0,\infty)$,
\[
\begin{aligned}
f_{(\widehat{\beta}_2 - \beta_2) / \widehat{\beta} + 1/(n \bar{\sigma}_{q_x})}(z)
&= f_{(\widehat{\beta}_2 - \beta_2) / \widehat{\beta}}(z - 1/(n \bar{\sigma}_{q_x})) \\
&= \frac{(n-1)}{(1-1/n)/\bar{\sigma}_{q_x}} \left(\frac{(1-1/n)/\bar{\sigma}_{q_x}}{(1-1/n)/\bar{\sigma}_{q_x} + z}\right)^{(n-1)+1},
\end{aligned}
\]
meaning that $(\widehat{\beta}_2 - \beta_2) / \widehat{\beta} + 1/(n \bar{\sigma}_{q_x})\sim \mathcal{P}\mathrm{areto}\text{-}\mathcal{T}\mathrm{ype}\text{-}\mathrm{II}(n-1, (1-1/n)/\bar{\sigma}_{q_x})$. This concludes the proof.
\end{proof}

\begin{proof}[Proof of Remark~\ref{rem:hat.beta.2.positive}]
Using the expression of the density $\smash{f_{(\widehat{\beta}_2 - \beta_2)/\beta}}$ in the proof of Proposition~\ref{prop:pivot.distributions}, calculations in \texttt{Mathematica} show that
\[
\EE(\widehat{\beta}_2) = \beta_2, \qquad \Var(\widehat{\beta}_2) = \frac{\beta^2}{n (n-1) \bar{\sigma}_{q_x}^2}.
\]
It follows from Chebyshev's inequality that, for all $\e > 0$,
\[
\PP(\widehat{\beta}_2 < \beta_2 - \e) \leq \PP(|\widehat{\beta}_2 - \beta_2| > \e) \leq \frac{\beta^2}{n (n-1) \e^2 \bar{\sigma}_{q_x}^2}.
\]
This concludes the proof.
\end{proof}

\begin{proof}[Proof of Proposition~\ref{prop:CR.new.mean.response}]
Using the expressions for $\widehat{\beta}_0$ and $\widehat{\beta}_1$ in Proposition~\ref{prop:MLE}, one has
\[
\Var(\widehat{\beta}_0)
= \frac{\sigma^2}{n} \bigg(1 + \frac{\bar{\mu}_{q_x}^2}{w}\bigg), \qquad
\Cov(\widehat{\beta}_0,\widehat{\beta}_1)
= \frac{- \sigma^2 \bar{\mu}_{q_x}}{n w}, \qquad
\Var(\widehat{\beta}_1)
= \frac{\sigma^2}{n w}.
\]
Since $\mu_{\widehat{Q}_{Y,n+1}} = \widehat{\beta}_0 + \widehat{\beta}_1 \mu_{q_{x,n+1}}$, one deduces
\[
\begin{aligned}
\Var(\mu_{\widehat{Q}_{Y,n+1}})
&= \Var(\widehat{\beta}_0) + 2 \mu_{q_{x,n+1}} \Cov(\widehat{\beta}_0,\widehat{\beta}_1) + \mu_{q_{x,n+1}}^2 \Var(\widehat{\beta}_1) \\
&= \frac{\sigma^2}{n} \left(1 + \frac{(\mu_{q_{x,n+1}} - \bar{\mu}_{q_x})^2}{w}\right).
\end{aligned}
\]
Given that $(\widehat{\beta}_0,\widehat{\beta}_1)$ is a nonsingular Gaussian random vector, it follows that
\begin{equation}\label{eq:hat.mu.Q.Y.n+1}
\mu_{\widehat{Q}_{Y,n+1}} = \widehat{\beta}_0 + \mu_{q_{x,n+1}} \widehat{\beta}_1 \sim \mathcal{N}\left(\beta_0 + \mu_{q_{x,n+1}} \beta_1, \frac{\sigma^2}{n} \left(1 + \frac{(\mu_{q_{x,n+1}} - \bar{\mu}_{q_x})^2}{w}\right)\right).
\end{equation}

The next goal is to derive the distribution of $\smash{\sigma_{\widehat{Q}_{Y,n+1}}} = \widehat{\beta}_2 \sigma_{q_{x,n+1}} + \widehat{\beta}$. Using the expression found in the proof of Proposition~\ref{prop:pivot.distributions} for the joint density function of $\smash{((\widehat{\beta}_2 - \beta_2)/\beta,\widehat{\beta}/\beta)}$, one obtains that, under the restriction $\sigma_{q_{x,n+1}} < n \bar{\sigma}_{q_x}$, and for all $t\in (0,\infty)$,
\[
\begin{aligned}
&f_{\sigma_{q_{x,n+1}} (\widehat{\beta}_2 - \beta_2)/\beta + \widehat{\beta}/\beta}(t) \\
&\quad= \int_0^{t (1-\sigma_{q_{x,n+1}}/(n \bar{\sigma}_{q_x}))^{-1}} f_{((\widehat{\beta}_2 - \beta_2)/\beta, \widehat{\beta}/\beta)}((t-v)/\sigma_{q_{x,n+1}},v) \frac{1}{\sigma_{q_{x,n+1}}} \rd v \\
&\quad= \int_0^{t (1-\sigma_{q_{x,n+1}}/(n \bar{\sigma}_{q_x}))^{-1}} \frac{n \bar{\sigma}_{q_x}}{\sigma_{q_{x,n+1}}} \exp\left(\frac{- n \bar{\sigma}_{q_x}}{\sigma_{q_{x,n+1}}} (t - v)\right) \times \frac{(n-1)^{n-1} v^{n-2} \exp(-n v)}{\Gamma(n-1)} \rd v \\
&\quad= \frac{\frac{n \bar{\sigma}_{q_x}}{\sigma_{q_{x,n+1}}} \exp\left(\frac{- n \bar{\sigma}_{q_x}}{\sigma_{q_{x,n+1}}} t\right)}{\left(\frac{n}{n-1} \left(1 - \frac{\bar{\sigma}_{q_x}}{\sigma_{q_{x,n+1}}}\right)\right)^{n-1}} \\[-1mm]
&\qquad\times \int_0^{t (1-\sigma_{q_{x,n+1}}/(n \bar{\sigma}_{q_x}))^{-1}} \frac{\left(n \left(1 - \frac{\bar{\sigma}_{q_x}}{\sigma_{q_{x,n+1}}}\right)\right)^{n-1} v^{n-2} \exp\left(- n \left(1 - \frac{\bar{\sigma}_{q_x}}{\sigma_{q_{x,n+1}}}\right) v\right)}{\Gamma(n-1)}\rd v \\
&\quad= \frac{\frac{n \bar{\sigma}_{q_x}}{\sigma_{q_{x,n+1}}} \exp\left(\frac{- n \bar{\sigma}_{q_x}}{\sigma_{q_{x,n+1}}} t\right)}{\left(\frac{n}{n-1} \left(1 - \frac{\bar{\sigma}_{q_x}}{\sigma_{q_{x,n+1}}}\right)\right)^{n-1}} \times \overline{\gamma}\left(n-1, n \left(1 - \frac{\bar{\sigma}_{q_x}}{\sigma_{q_{x,n+1}}}\right) \left(1 - \frac{\sigma_{q_{x,n+1}}}{n \bar{\sigma}_{q_x}}\right)^{-1} t\right),
\end{aligned}
\]
where recall $\overline{\gamma}$ denote the regularized lower incomplete gamma function. Hence, under the restriction $\sigma_{q_{x,n+1}} < n \bar{\sigma}_{q_x}$, one has, for all $t > \beta_2 \sigma_{q_{x,n+1}}$,
\begin{equation}\label{eq:hat.sigma.Q.Y.n+1}
\begin{aligned}
f_{\sigma_{\widehat{Q}_{Y,n+1}}}(t)
&= \frac{\frac{n \bar{\sigma}_{q_x}}{\beta \sigma_{q_{x,n+1}}} \exp\left(\frac{- n \bar{\sigma}_{q_x}}{\beta \sigma_{q_{x,n+1}}} (t - \beta_2 \sigma_{q_{x,n+1}})\right)}{\left(\frac{n}{n-1} \left(1 - \frac{\bar{\sigma}_{q_x}}{\sigma_{q_{x,n+1}}}\right)\right)^{n-1}} \\
&\quad\times \overline{\gamma}\left(n-1, n \frac{\left(1 - \frac{\bar{\sigma}_{q_x}}{\sigma_{q_{x,n+1}}}\right)}{\left(1 - \frac{\sigma_{q_{x,n+1}}}{n \bar{\sigma}_{q_x}}\right)} \frac{(t - \beta_2 \sigma_{q_{x,n+1}})}{\beta}\right).
\end{aligned}
\end{equation}
Since $\smash{\mu_{\widehat{Q}_{Y,n+1}}}$ and $\smash{\sigma_{\widehat{Q}_{Y,n+1}}}$ are independent, it follows from \eqref{eq:hat.mu.Q.Y.n+1} and \eqref{eq:hat.sigma.Q.Y.n+1} that, under the restriction $\sigma_{q_{x,n+1}} < n \bar{\sigma}_{q_x}$, one has, for all $(s,t)\in \R \times (\beta_2 \sigma_{q_{x,n+1}},\infty)$,
\[
\begin{aligned}
&f_{\mu_{\widehat{Q}_{Y,n+1}},\sigma_{\widehat{Q}_{Y,n+1}}}(s,t)
= f_{\mu_{\widehat{Q}_{Y,n+1}}}(s) f_{\sigma_{\widehat{Q}_{Y,n+1}}}(t) \\
&= \frac{1}{\sqrt{2\pi \frac{\sigma^2}{n} \left(1 + \frac{(\mu_{q_{x,n+1}} - \bar{\mu}_{q_x})^2}{w}\right)}} \exp\left(- \frac{(s - (\beta_0 + \mu_{q_{x,n+1}} \beta_1))^2}{2 \frac{\sigma^2}{n} \left(1 + \frac{(\mu_{q_{x,n+1}} - \bar{\mu}_{q_x})^2}{w}\right)}\right) \\
&\quad\times \frac{\frac{n \bar{\sigma}_{q_x}}{\beta \sigma_{q_{x,n+1}}} \exp\left(\frac{- n \bar{\sigma}_{q_x}}{\beta \sigma_{q_{x,n+1}}} (t - \beta_2 \sigma_{q_{x,n+1}})\right)}{\left(\frac{n}{n-1} \left(1 - \frac{\bar{\sigma}_{q_x}}{\sigma_{q_{x,n+1}}}\right)\right)^{n-1}} \, \overline{\gamma}\left(n-1, n \frac{\left(1 - \frac{\bar{\sigma}_{q_x}}{\sigma_{q_{x,n+1}}}\right)}{\left(1 - \frac{\sigma_{q_{x,n+1}}}{n \bar{\sigma}_{q_x}}\right)} \frac{(t - \beta_2 \sigma_{q_{x,n+1}})}{\beta}\right).
\end{aligned}
\]
Now replacing $\beta_0$, $\beta_1$, $\sigma^2$, $\beta_2$, and $\beta$, in the last equation by the observed values of the unbiased estimators $\smash{\widehat{\beta}_0}$, $\smash{\widehat{\beta}_1}$, $\smash{\widehat{\sigma^2}}$, $\smash{\widehat{\beta}_2}$, and $\smash{\widehat{\beta}}$, one defines, for all $\smash{(s,t)\in \R \times (\widehat{\beta}_2 \sigma_{q_{x,n+1}},\infty)}$,
\[
\begin{aligned}
&\widehat{f}_{\widehat{Q}_{Y,n+1}}(s,t)
\equiv \widehat{f}_{\mu_{\widehat{Q}_{Y,n+1}},\sigma_{\widehat{Q}_{Y,n+1}}}(s,t) \\[-1mm]
&\leqdef \frac{1}{\sqrt{2\pi \frac{\widehat{\sigma^2}}{n} \left(1 + \frac{(\mu_{q_{x,n+1}} - \bar{\mu}_{q_x})^2}{w}\right)}} \exp\left(- \frac{(s - (\widehat{\beta}_0 + \mu_{q_{x,n+1}} \widehat{\beta}_1))^2}{2 \frac{\widehat{\sigma^2}}{n} \left(1 + \frac{(\mu_{q_{x,n+1}} - \bar{\mu}_{q_x})^2}{w}\right)}\right) \\
&~\times \frac{\frac{n \bar{\sigma}_{q_x}}{\widehat{\beta} \sigma_{q_{x,n+1}}} \exp\left(\frac{- n \bar{\sigma}_{q_x}}{\widehat{\beta} \sigma_{q_{x,n+1}}} (t - \widehat{\beta}_2 \sigma_{q_{x,n+1}})\right)}{\left(\frac{n}{n-1} \left(1 - \frac{\bar{\sigma}_{q_x}}{\sigma_{q_{x,n+1}}}\right)\right)^{n-1}} \, \overline{\gamma}\left(n-1, n \frac{\left(1 - \frac{\bar{\sigma}_{q_x}}{\sigma_{q_{x,n+1}}}\right)}{\left(1 - \frac{\sigma_{q_{x,n+1}}}{n \bar{\sigma}_{q_x}}\right)} \frac{(t - \widehat{\beta}_2 \sigma_{q_{x,n+1}})}{\widehat{\beta}}\right).
\end{aligned}
\]
This concludes the proof.
\end{proof}

The following technical lemma is needed before proving Proposition~\ref{prop:residuals}. It is a consequence of the memoryless property of exponential random variables.

\begin{lemma}\label{lem:memoryless}
Let $\theta_1,\ldots,\theta_n\in (0,\infty)$ and let $X_1,\ldots,X_n$ be a sequence of independent scale-parametrized exponential random variables such that $X_j\sim \mathcal{E}\mathrm{xp}(\theta_j)$ for each $j$. Then, for any given $\ell\in \{1,\ldots,n\}$ and $t\in (0,\infty)$, one has
\[
\PP\left(\bigcap_{\substack{j=1 \\ j\neq \ell}}^n \{X_j - X_{\ell} > t\} ~\Bigg|\Bigg.~ \bigcap_{\substack{j=1 \\ j\neq \ell}}^n \{X_j > X_{\ell}\}\right) = \prod_{j\neq \ell}^n \exp(-\theta_j^{-1} t) = \PP\left(\bigcap_{j\neq \ell}^n \{X_j > t\}\right).
\]
In simple terms, the above equation indicates that the joint distribution of the random vector $(X_1,\ldots,X_{\ell-1},X_{\ell+1},\ldots,X_n)$, given that each component exceeds the exponential time $X_{\ell}$, is the same as the unconditional distribution of $(X_1,\ldots,X_{\ell-1},X_{\ell+1},\ldots,X_n)$ on its own.
\end{lemma}

\begin{proof}[Proof of Lemma~\ref{lem:memoryless}]
Let $\ell\in \{1,\ldots,n\}$ be given. Note that, for all $s,t\in (0,\infty)$,
\[
\begin{aligned}
\PP\left(\bigcap_{j\neq \ell}^n \{X_j - s > t\} ~\Bigg|\Bigg.~ \bigcap_{j\neq \ell}^n \{X_j > s\}\right)
&= \PP\left(\bigcap_{j\neq \ell}^n \{X_j > s + t\} ~\Bigg|\Bigg.~ \bigcap_{j\neq \ell}^n \{X_j > s\}\right) \\
&= \frac{\prod_{j\neq \ell}^n \exp(-\theta_j^{-1} (s + t))}{\prod_{j\neq \ell}^n \exp(-\theta_j^{-1} s)}
= \prod_{j\neq \ell}^n \exp(-\theta_j^{-1} t).
\end{aligned}
\]
This last expression does not depend on $s$, so integrating it over $s\in (0,\infty)$ with respect to the probability measure $f_{X_{\ell}}(s) \rd s$ yields the conclusion.
\end{proof}

\begin{proof}[Proof of Proposition~\ref{prop:residuals}]
Let $i\in \{1,\ldots,n\}$ be given. Using the expressions for $\smash{\widehat{\beta}_0}$ and $\smash{\widehat{\beta}_1}$ in Proposition~\ref{prop:MLE}, it is straightforward to show that
\[
\widehat{\beta}_0 = \frac{1}{n} \sum_{j=1}^n \left(1 - \frac{\bar{\mu}_{q_x} (\mu_{q_{x,j}} - \bar{\mu}_{q_x})}{w}\right) \mu_{Q_{Y,j}}, \qquad
\widehat{\beta}_1 = \frac{1}{n} \sum_{j=1}^n \frac{(\mu_{q_{x,j}} - \bar{\mu}_{q_x})}{w} \mu_{Q_{Y,j}}.
\]
Starting from \eqref{eq:mu.sigma.residuals}, one has
\[
\begin{aligned}
\mu_{\widehat{E}_i}
&= \mu_{Q_{Y,i}} - \mu_{\widehat{Q}_{Y,i}}
= \mu_{Q_{Y,i}} - (\widehat{\beta}_0 + \widehat{\beta}_1 \mu_{q_{x,i}}) \\
&= \mu_{Q_{Y,i}} - \frac{1}{n} \sum_{j=1}^n \left(1 - \frac{\bar{\mu}_{q_x} (\mu_{q_{x,j}} - \bar{\mu}_{q_x})}{w} + \frac{(\mu_{q_{x,j}} - \bar{\mu}_{q_x})}{w} \mu_{q_{x,i}}\right) \mu_{Q_{Y,j}} \\[-1mm]
&= \left(1 - \frac{1}{n} - \frac{(\mu_{q_{x,i}} - \bar{\mu}_{q_x})^2}{n w}\right) \mu_{Q_{Y,i}} - \frac{1}{n} \sum_{\substack{j=1 \\ j\neq i}}^n \left(1 + \frac{(\mu_{q_{x,j}} - \bar{\mu}_{q_x})}{n w} (\mu_{q_{x,i}} - \bar{\mu}_{q_x})\right) \mu_{Q_{Y,j}}.
\end{aligned}
\]
Furthermore, it is known from \eqref{eq:model.reformulation} that $\mu_{Q_{Y,i}} = \beta_0 + \beta_1 \mu_{q_{x,i}} + \mu_{E_i}$, and the model~\eqref{eq:model} also assumes that $\mu_{E_i}\sim \mathcal{N}(0,\sigma^2)$, so $\smash{\mu_{\widehat{E}_i}}$ is a centered Gaussian random variable such that
\[
\begin{aligned}
\Var(\mu_{\widehat{E}_i})
&= \sigma^2 \left[\frac{1}{n^2} \left(n - 1 - \frac{(\mu_{q_{x,i}} - \bar{\mu}_{q_x})^2}{w}\right)^2 \right. \\[-5mm]
&\hspace{30mm} \left.+ \frac{1}{n^2} \sum_{\substack{j=1 \\ j\neq i}}^n \left(1 + \frac{(\mu_{q_{x,j}} - \bar{\mu}_{q_x})}{w} (\mu_{q_{x,i}} - \bar{\mu}_{q_x})\right)^2\right].
\end{aligned}
\]
Expanding the squares and simplifying the variance leads to
\[
\mu_{\widehat{E}_i} \sim \mathcal{N}\left(0, \sigma^2 \left(1 - \frac{1}{n} - \frac{(\mu_{q_{x,i}} - \bar{\mu}_{q_x})^2}{n w}\right)\right).
\]

Next, starting from \eqref{eq:mu.sigma.residuals}, one has
\[
\begin{aligned}
\sigma_{\widehat{E}_i}
&= (\beta_2 \sigma_{q_{x,i}} + \sigma_{E_i}) - \widehat{\beta}_2 \sigma_{q_{x,i}} \\
&= \sigma_{E_i} - \left[\left(\frac{n}{n-1} \min_{1 \leq j \leq n} \left(\frac{\sigma_{Q_{Y,j}}}{\sigma_{q_{x,j}}}\right) - \frac{\sum_{j=1}^n \sigma_{Q_{Y,j}}}{n (n-1) \bar{\sigma}_{q_x}}\right) - \beta_2\right] \sigma_{q_{x,i}}\\
&= \sigma_{E_i} - \left[\frac{n}{n-1} \min_{1 \leq j \leq n} \left(\frac{\beta_2 \sigma_{q_{x,j}} + \sigma_{E_j}}{\sigma_{q_{x,j}}}\right) - \frac{\sum_{j=1}^n (\beta_2 \sigma_{q_{x,j}} + \sigma_{E_j})}{n (n-1) \bar{\sigma}_{q_x}} - \beta_2\right] \sigma_{q_{x,i}} \\
&= \frac{-n \sigma_{q_{x,i}}}{n-1} \min_{1 \leq j \leq n} \left(\frac{\sigma_{E_j}}{\sigma_{q_{x,j}}}\right) + \sigma_{q_{x,i}} \left(\frac{\sigma_{E_i}}{\sigma_{q_{x,i}}}\right) + \sum_{j=1}^n \frac{\sigma_{q_{x,i}} \sigma_{q_{x,j}}}{n (n-1) \bar{\sigma}_{q_x}} \left(\frac{\sigma_{E_j}}{\sigma_{q_{x,j}}}\right) \\[-2mm]
&= b_i \min_{1 \leq j \leq n} X_j + \sum_{j=1}^n c_{i,j} X_j,
\end{aligned}
\]
where the constants $b_i,c_{i,1},\ldots,c_{i,n}$ and the independent scale-parametrized exponential random variables $X_1,\ldots,X_n$ are defined, for all $j\in \{1,\ldots,n\}$, by
\[
b_i \leqdef \frac{-n \sigma_{q_{x,i}}}{n-1}, \quad c_{i,j} \leqdef \sigma_{q_{x,i}} \ind_{\{j=i\}} + \frac{\sigma_{q_{x,i}} \sigma_{q_{x,j}}}{n(n-1) \bar{\sigma}_{q_x}}, \quad X_j \leqdef \frac{\sigma_{E_j}}{\sigma_{q_{x,j}}} \sim \mathcal{E}\mathrm{xp}\bigg(\theta_j = \frac{\beta}{\sigma_{q_{x,j}}}\bigg).
\]

For every $\ell\in \{1,\ldots,n\}$, the probability that the $\ell$th exponential random variable $X_{\ell}$ is the smallest in the sequence $X_1,\ldots,X_n$ is equal to
\[
\begin{aligned}
\PP\left(\bigcap_{\substack{j=1 \\ j\neq \ell}}^n \{X_j > X_{\ell}\}\right)
&= \int_0^{\infty} \PP\left(\bigcap_{\substack{j=1 \\ j\neq \ell}}^n \{X_j > s\} ~\Bigg|\Bigg.~ X_{\ell} = s\right) f_{X_{\ell}}(s) \rd s \\
&= \int_0^{\infty} \left(\prod_{j\neq \ell}^n \exp(-\theta_j^{-1} s)\right) \theta_{\ell}^{-1} \exp(-\theta_{\ell}^{-1} s) \rd s \\
&= \int_0^{\infty} \theta_{\ell}^{-1} \exp\bigg(-\sum_{j=1}^n \theta_j^{-1} s\bigg) \rd s = \frac{\theta_{\ell}^{-1}}{\sum_{j=1}^n \theta_j^{-1}} = \frac{\sigma_{q_{x,\ell}}}{n \bar{\sigma}_{q_x}} \reqdef w_{\ell},
\end{aligned}
\]
where the second equality exploited the independence of the $X_j$'s. Hence, by conditioning on which exponential random variable is the smallest, then applying Lemma~\ref{lem:memoryless} and the fact that $b_i + \sum_{j=1}^n c_{i,j} = 0$, one can calculate the survival function of $\sigma_{\widehat{E}_i}$ for all $t\in (0,\infty)$:
\[
\begin{aligned}
\PP(\sigma_{\widehat{E}_i} > t)
&= \sum_{\ell=1}^n w_{\ell} ~ \PP\left(b_i X_{\ell} + \sum_{j=1}^n c_{i,j} X_j > t ~\Bigg|\Bigg.~ \bigcap_{\substack{j=1 \\ j\neq \ell}}^n \{X_j > X_{\ell}\}\right) \\
&= \sum_{\ell=1}^n w_{\ell} ~\PP\left(\sum_{\substack{j=1 \\ j\neq \ell}}^n c_{i,j} (X_j - X_{\ell}) > t - \left(b_i + \sum_{j=1}^n c_{i,j}\right) X_{\ell} ~\Bigg|\Bigg.~ \bigcap_{\substack{j=1 \\ j\neq \ell}}^n \{X_j > X_{\ell}\}\right) \\
&= \sum_{\ell=1}^n w_{\ell} ~\PP\left(\sum_{j\neq \ell}^n c_{i,j} X_j > t\right).
\end{aligned}
\]
By isolating the $i$th term and differentiating on both sides of the last equation, one finds that
\[
f_{\sigma_{\widehat{E}_i}}(t)
= w_i \, f_{\sum_{j\neq i}^n c_{i,j} X_j}(t) + \sum_{\substack{\ell=1 \\ \ell\neq i}}^n w_{\ell} \, f_{c_{i,i} X_i + \sum_{j\not\in \{\ell,i\}}^n c_{i,j} X_j}(t), \quad t\in (0,\infty).
\]
Using the scale parametrization, note that
\[
c_{i,j} X_j
\sim
\begin{cases}
\mathcal{E}\mathrm{xp}\left(\nu_{i,1} \leqdef c_{i,i} \theta_i = \beta + \frac{\beta \sigma_{q_{x,i}}}{n(n-1) \bar{\sigma}_{q_x}}\right), &\mbox{if } j = i, \\[2mm]
\mathcal{E}\mathrm{xp}\left(\nu_{i,2} \leqdef c_{i,j} \theta_j = \frac{\beta \sigma_{q_{x,i}}}{n(n-1) \bar{\sigma}_{q_x}}\right), &\mbox{if } j\neq i.
\end{cases}
\]
Therefore, one has
\[
\begin{aligned}
f_{\sigma_{\widehat{E}_i}}(t)
&= w_i f_{\mathcal{G}\mathrm{amma}(n-1,\nu_{i,2})}(t) \\[-1mm]
&\qquad+ (1 - w_i) \int_0^t f_{\mathcal{E}\mathrm{xp}(\nu_{i,1})}(t - v) f_{\mathcal{G}\mathrm{amma}(n-2,\nu_{i,2})}(v) \rd v,
\end{aligned}
\]
where the last integral is the convolution between the density of $c_{i,i} X_i \sim \mathcal{E}\mathrm{xp}(\nu_{i,1})$ and the density of $\sum_{j\not\in \{\ell,i\}}^n c_{i,j} X_j \sim \mathcal{G}\mathrm{amma}(n-2,\nu_{i,2})$, which, using the notation
\[
\theta \leqdef \left(\widehat{\nu}_{i,2}^{-1} - \widehat{\nu}_{i,1}^{-1}\right)^{-1},
\]
can be simplified as follows:
\[
\begin{aligned}
&\int_0^t f_{\mathcal{E}\mathrm{xp}(\widehat{\nu}_{i,1})}(t - v) f_{\mathcal{G}\mathrm{amma}(n-2,\widehat{\nu}_{i,2})}(v) \rd v \\
&\qquad= \int_0^t \widehat{\nu}_{i,1}^{-1} \exp\left(-\widehat{\nu}_{i,1}^{-1}(t - v)\right) \frac{v^{n-3} \exp(-\widehat{\nu}_{i,2}^{-1} v)}{\widehat{\nu}_{i,2}^{n-2} \Gamma(n-2)} \rd v \\
&\qquad= \widehat{\nu}_{i,1}^{-1} \exp(-\widehat{\nu}_{i,1}^{-1} t) \left(\frac{\theta}{\widehat{\nu}_{i,2}}\right)^{n-2} \int_0^t \frac{v^{n-3} \exp(-\theta^{-1} v)}{\theta^{n-2} \Gamma(n-2)} \rd v \\
&\qquad= \widehat{\nu}_{i,1}^{-1} \exp(-\widehat{\nu}_{i,1}^{-1} t) \left(\frac{\theta}{\widehat{\nu}_{i,2}}\right)^{n-2} F_{\mathcal{G}\mathrm{amma}(n-2,\theta)}(t).
\end{aligned}
\]
This concludes the proof.
\end{proof}

\end{document}